\documentclass[aoas,preprint]{imsart}
\setattribute{journal}{name}{}
\RequirePackage[round]{natbib}
\RequirePackage[colorlinks,citecolor=blue,urlcolor=blue]{hyperref}

\usepackage{bbm}
\usepackage{comment}
\usepackage{amsmath, amsthm, amssymb}
\usepackage{algorithm, algorithmicx}
\usepackage{graphicx}
\usepackage{subfigure}
\usepackage{lineno}
\usepackage[usenames,dvipsnames]{color}
\usepackage{multirow}
\usepackage{booktabs}

\let\hat\widehat
\newtheorem{thm}{Theorem}
\newtheorem{lem}[thm]{Lemma}

\theoremstyle{remark}

\newcommand\E{\mathbb{E}}

\newcommand\cL{{\cal L}}

\renewcommand\P{\mathsf{P}}
\newcommand\p{\mathsf{p}}
\newcommand\B{\mathsf{B}}

\catcode`@=11
\newskip\beforeproofvskip
\newskip\afterproofvskip
\beforeproofvskip=\medskipamount
\afterproofvskip=\bigskipamount

\def\prooftag{Proof}
\def\proofskip{\enspace}

\def\proof{\@ifnextchar[{\@@proof}{\@proof}}  %] for emacs matching
\def\@startproof{\par\vskip\beforeproofvskip\leavevmode}
\def\@proof{\@startproof{\scshape\prooftag.}\proofskip}
\def\@@proof[#1]{\@startproof {\scshape\prooftag #1.}\proofskip}

\catcode`@=12

\let\hat\widehat

\DeclareMathOperator*{\di}{\mathrm{d}\!}

\begin{document}

\begin{frontmatter}

\title{Measuring Human Activity Spaces from GPS Data with Density Ranking and Summary Curves}
\runtitle{Measuring Human Activity Spaces}

\begin{aug}
  \author{\fnms{Yen-Chi}
    \snm{Chen}\ead[label=e1]{yenchic@uw.edu}}\and
  \author{\fnms{Adrian}
    \snm{Dobra}\ead[label=e2]{adobra@uw.edu}}
  \affiliation{Department of Statistics\\University of Washington}
  \runauthor{Y.-C. Chen and A. Dobra}
  \address{Department of Statistics\\University of Washington\\
Box 354322\\ Seattle, WA 98195 \\
          \printead{e1}}
    \address{Department of Statistics,\\ Department of Biobehavioral Nursing and Health Informatics,\\ and Center for Statistics and the Social Sciences\\University of Washington\\
Box 354322\\ Seattle, WA 98195 \\
          \printead{e2}}      
        \today
\end{aug}

\begin{abstract}
 Activity spaces are fundamental to the assessment of individuals' dynamic exposure to social and environmental risk factors associated with multiple spatial contexts that are visited during activities of daily living.  In this paper we survey existing approaches for measuring the geometry, size and structure of activity spaces based on GPS data, and explain their limitations. We propose addressing these shortcomings through a nonparametric approach called density ranking, and also through three summary curves: the mass-volume curve, the Betti number curve, and the persistence curve. We introduce a novel mixture model for human activity spaces, and study its asymptotic properties. We prove that the kernel density estimator which, at the present time, is one of the most widespread methods for measuring activity spaces is not a stable estimator of their structure. We illustrate the practical value of our methods with a simulation study, and with a recently collected  GPS dataset that comprises the locations visited by ten individuals over a six months period.
\end{abstract}


\end{frontmatter}

\section{Introduction}

%need to rephrase this paragraph
Collecting and statistical modeling of data on human movement in time and space is an important research endeavor in many fields, such as spatial epidemiology, demography and population science, urban design and planning, transportation research and environmental psychology \citep{RN156,RN143,entwisle-2007,hurvitz-2014,chen-et-2016,dobra-tanser-aids}. Mapping individuals is difficult because a person's residence does not reflect their interaction with the physical and social environment \citep{kwan-2009}. Individuals spend considerable time away from their residences and traverse multiple administrative boundaries in their daily activities \citep{RN151,RN147}. For this reason, it is paramount to trace an individual through multiple spatial contexts to study environmental risk factors for disease \citep{cummins-et-2007}. Statistical analyses that connect individuals to places by focusing on residential neighborhoods or administrative boundaries (e.g., census tracts) cannot capture short term but repetitive exposures to neighborhood-based risk factors (e.g., risk of violence or density of alcohol outlets). Going beyond the residential neighborhood of a person by collecting fine-grained positional data about where people actually spend time is especially relevant in studies that relate individual health to locally variable environmental factors \citep{basta-et-2010}.

As human beings are inherently mobile, data about their spatiotemporal trajectories of travel are needed to construct relevant representations of their activity spaces. The notion of activity space has been introduced in the social sciences \citep{golledge-stimson-1997}, and has its roots in the space-time-travel geography in which an individual's movements in time and space are conceptualized as space-time prisms \citep{hagerstrand-1963,hagerstrand-1970}. Activity spaces measure individual spatial behavior, and capture individuals' experience of place in the course of their daily living through their observed location choices \citep{golledge-1999}. They have been used to study the influence of the built environment on individuals' healthcare accessibility \citep{sherman-et-2005}. Activity spaces play a role in the study of social exclusion of individuals with low use of physical space which are less likely to be engaged in society \citep{schonfelder-axhausen-2003}. Questions of interest relate to whether such individuals concentrate spatially, or are randomly scattered in the population. Are these individuals socially excluded from certain parts of the physical environment which could lead, for example, to lower chances of securing a job or higher costs of living? Activity spaces have also been used, among many applications, to assess segregation \citep{wong-shaw-2011}, to measure exposure to food environments \citep{kestens-et-2010,christian-2012}, and to understand the geographic mobility patterns of older adults \citep{hirsch-2014}.

Until about 15 years ago, research on activity spaces relied on locational data from travel diaries in which participants shared information about the trips they took in the past \citep{schonfelder-axhausen-2003,schonfelder-axhausen-2004}. However, places outside the home neighborhood that are not socially significant are harder to be remembered, and consequently they will be more likely to be missing from surveys. Smartphone-based location traces have recently become available for the study of human mobility and have proven particularly interesting, by providing the possibility of recording movements over time of individual people and aggregate movements of whole populations \citep{dobra-et-2015,williams-et-2015}. This exciting new type of data holds immense promise for studying human behavior with a precision and accuracy never before possible with surveys or other data collection techniques \citep{RN143}. Many high-resolution smartphone-based GPS location datasets have already been successfully collected, and subsequently employed to assess human spatial behavior and spatiotemporal contextual exposures \citep{RN144,RN146,RN149}, to characterize the relationship between geographic and contextual attributes of the environment (e.g., the built environment) and human energy balance (e.g., diet, weight, physical activity) \citep{RN145,RN151}, to study segregation, environmental exposure, and accessibility in social science research \citep{RN147}, or to understand the relationship between health-risk behavior in adolescents (e.g., substance abuse) and community disorder \citep{RN148,basta-et-2010,RN153}. The wide array of completed and ongoing GPS studies provide key evidence that many people feel comfortable having their movements tracked \citep{RN155}.

In this paper, we survey existing approaches for measuring the geometry, size and structure of activity spaces based on GPS data such as ellipses, shortest-path spanning trees, and kernel density estimation, and explain the disadvantages of their use. To correct their shortcomings, we put forward a set of tools for measuring human activity spaces that comprise a nonparametric approach called density ranking and three types of summary curves. These curves fall within the broader domain of topological data analysis which is a flexible framework for detecting the structure and creating lower-dimensional summaries of distributions of complex or high-dimensional datasets \citep{kaczynski-et-2004,edelsbrunner2008persistent,edelsbrunner-harer-2009,carlsson2009topology,lum-et-2013,ghrist-2014,chazal2017introduction,wasserman2016topological,wasserman-2018}. The summary curves we discuss are based on level sets of density ranking, which is closely related to level sets of a probability density function and the minimum volume set \citep{polonik1997minimum,garcia2003level,scott2006learning,cadre2013estimation}. 

The structure of the paper is as follows. In Section \ref{sec::GPS} we describe the GPS data we use to motivate and illustrate our developments. This is a never before analyzed dataset that comprises the spatiotemporal trajectories of daily living over a six months period of ten individuals from a rural area in sub-Saharan Africa. In Section \ref{sec::activity} we present background on human activity spaces, and describe existent methods for measuring them.  In Section \ref{sec::DR} we present density ranking, and in Section \ref{sec::tda} we discuss three types of summary curves: the mass-volume curve, the Betti number curve, and the persistence curve. In Section \ref{sec::model} we introduce a novel mixture model for activity spaces, and study its asymptotic properties. In Section \ref{sec::sim} we present a simulation study. In Section \ref{sec::DA} we apply density ranking and summary curves to the GPS data described in Section~\ref{sec::GPS}.  Finally, in Section \ref{sec::discussion} we comment on the relevance of the proposed set of tools in the context of health research. We provide R scripts that implement our proposed methods  at \url{https://github.com/yenchic/density_ranking}.

\section{GPS data}	\label{sec::GPS}

We employ data from a GPS pilot study that involved three men and seven women that reside in a rural region of sub-Saharan Africa. The study took place in 2016 with the approval of the local biomedical research ethics committee. These data have not been analyzed before. Each study participant was provided with a GPS-enabled Android smartphone for a period of six months. The smartphones together with their voice and data plans whose costs have been covered by the pilot study served as an effective incentive for study participation and adherence to the data collection protocol. The participants were asked to carry the smartphones with them at all times, and also to keep them operational by regularly charging them. All ten participants have been compliant with the protocol of the study, and have returned their devices at the end of the study period.

The Android smartphones employ an assisted GPS system which produces accurate coordinate data with less battery power (allowing a phone to remain charged for at least 48 hours) than traditional GPS devices (e.g., GPS trackers). The positional data that were recorded contain timestamps, smartphone unique identifiers, latitude and longitude coordinates, and information related to the accuracy of the reported coordinates (e.g. satellite connectivity). The smartphones were registered with a Mobile Device Management (MDM) software that allowed the study personnel to manage, secure, monitor and track the smartphones from an easy to use online dashboard. The positional data were securely transmitted to a study database residing on a secure server over cellular or wireless networks using state of the art encryption techniques every time the smartphones had a data connection. The data were deleted from the smartphones immediately after transmission. This protocol guarantees that no confidential positional information could be accessed if a smartphone was lost or stolen.

The ages of the study participants were between 34 and 48 years. They share the same place of work. Their residences are located within a short commute of a couple of kilometers. The rural study area has a township in which most stores and markets are located. The local road network comprises a major primary road that traverses the township and several secondary roads. There are additional unpaved roads about which we did not have GIS data.  The data comprise between 3,500 and 8,500 GPS locations for each of the ten study participants. The MDM software installed on the phones was set to transmit a new location every time a device moved more than 250 meters. For this reason, more locations were recorded for those participants that traveled more. Figure \ref{fig::gps1} shows the GPS locations recorded for one of the study participants who was most active in the rectangular area shown in red in the left panel, but also took several trips to more distant locations.

\begin{figure}[!ht]
\includegraphics[width=1.6in]{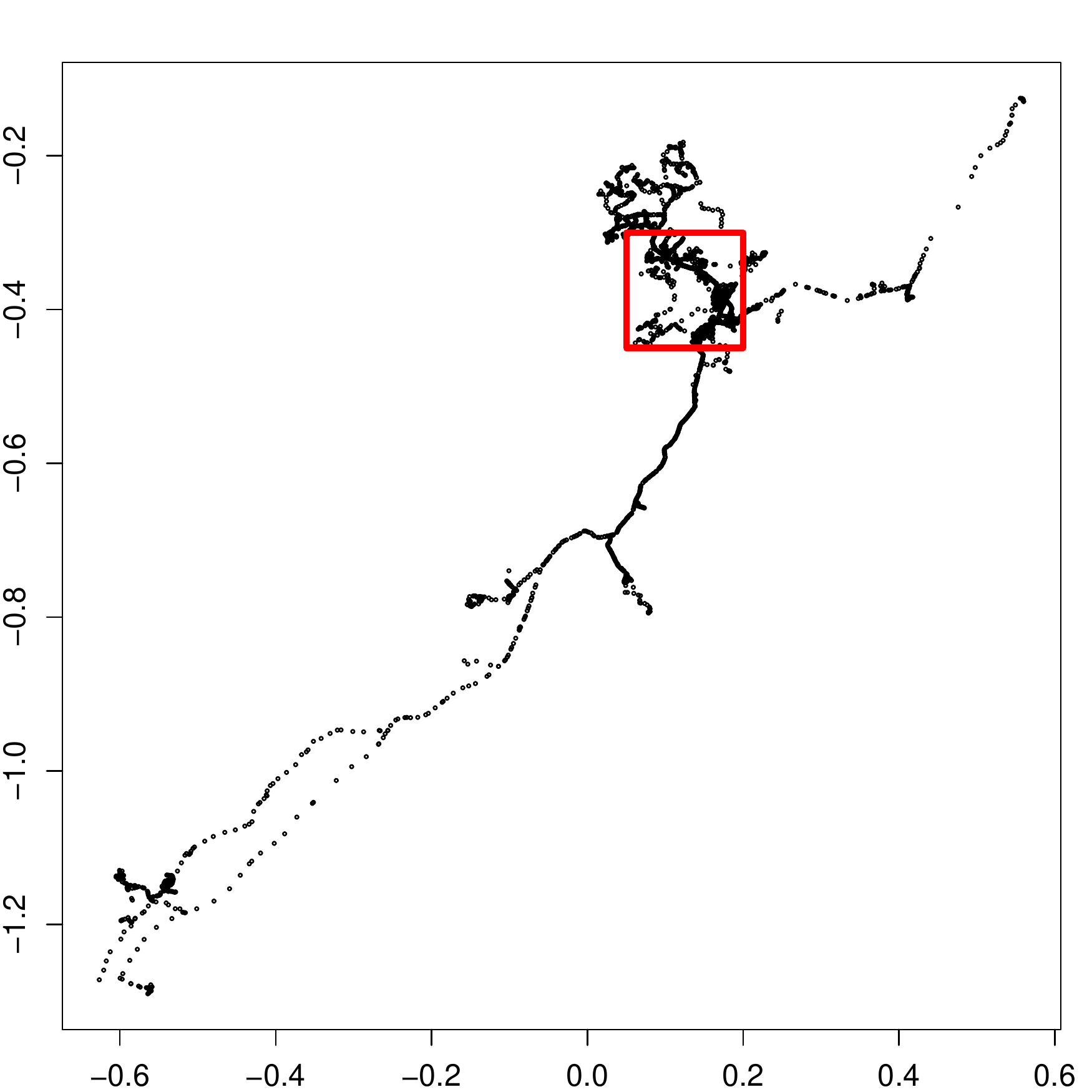}
\includegraphics[width=1.6in]{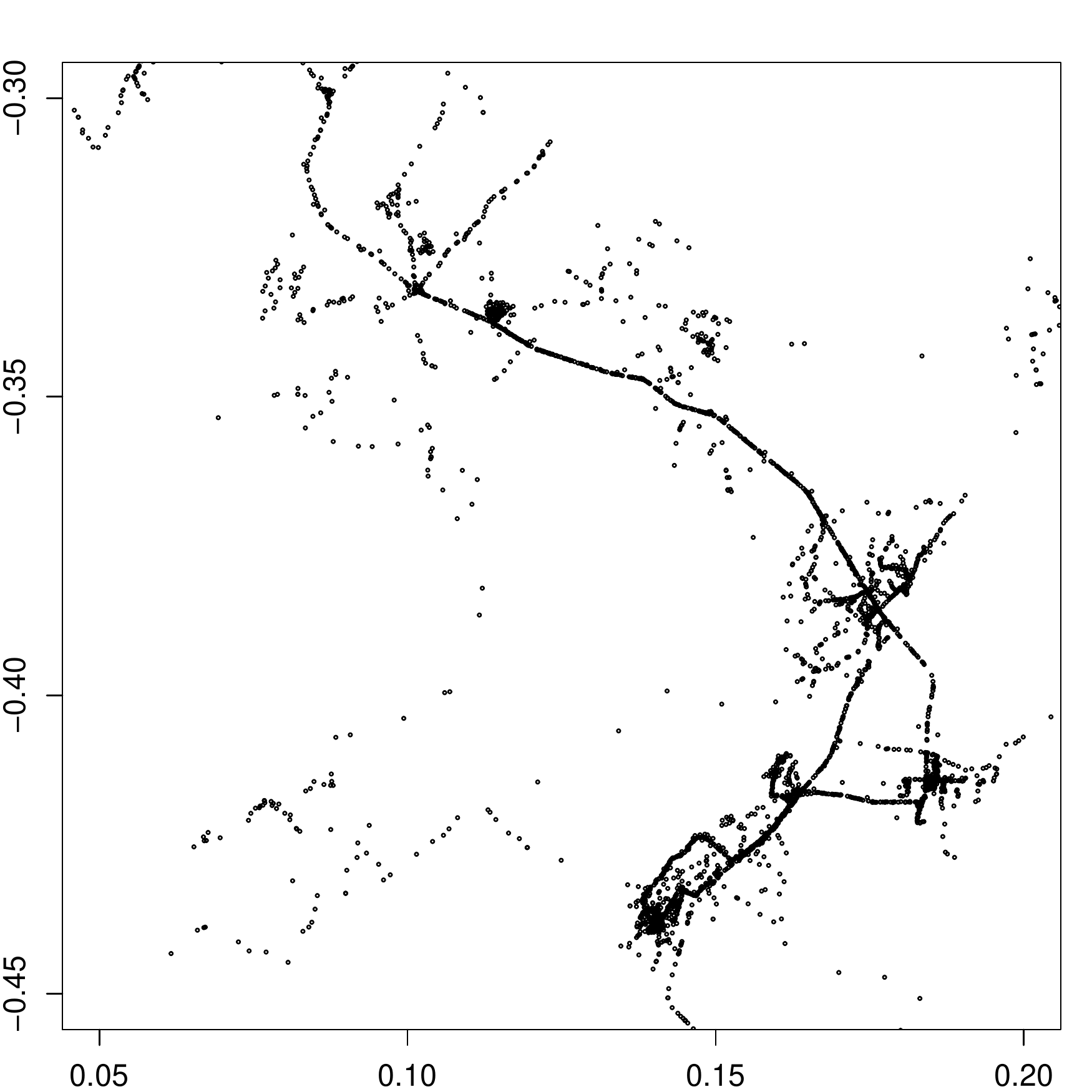}
\includegraphics[width=1.6in]{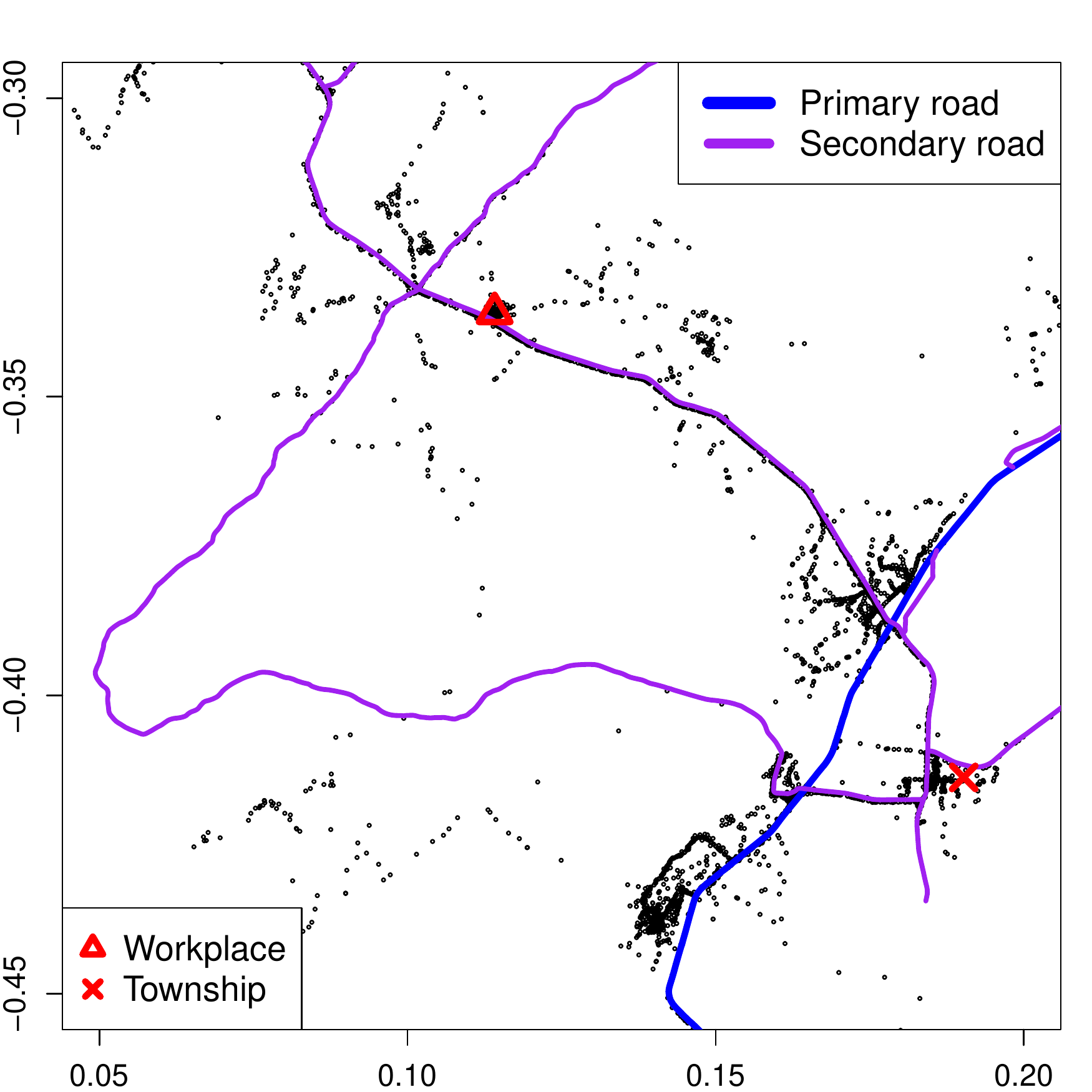}
\caption{Plots of the GPS data from one study participant. The left panel displays the complete GPS records of this individual.  The latitude (x-axis) and longitude (y-axis) coordinates were shifted and scaled to preserve the privacy of the study participants.
%the X-axis and Y-axis have been rescaled though the units are still degree (in terms of latitude and longitude).
The middle panel represents the zoom-in area of the rectangular region shown in red in the left panel.
In the right panel, the relevant GIS data was superimposed on the GPS locations: primary (blue) and secondary (purple) roads, the workplace (red triangle) of the study participants, and the location of the township of the study area (red cross).}
\label{fig::gps1}
\end{figure}

\section{Existent approaches for measuring human activity spaces}	\label{sec::activity}

Activity spaces represent the spatial areas within which an individual has direct contact during their daily travels. However, people do not move randomly in space. Due to the various preferences, needs, knowledge, constraints and limitations of movement, the areas visited by an individual are concentrated around one, two or more anchor locations that serve as origin and destination hubs of the routes followed by an individual. The anchor locations have key material or symbolic meaning for an individual: they include their home and work locations, together with, for example, the location of a child's school, favorite market and grocery store, a preferred entertainment venue or an airport. The probability of visiting a particular spatial location decreases as a function of its distance to the anchor locations, and depends on its relative position with respect to the most frequent directions of daily travel. The shape, structure and spatial extent of activity spaces are a function of the spatial configuration of the anchor locations, and of the routes travelled between and around them \citep{schonfelder-axhausen-2003}.

Characterizing the activity space of an individual involves: (i) determining the number and the spatial configuration of the anchor locations; (ii) identifying the places the individual is most likely to experience in addition to the anchor locations, and differentiating these places from other places which the individual is considerably less likely to come in direct contact with; (iii) mapping the spatial configuration of these locations; and (iv) developing measures that quantify the geometry and spatial structure of the individual's activity space. Such measures capture the individual's degree of mobility while accounting for the underlying preferences for certain travel routes. Activity spaces are not designed to capture the maximal area in which an individual is active. Instead, they consist of one, two or more spatially contiguous areas structured around the anchor location in which an individual regularly engages in activities of daily living, together with the routes used by the individual to travel between these areas.

We denote by $\mathcal{T} = \{ T_1,T_2,\ldots,T_n\}$ the GPS positional data of an individual. The $j$-th location is $T_j=(x_j,y_j,t_j)$ where $x_j$ denotes latitude, $y_j$ denotes longitude, and $t_j$ denotes the time when the location $(x_j,y_j)$ was visited. We assume that $t_1< t_2 < \ldots < t_n$. The set of visited locations $\mathcal{X} = \{ X_1,X_2,\ldots,X_n\}$ where $X_i=(x_i,y_i)$ is the projection of $\mathcal{T}$ onto the latitude and longitude coordinates. The times when the locations were visited together with the order in which locations were visited are lost through this projection. Information about the routes travelled by an individual are comprised in $\mathcal{T}$, but are absent in $\mathcal{X}$. The set of anchor locations are denoted by $\mathcal{A} = \{ A_1,A_2,\ldots,A_{n_0}\}$. We note that $\mathcal{A}$ is not necessarily a subset of $\mathcal{X}$ since some anchor locations might need to be inferred from possibly noisy GPS measurements.

The existent literature has introduced several approaches for characterizing activity spaces. We describe them below, together with their advantages and limitations.

\subsection{Ellipses} This appears to be one of the earliest and most popular method for measuring activity spaces which has been concurrently developed in several research domains such as biological habitat research, transportation research, and human geography \citep{schonfelder-axhausen-2003,schonfelder-axhausen-2004}. Ellipses are fit to the set of visited locations $\mathcal{X}$ based on knowledge of the most relevant anchor locations in $\mathcal{A}$ such as residence and workplace. 

There are two kinds of ellipses: the standard deviational or confidence ellipse, and the home-work ellipse \citep{chaix-et-2012}. The standard deviational ellipse is determined based on the assumption that the locations $\mathcal{X}$ follow a bivariate normal distribution. This distribution can be centered around a central location determined as the arithmetic mean of the unique coordinates in $\mathcal{X}$, or the weighted average by the frequency of visits at some locations. Since these averages might not designate an actual real-world address, the central location can be an anchor location \--- typically the home location which is recognized as the focal point of the lives of most people. The major axis of the standard deviational ellipse is the regression line of the latitude on the longitude coordinates, thus the orientation of the ellipse reflects the sign of the correlation between coordinates. It is customary to report one and two standard deviational ellipses corresponding to 68\% and 95\% coverage probabilities \citep{sherman-et-2005}.

The home-work ellipses differ from the standard deviational ellipses in that they are defined with respect to two anchor locations which become the two focal points of the ellipse. Typically the focal points of the ellipse are selected to be the home and work locations. This defines the major axis of the ellipse. Its minor axis is determined by selecting one additional visited location which could be another anchor location, or the most distant location in $\mathcal{X}$ from the two focal points \citep{newsome-et-1998}. 

Measures that describe an activity space represented through the space inside an ellipse are the area of the ellipse which expresses the extent of the activity space, and the ratio of the length of the major and minor axes which represents the relative extent to which an individual deviates from its most frequently used route (e.g., home to work and back) \citep{newsome-et-1998}.

One major disadvantage of representing activity spaces through ellipses are their relatively inflexible geometry: the spatial distribution of  activity locations is constrained to the shape of the ellipse. Locations inside the ellipse are considered to be likely places of daily activities, while the locations outside the ellipse are viewed as unlikely travel locations. This is a problem because the actual shape of activity spaces could be quite different than that of an ellipse, and could comprise non-overlapping spatial regions. Moreover, ellipses could suggest larger activity spaces since they capture the underlying variability of locations and are not robust to outliers. In addition, an ellipse imposes a symmetry of the activity space around its center even if half of the area covered by the ellipse does not contain any locations in $\mathcal{X}$. To get around these issues, \citet{schonfelder-axhausen-2004} proposed using amalgamations of ellipses constructed around two or more anchor locations (e.g, one ellipse having home location as its center, and another ellipse having the work location as its center), while \citet{rai-et-2007} have shown how to fit three other curved geometrical shapes: the Cassini oval, the bean curve and the superellipse which comprises a circle and an ellipse. Selecting one of these shapes is based on particular assumptions about the form of the activity space: one, two, three or four clusters of locations with or without intermediate locations between them. Nevertheless, determining which (if any) of these assumptions is appropriate for a certain spatial pattern of locations $\mathcal{X}$ cannot be done without performing a visual inspection which is problematic for applications that involve a large number of mobility profiles.

\subsection{Minimum convex polygons} In this approach, the activity space of an individual is defined as the area delimited by the smallest convex polygon that contains all the locations in $\mathcal{X}$. This method has been applied to study both animal and human activity spaces \citep{worton1987,buliung-kanaroglou-2006,fan-et-2008,lee-et-2016}. Although the determination of minimum convex polygons is computationally straightforward, they cannot properly capture the shape of an individual's activity space which is typically irregular due to certain areas in the proximity of the locations in $\mathcal{X}$ being very unlikely to be visited (e.g., inaccessible or undesirable locations). As such, they identify activity spaces as being spatially larger than other approaches \citep{hirsch-2014}. Other shortcomings of minimum convex polygons include: (i) the anchor locations $\mathcal{A}$ and other most frequently visited locations are not represented or even identified; (ii) they imply that an individual is active in only one contiguous spatial area; and (iii) outlier locations in $\mathcal{X}$ can significantly change the coverage and the shape of the resulting activity spaces. The spatial extent of minimum convex polygons is typically measured using their area and perimeter, while their shape is measured through their compactness \citep{manaugh-et-2012,harding-et-2013}. This is a measure of how circular a polygon is defined as the ratio between the area and the perimeter squared, multiplied by $4\pi$. Its values range from near 1 (a polygon very close to a circle) to near 0 (an elongated polygon close to a line). The shape of ellipses can also be measured using their compactness scores.

\subsection{Shortest-path spanning trees} This method employs a more realistic representation of human travel: individuals most often move via road networks instead of by apparition or ``as crow flies'' from one place to another. As opposed to the other three approaches which employ only the locations $\mathcal{X}$, the shortest-path spanning trees are constructed with respect to a road network that spans the reference area, and also with respect to the order in which the locations in $\mathcal{X}$ were visited. The routes followed by an individual during their daily travels are approximated by projecting the locations in $\mathcal{X}$ on the road network, then by connecting each pair of consecutive locations (seen as an origin-destination trip) by the shortest path on the road network between them \citep{schonfelder-axhausen-2003,schonfelder-axhausen-2004}. \citet{golledge-1999} argues that road networks affect the individuals' perception and knowledge of places, therefore activity spaces should be based on the paths followed by the travelers. As such, the activity space of an individual is represented by the spanning tree that covers the part of the network defined by the union of the shortest road network paths that connect consecutive visited locations. The spanning tree can be measured using its length, or using the total area of buffers with a fixed length (e.g., 200 meters) around the road network segments. These buffers attempt to capture the space around the road network segments that might be known to an individual by walking around \citep{kim-et-2015}. Anchor locations and segments that are more intensely used on the road network can be determined based on the visitation frequencies.

An advantage of the shortest-path spanning trees is that this approach moves away from the assumption that individuals have a continuous knowledge about the space around and between the locations they visit \--- ellipses and minimum convex polygons are based on this assumption. Their shortcomings come from their dependence on the availability of road network data. Such data might not have been collected at all or have lower quality in rural areas or in low resource countries. Moreover, if the visited locations are recorded at larger time intervals, approximating the route followed by an individual by the shortest path between two consecutive locations might be crude: the individual might have traveled significantly more than the shortest path would indicate.

\subsection{Kernel density estimation}\label{sec::kde} This approach considers a raster grid cells that partitions a wider area that includes the set of visited locations $\mathcal{X}$ which is seen as a point pattern. An activity surface over this wider area is generated by assigning a value to each cell in the raster based on the distances from the center of the cell to the locations $\mathcal{X}$ \citep{kwan-2000,buliung-2001}. The probability that the individual that visited the locations $\mathcal{X}$ was also active in a particular cell is proportional with the value assigned to that cell. The kernel density estimator (KDE) is the sum of ``bumps'' centered at the locations $\mathcal{X}$. The estimate of the bivariate density at grid point $x$ (also referred to as the intensity at $x$) is given by \citep{silverman-1986}:
\begin{eqnarray}\label{eq:kde}
\hat{\sf p}(x) = \frac{1}{nh^2} \sum_{i=1}^n K\left(\frac{d_i(x)}{h}\right).
\end{eqnarray}
Here $K(\cdot)$ is a kernel, $h$ is the bandwidth or smoothing parameter, and $d_i(x)$ is the distance between the grid point $x$ and the $i$-th visited location $X_i=(x_i,y_i)\in \mathcal{X}$. The most usual choice for $K(\cdot)$ is a radially symmetric unimodal probability density function such as the bivariate normal density. However, since the number of visited locations $\mathcal{X}$ could be very large for some GPS studies, it is preferable to employ a kernel that does not require evaluating the value of the kernel at all points in $\mathcal{X}$ for every grid point. For example, consider the quartic (biweight) kernel function \citep{silverman-1986}:
$$
 K_2(x) = \left\{ \begin{array}{cc} \frac{3}{\pi}\left(1-x^2\right)^2, & \mbox{if } |x|<1,\\ 0, & \mbox{otherwise.}\end{array}\right.
$$
With this choice, the KDE from Eq. (\ref{eq:kde}) becomes:
\begin{eqnarray}\label{eq:kdequart}
 \hat{\sf p}(x) = \frac{3}{\pi h^2} \sum\limits_{d_i(x) < h} \left( 1 - \left(\frac{d_i(x)}{h}\right)^2\right)^2.
\end{eqnarray}
Thus, locations in $\mathcal{X}$ outside a circle with radius $h$ centered at $x$ are dropped in the evaluation of $\hat{\sf p}(x)$. The probabilities of visiting grid cells that are at larger distances from the most frequently visited areas will be smaller compared to the probabilities of visiting grid cells that are at smaller distances. The choice of bandwidth $h$ is very important as larger bandwidths give more smoothing. However, for the KDE (\ref{eq:kdequart}), $h$ also represents the maximum distance of spatial interaction between locations. Therefore the choice of $h$ for a particular application could reflect the understanding of proximity and neighborhood in daily travel for the area in which the location data was collected \citep{schonfelder-axhausen-2003}. 

In the KDE approach, the activity space of an individual comprises all the grid cells with an estimated probability (density) of visitation above a certain threshold $\tau_1>0$. The anchor locations can be identified as those grid cells with an estimated probability density of visitation above a second threshold $\tau_2\in (\tau_1,\infty)$. Kernel density estimation can identify activity spaces of any shape, and can also estimate the corresponding anchor locations which is something the ellipse and the minimum convex polygon methods cannot do. The shortest-path spanning trees rule out locations that are not on the road network they were defined on. For this reason, the KDE approach seems to be the most flexible existent approach for activity space determination. Measuring the resulting activity spaces can be done by calculating the area covered by the grid cells included in them. It is possible to eliminate some of these areas if they are known to be unfavorable to activities of daily living (e.g., heavy industrial and utility areas), thereby refining the shape of the activity spaces \citep{schonfelder-axhausen-2003,schonfelder-axhausen-2004}.

\section{Density ranking}	\label{sec::DR}
Despite its flexibility in measuring activity spaces, kernel density estimation sometimes fails to yield adequate results when applied to GPS datasets. 
Consider the left panel of Figure~\ref{fig::alpha} in which we show the KDE of locations from the region in the middle panel of Figure~\ref{fig::gps1}. 
Although it correctly identifies two peaks with the highest concentration of locations, the KDE does not capture much of the underlying structure of the GPS data. In this section, we discuss an alternative to KDE called density ranking that captures much more of the underlying mobility patterns of this individual \--- see the middle and right panels of Figure~\ref{fig::alpha}. It is apparent that many finer structures are not discernible using KDE, but they can be easily recognized when using density ranking. The KDE map only shows two grid cells that have high intensity: the workplace and another location that might be the home of this individual. On the other hand, the density ranking maps show the existence of numerous other grid cells located on the spatial trajectory followed by this individual. These regions represent the location of the township, road intersections or road segments. 

\begin{figure}[!ht]
\includegraphics[width=1.6in]{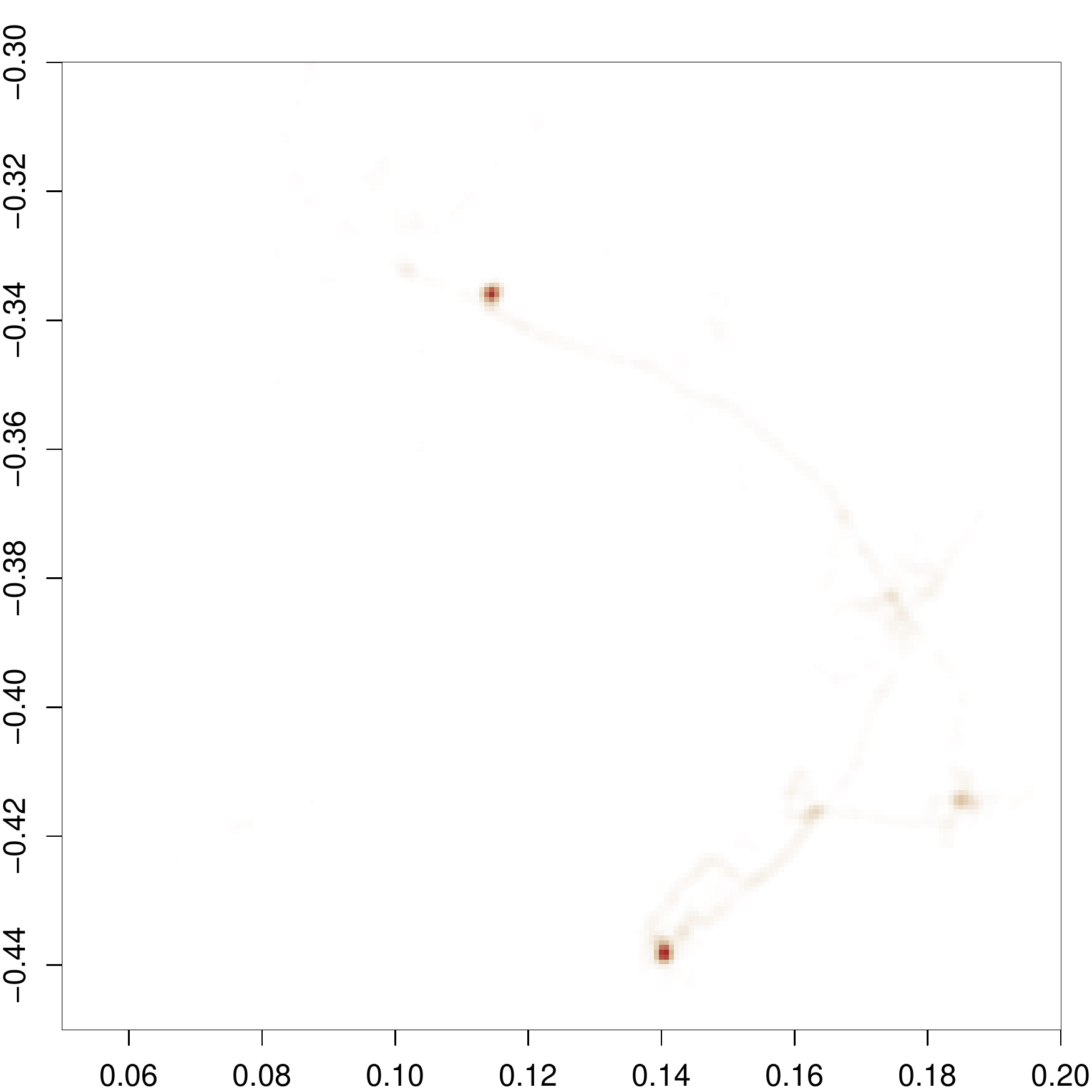}
\includegraphics[width=1.6in]{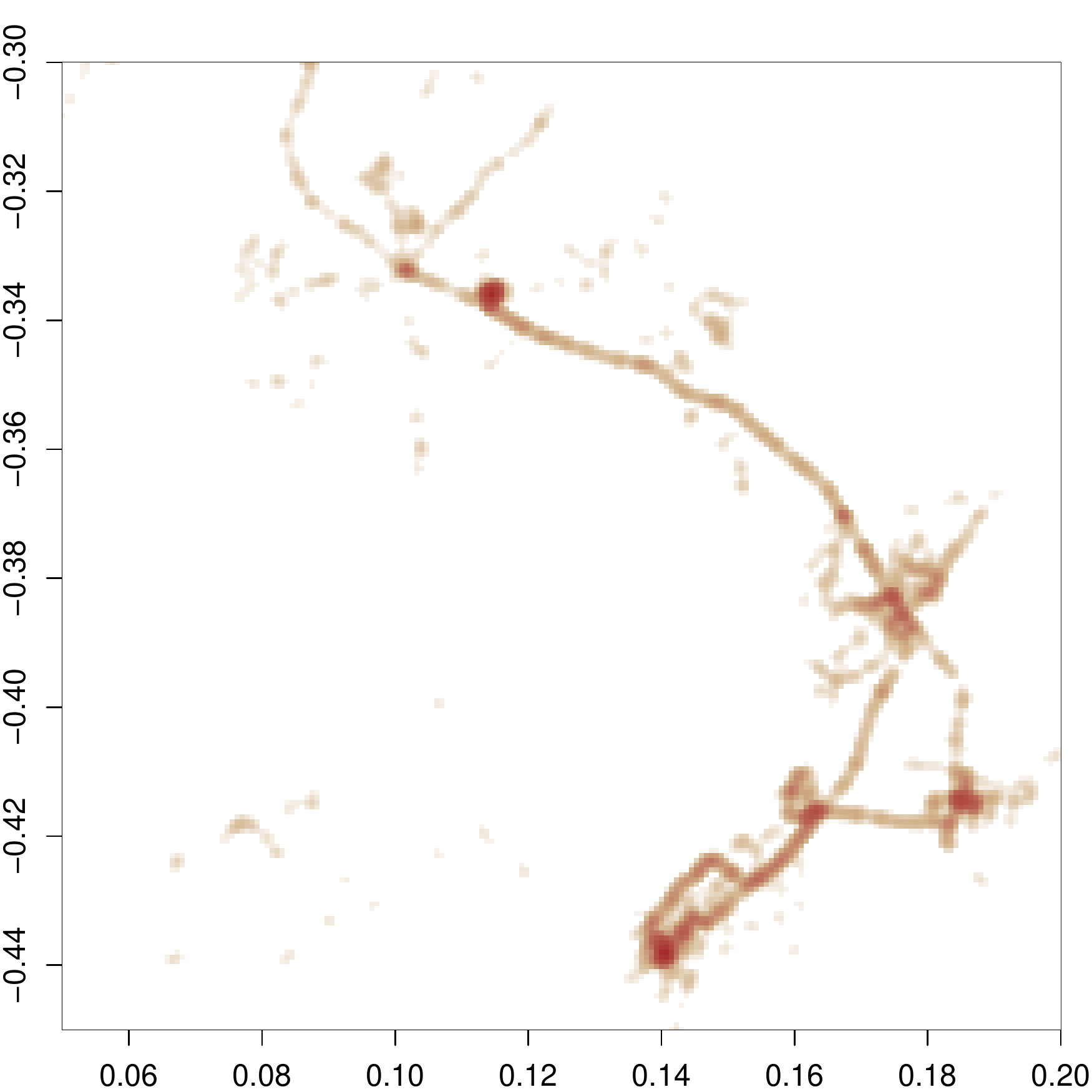}
\includegraphics[width=1.6in]{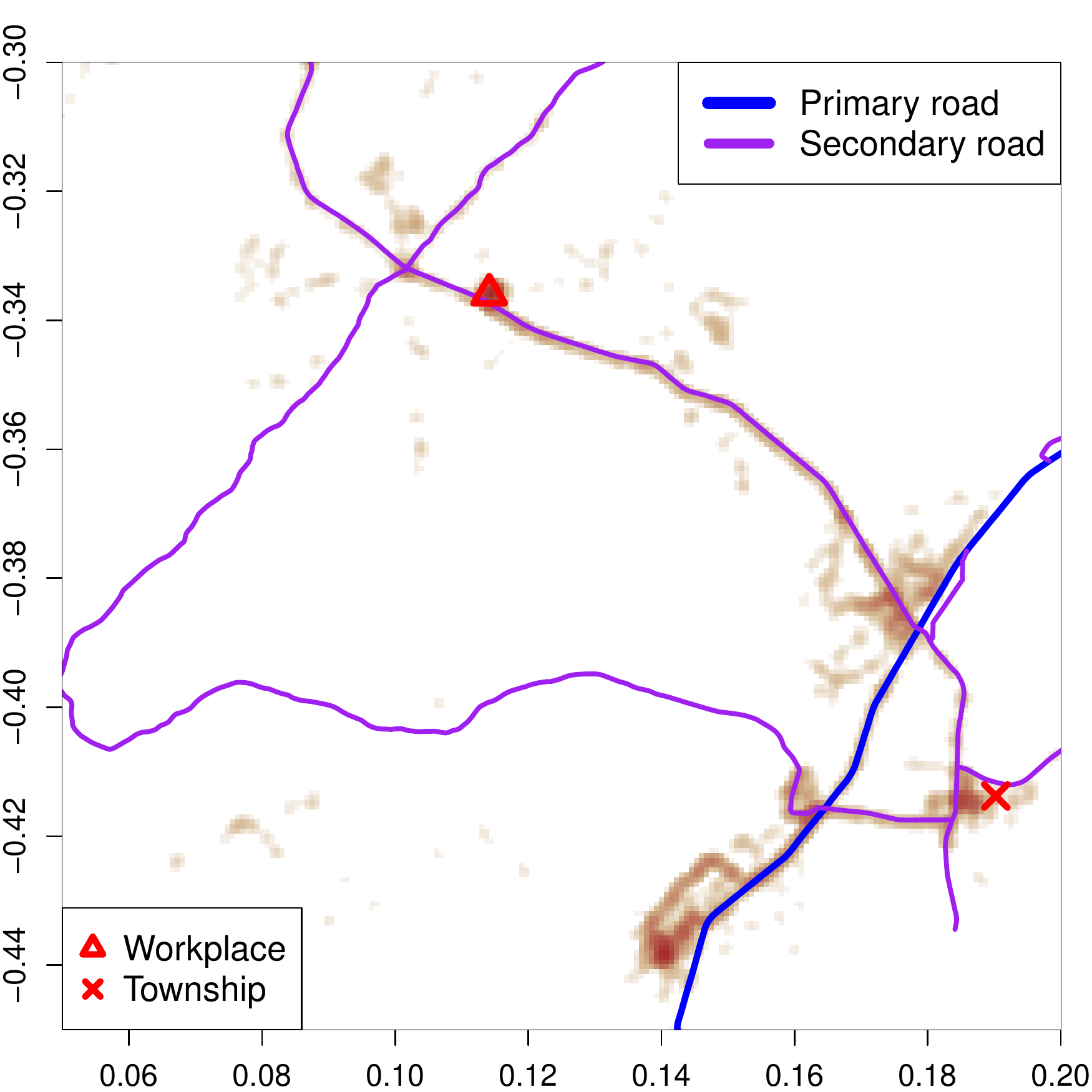}
\caption{
%xxxx
A comparison between KDE and density ranking. In the left panel, we display the density contours from the KDE associated with the locations shown in the middle and right panels of Figure~\ref{fig::gps1}. In the middle panel, we show the contours identified by density ranking. In the right panel, we superimpose GIS data to the density ranking contours.
%The regions circled by black contours are where the density ranking is above $0.5$.
%Comparing to the road maps, we see that these key locations are intersections of roads,
%work place, township, and a region possible around the individual's home location. 
%This is because many GPS records are degenerated at the two peaks (probably the home and work address), 
%making the KDE unstable.
%On the other hand, the density ranking has been shown to be stable even under such a scenario (a singular distribution).
}
\label{fig::alpha}
\end{figure}

Density ranking is a quantity derived from the KDE defined as
$$
\hat{\alpha}(x) = \frac{1}{n}\sum_{i=1}^n {\sf I}(\hat{\sf p}(X_i)\leq \hat{\sf p}(x)),
$$
where ${\sf I}(\Omega)$ is the indicator function. The density ranking function $\hat{\alpha}(x)$ is the fraction of observations in $\mathcal{X}=\{X_1,X_2,\ldots,X_n\}$ whose estimated density is lower than the estimated density of the given point $x$.
This function was called the $\alpha$-function in \citet{chen2016generalized}. The density ranking function $\hat{\alpha}(x)$ is a probability-like quantity that takes values between $0$ and $1$. It has a natural relationship with the rank of the data points with respect to the KDE $\hat{\sf p}$. Let $R_i = \sum_{j=1}^n {\sf I}(\hat{\sf p}(X_j)\leq \hat{\sf p}(X_i))$ be the rank of $X_i$ with respect to $\hat{\sf p}$. We have $R_i=1$ if $X_i$
has the lowest density, and $R_i=n$ if $X_i$ has the highest density. %; note that $\underset{i=1,\cdots,n}{\max} R_i= n$.
Then 
$$
\hat{\alpha}(X_i) = \frac{R_i}{n},
$$
which implies that the density ranking at each observed data point  is just the relative ranking of that point. 

%What is the reason why density ranking successfully reveals more grid cells than KDE that are relevant for the trajectory followed by an individual? It turns out that many GPS locations are distributed on a lower dimensional structure so the density function does not exist, leading to an inconsistent KDE. Several GPS records fall in the same grid cell if they are very close to each other, creating  a point mass. These cases most likely correspond to locations that are repeatedly visited by an individual. When an individual is traveling on a road segment, their GPS records distribute along a one-dimensional structure (a curve). In both situations, the GPS locations will be distributed  on a support that has a dimension less than 2, making the  probability density function ill-defined. This causes KDE to fail to associate high intensities to grid cells along the trajectory followed by an individual. Density ranking has a powerful property that guarantees its convergence even when the underlying distribution contains lower dimensional structures \citep{chen2016generalized}. For this reason, it is a more appropriate quantity to employ in the determination of activity spaces from GPS data. We formally prove a version of this result that is specific to human activity spaces in Section \ref{sec::model}. 

Density ranking has a straightforward interpretation related to the locations visited by an individual:  for a point $x$ with $\hat{\alpha}(x)=0.8$, the probability density (measured by the KDE $\hat{\sf p}$) at point $x$ is higher than the probability density of $80\%$ of all observed GPS locations. We say that $x$ is in the region of the \emph{top $20\%$} activity. Given a level $\gamma\in[0,1]$, the level set of density ranking 
$$
\hat{A}_\gamma =\{x: \hat{\alpha}(x)\geq 1-\gamma\},
$$
can be interpreted as
\emph{the area of the top $\gamma\times100\%$ activities}. 
The set $\hat{A}_\gamma$ is the region within the contours of level $1-\gamma$.
%so it is also called a \emph{level set} \citep{garcia2003level,rinaldo2010generalized}. 
Note that  $\hat{A}_\gamma$ is related to the \emph{minimum volume set} \citep{polonik1997minimum,garcia2003level,scott2006learning}
and can be interpreted as a density level set with a probability content of $\gamma$ \citep{cadre2013estimation}.  
In this view, $\hat{A}_\gamma$ can be interpreted as an estimator of the smallest (in terms of volume) area covering at least $\gamma\times100\%$ activities.

As explained in Section \ref{sec::kde}, given two pre-specified levels $\tau_1,\tau_2$ with $\tau_1<\tau_2$, the activity space based on the KDE $\hat{\sf p}$ comprises the grid cells $x$ with $\hat{\sf p}(x)\geq \tau_1$. The anchor locations are those grid cells $x$ with $\hat{\sf p}(x)\geq \tau_2$. Similar definitions of activity spaces and anchor locations can be given based on density ranking. We choose two levels $\gamma_1,\gamma_2\in(0,1)$ with $\gamma_2<\gamma_1$. We define $\hat{A}_{\gamma_1}$ to be the \emph{top $\gamma_1\times100\%$ activity space or $\gamma_1$-activity space}. Then anchor locations are defined as the $\gamma_2$-activity space. Choosing particular levels for the determination of human activity spaces or anchor locations can be done by examining the top $90, 80,\ldots, 10$\% activity spaces. In Section \ref{sec::sim} we show that the summary curves we introduce in the next section can be used to guide the choice of levels.

\section{Summary curves}\label{sec::tda}

Based on density ranking, we obtain a two dimensional function (a map) of human activity spaces. However, comparing maps associated with the activity patterns of multiple individuals is not straightforward without adequate summaries of the shape of these functions. To define such summaries, we use tools from topological data analysis \citep{edelsbrunner2008persistent,wasserman2016topological,chazal2017introduction}. 
Specifically, we describe three types of summary curves that quantify the shape of a two dimensional function. These curves provide additional information about the geometry, size and structure of human activity spaces. 

\subsection{Mass-volume curves}

Given a level $\gamma$, the size of the region $\hat{A}_\gamma$ can be used to quantify an individual's mobility in terms of the spatial extent of the $\gamma$-activity space. We measure size with the mass-volume function \citep{garcia2003level,clemenccon2013scoring,clemenccon2017mass} which is defined as
$$
\hat{V}(\gamma) = {\sf Vol}(\hat{A}_\gamma),
$$
\noindent where ${\sf Vol}(A) = \int_A dx$ is the volume of the set $A$.  For example, if an individual has $\hat{V}(0.2) = 3\mbox{ km}^2$, we say that the top $20$\% activities of this individual occur within a region of size $3\mbox{ km}^2$.

We subsequently define the \emph{mass-volume curve} $\hat{V}=\{(\gamma, \hat{V}(\gamma)): \gamma \in [0,1]\}$ which describes how the volume of the $\gamma$-activity space $\hat{A}_\gamma$ evolves when we vary the level $\gamma$. Mass-volume curves can be used to compare the degree of mobility of two individuals. Consider two example individuals with mass-volume curves $\hat{V}_1$ and $\hat{V}_2$ such that $\hat{V}_1(\gamma_0)>\hat{V}_2(\gamma_0)$ for some level $\gamma_0\in [0,1]$. We say that first individual has a higher mobility than the second individual  in terms of the top $\gamma_0\times100\%$ activities.

\subsection{Betti number curves}

The mass-volume curve quantifies the activity space in terms of its size, but it does not provide any information about the shape of the activity space. Key concepts from topological data analysis 
\citep{edelsbrunner2008persistent,wasserman2016topological,chazal2017introduction} turn out to be very useful for this purpose. Two points in a set $\mathcal{S}$ are connected if and only if there exists a curve inside $\mathcal{S}$ that connects them. The set $\mathcal{S}$ is connected if any two points in the set are connected. The connected components of $\mathcal{S}$ are the induced partition from this relation. The connected components of $\mathcal{S}$ are a partitioning of $\mathcal{S}$ into sub-regions. Each connected component must not overlap with other connected components, and must be connected. Two points that belong to two different connected components of $\mathcal{S}$ cannot be connected with a curve inside $\mathcal{S}$.

We consider the connected components of the $\gamma$-activity space $\hat{A}_\gamma$. We define the Betti number function
$$
\hat{\beta}(\gamma) = \mbox{number of connected components of } \hat{A}_\gamma,
$$
and the Betti number curve $\hat{\beta} = \{(\gamma, \hat{\beta}(\gamma)): \gamma \in [0,1]\}$. This curve captures how the number of connected components of the $\gamma$-activity space changes with the level $\gamma$. Note that the Betti number function is related to the number of local maxima of $\hat{\alpha}(x)$ and $\hat{\sf p}(x)$. In the cluster analysis literature, the value $\hat{\beta}(\gamma)$  is interpreted as the number of clusters \citep{hartigan1975clustering,rinaldo2010generalized}.
The Betti number curve is closely related to the barcode plot in topological data analysis 
\citep{ghrist2008barcodes,wasserman2016topological}.

We provide an example in  Figure~\ref{fig::TP}. The left panel shows a univariate function with four local maxima that correspond to levels $b_1,b_2,b_3,b_4$, and four local minima  that correspond to levels $d_1,d_2,d_3,d_4$. A level set $\hat{A}_\gamma$ comprises those regions where the function has a value above $1-\gamma$.  A new connected component is created when, as $\gamma$ increases from $0$ to $1$, it passes one of $b_1,b_2,b_3,b_4$. An existing connected component disappears when $\gamma$ passes one of $d_1,d_2,d_3,d_4$. Each of the orange vertical line segments represents a connected component: its upper and lower ends  correspond to the birth and death time of this connected component, respectively. The middle panel shows the Betti number curve that corresponds to the function in the left panel. The Betti number curve goes up when the level $\gamma$ hits the density ranking value of a local maximum, and may drop when passing through the density ranking value of a local minimum or a saddle point. In this example, the Betti number curve increases whenever it passes $b_1,b_2,b_3,b_4$, and decreases when it passes $d_1,d_2,d_3,d_4$. For a given value $\gamma$, the Betti number function $\hat{\beta}(\gamma)$ tells us the number of connected components in the top $\gamma\times 100\%$ activity region $\hat{A}_\gamma$. For example, if $\hat{\beta}(0.2)= 2$, the region of top $20\%$ activities has two disjoint components. This implies that the individual's top $20\%$ activities are concentrated around two areas that could correspond with the locations of this person's home and workplace. Individuals that record higher values of the Betti number function are those who tend to repeatedly visit a larger number of spatially distinct locations. 

\begin{figure}[!ht]
\center
\includegraphics[width=1.6in]{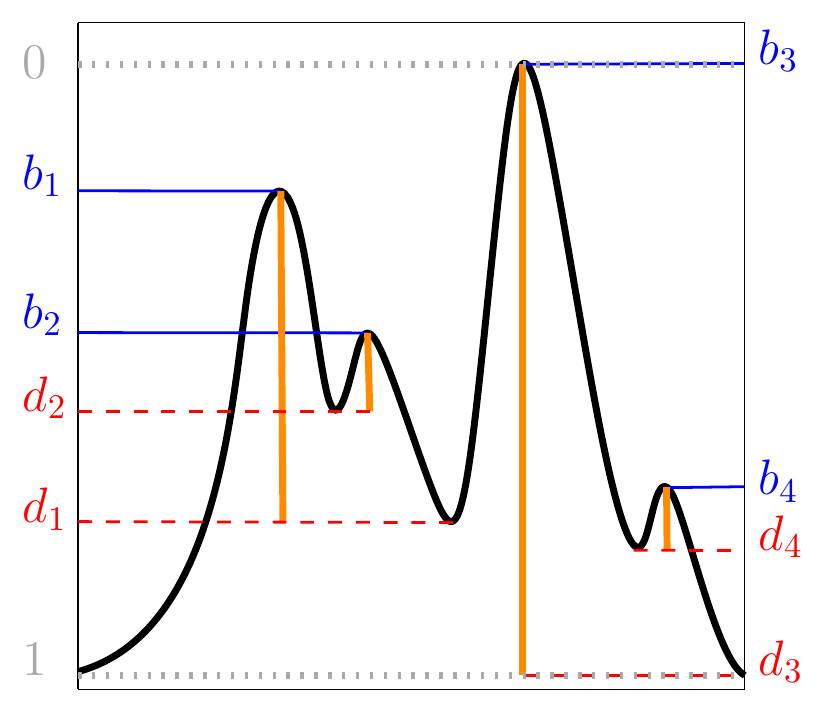}
\includegraphics[width=1.6in]{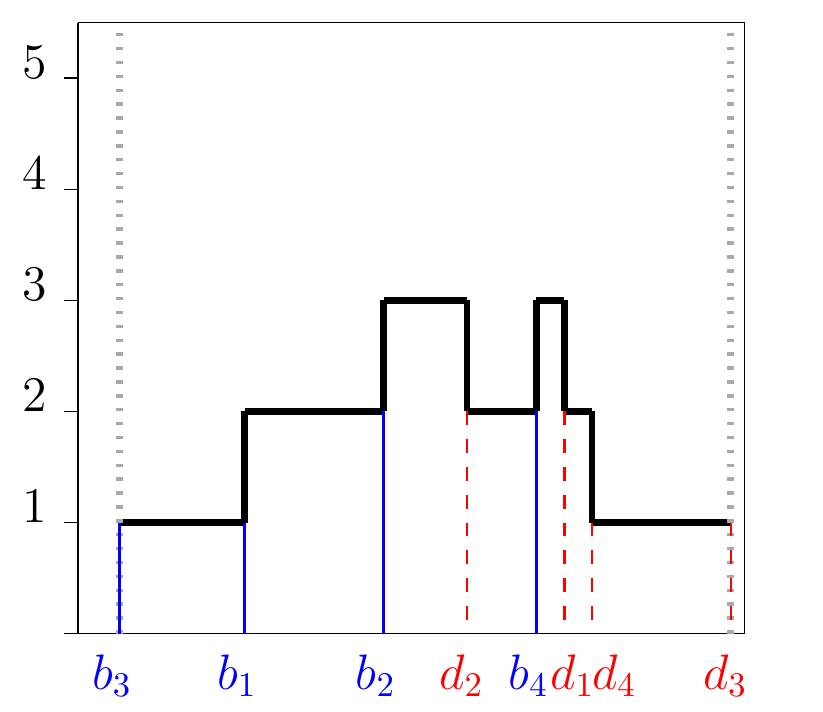}
\includegraphics[width=1.6in]{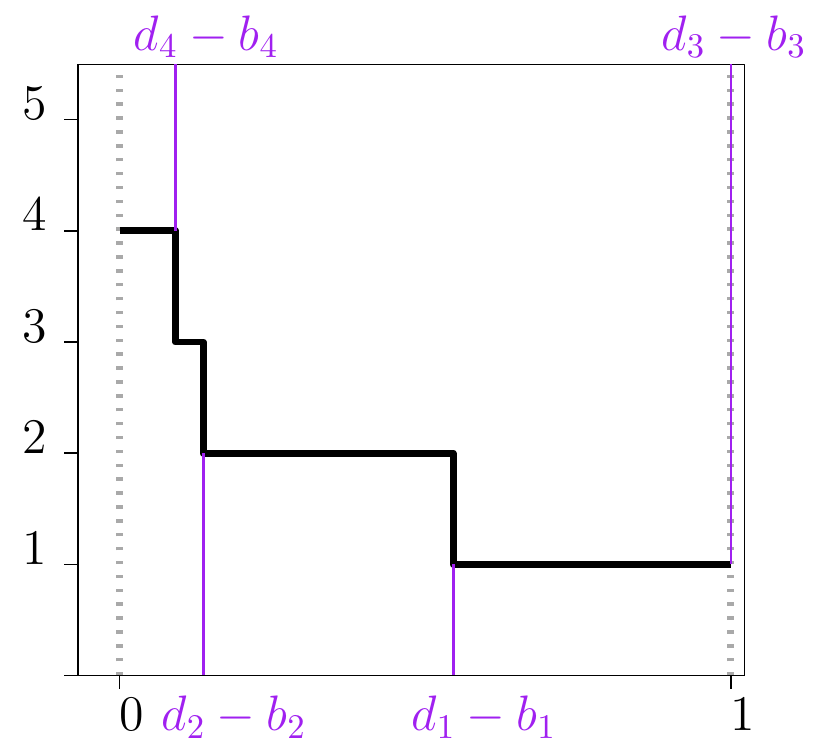}
\caption{
An illustration of how the Betti number curve (middle panel) and the persistence curve (right panel) are computed for the function shown in the left panel.}
\label{fig::TP}
\end{figure}

\subsection{Persistence curves}

%xxx: need to revise this part a bit.
%The mass-volume curve and Betti number curve quantify the activity space and mobility
%of individuals but they are not useful in 
%analyzing anchor locations of an individual. 
%Anchor locations are generally where an individual visited frequently 
%and spent a substantial amount of time. 
%Here we propose a new approach called \emph{persistence curve} that quantifies
%the number of anchor locations.

The previous two types of curves focus refer to specific $\gamma$-activity spaces. Next we define a third type of curve called a persistence curve that simultaneously consider all levels $\gamma\in [0,1]$. The concept of persistence plays a key role in persistent homology, a branch of topological data analysis \citep{edelsbrunner2008persistent,wasserman2016topological,chazal2017introduction}. The persistence curve is related to the accumulated persistence function \citep{biscio2016accumulated}. The Betti number curve and the persistence curve can be viewed as functional summaries
of topological features \citep{berry2018functional}. 

We first define the persistence of a connected component.  When we vary the level $\gamma$, new connected components may be created
and existing connected components may disappear by merging with other connected components. We define the birth time of a connected component to be the level when this component is created, and its death time to be the level at which this component disappears. A connected component is created at the level of the density ranking of a local mode, and is often
eliminated at the level of the density ranking of a local minimum or a saddle point.  When two connected components merge into one, we apply the following seniority rule \citep{wasserman2016topological,chazal2017introduction}: the older one (created at a lower level) stays alive while the younger one (created at a higher level) is eliminated. We define the death time of the connected components at level $\gamma=0$ to be $0$.

In the left panel of Figure~\ref{fig::TP}, the two end points of an orange line segment correspond to the birth (creation) and 
death (elimination) of a connected component.  The corresponding levels of the end points, $b_\ell$ and $d_\ell$, are the birth time and death time of that connected component. 
There is a direct relationship between birth and death times and the Betti number curve: the Betti number increases by 1 whenever it passes the birth time of a connected component,
and it decreases by 1 when it passes the death time of a connected component.
In Figure~\ref{fig::TP}, two connected components are created at levels $b_1$ and $b_2$, and are eliminated at levels $d_1$ and $d_2$. At $b_1$ and $b_2$, the Betti number increases by 1, and it decreases by 1 at $d_1$ and $d_2$. Since $b_1<b_2$, the connected component created at $b_1$ is older than the connected component created at $b_2$. When the two connected components merge at level $d_2\in (b_2,d_1)$, the connected component created at $b_1$ remains, while the connected component created at $b_2$ is eliminated.

For each connected component, its persistence (also called life time) 
is the difference between the birth and the death time. 
In Figure~\ref{fig::TP}, the length of an orange line segment is the persistence
of that connected component. We define the persistence function
$$
\hat{\rho}(t) = \mbox{number of connected components whose persistence }\geq t,
$$
and the persistence curve $\hat{\rho}=\{(t, \hat{\rho}(t)): t \in [0,1]\}$. An example persistence curve is shown in the right panel of Figure~\ref{fig::TP}. The persistence curve is a non-increasing curve since we are thresholding on the life time of connected components. There will always be a connected component with a life time close to $1$ because, by definition, the data point with the highest KDE value has rank equal to the sample size $n$, making its density ranking equal to 1. Due to the resolution of the underlying grid used, it is possible to see a connected component with life time close to but less than 1.

%%% 
%%% need to think about how to associate persistence curve to the anchor locations or activity space
%%%

The persistence curve provides new information about the spatial distribution of the activity space. Unlike the mass-volume curve or Betti number curve that describes characteristics of level sets $\hat{A}_\gamma$ at particular levels $\gamma$, the persistence curve characterizes the collection of all level sets $\left\{\hat{A}_\gamma: \gamma \in[0,1]\right\}$. This is because in order to compute the persistence of each connected component, we need to consider various levels to determine its persistence. An individual has a high persistence curve when the corresponding density curve has many highly persistent connected components. These are regions this individual repeatedly visits: most likely, these represent their anchor locations. This type of information is not directly related to a particular $\gamma$-activity space. Instead, it is a quantity describing patterns across activity spaces at different levels.

\section{A mixture model for human activity spaces}	\label{sec::model}

In this section we propose a statistical model that captures the most significant features of human activity spaces. We denote by $\P_{\sf GPS}$ the probability distribution that defines the activity space of an individual. The observed locations $\mathcal{X}=\{X_1,\ldots,X_n\}$ are independent samples from $\P_{\sf GPS}$. We write this distribution as a mixture with three components
\begin{equation}
\P_{\sf GPS}(x) = \pi_0 \P_0(x) + \pi_1 \P_1(x) + \pi_2 \P_2(x),
\label{eq::GPSfn}
\end{equation}
where $\P_0(x)$ is an atomic distribution, $\P_1(x)$ is a one-dimensional distribution, $\P_2(x)$ is a two dimensional distribution, and $\pi_0+\pi_1+\pi_2 = 1$ with $\pi_j\geq 0$ are proportions. The three components of the mixture \eqref{eq::GPSfn} represent the key elements of the activity space represented by $\P_{\sf GPS}$: $\P_0$ is a distribution that puts probability on the anchor locations $\mathcal{A}$; $\P_1$ is a distribution describing the roads $\mathcal{R}$ used by an individual when traveling between anchor locations; and $\P_2$ is a distribution describing the areas $\mathcal{O}$ around the anchor locations in which an individual moves. %As an aside, we note that the model in Eq. \eqref{eq::GPSfn} does not necessarily capture the actual data generating process of the spatial pattern of observed locations $\mathcal{X}$.
Although $\P_0$ and $\P_1$ do not have conventional probability density functions, they admit a generalized density function called the Hausdorff density \citep{preiss1987geometry,mattila1999geometry}. Let $\B(x,r)$ be the ball centered at $x$ with a radius $r>0$. For a positive integer $s$, the $s$-dimensional Hausdorff density ($s$-density) at $x$ given $\P_{\sf GPS}$ is
$$
\mathcal{H}_s(x) = \lim_{r\rightarrow0}\frac{P_{\sf GPS}(B(x,r))}{C_s \cdot r^s},
$$
where $C_j$ is the volume of an $j$-dimensional unit ball ($C_0=1$, $C_1=2$, and $C_2 = \pi$).  We denote by $\p_0$, $\p_1$ and $\p_2$ the $0$, $1$, and $2$-dimensional Hausdorff densities given $\P_0$, $\P_1$ and $\P_2$. Namely,  $\p_0(x)$ is a mass at point $x$, $\p_1(x)$ is a one dimensional density value at $x$, and $\p_2(x)$ is a two dimensional density value at $x$. Furthermore, $\mathcal{A}, \mathcal{R},$ and $\mathcal{O}$ represent the support of $\p_0$, $\p_1$ and $\p_2$, respectively. 

%To define $\alpha(x)$, the population quantity of the density ranking, 
%we first define a `density' corresponds to each $P_j$ using the 
%derivative with respect to different dimensional volume measure. 
%Specifically, we define
%$$
%p_j(x) = \frac{dP_j(x)}{d\mu_j(x)},
%$$
%where $\mu_j(x)$ is a $j$-dimensional volume measure for $j=1,2$ and $\mu_0(x)$ is the counting measure. 

We define the dimension $\omega(x)$ of a point $x$ with respect to $\P_{\sf GPS}$ as follows:
\begin{equation}
\omega(x) = \begin{cases}
0,&\mbox{ if } x\in \mathcal{A},\\
1,&\mbox{ if } x\in \mathcal{R}\backslash\mathcal{A},\\
2,&\mbox{ if } x\notin\mathcal{A}\cup\mathcal{R}.
\end{cases}
\label{eq::omega}
\end{equation}
The dimension of an anchor location is $0$. The dimension of a location on a road inside the activity space that is not an anchor location is 1. The dimension of all the other locations is 2.

The following result shows that there are two equivalent ways to define the dimension $\omega(x)$.

\begin{thm}
Given assumption (S) from Appendix \ref{sec::assumption}, the definition of $\omega(x)$ in equation \eqref{eq::omega} is equivalent with the following two definitions:
$$
\omega(x) = \max \{s: \mathcal{H}_s(x)<\infty, s=0,1,2\} ,
$$
and
$$
\omega(x) =
\begin{cases}
0,&\quad \mbox{ if } \p_0(x)>0,\\
1,&\quad \mbox{ if } \p_0(x)=0, \p_1(x)>0,\\
2,&\quad \mbox{ if } \p_0(x)=0, \p_1(x)=0, \p_2(x)\geq 0.
\end{cases}
$$
\label{lem::equiv}
\end{thm}
The proof of Theorem \ref{lem::equiv} is given in Appendix \ref{sec::proofs}.
%\cite{chen2016generalized} defines density ranking using the second expression, $\omega(x) = \max \{s: \mathcal{H}_s(x)<\infty\}$, for a general probability measure. In the case of GPS data and the study of activity space,  our definition is better because $\omega(x)$ has a direct meaning in activity space. 
Using $\omega(x)$ and $\p_j(x)$, $j=0,1,2$, we define a ranking comparison between two points $x_1$ and $x_2$. We write
$
x_1 \succ x_2
$
if
$$
\omega(x_1)<\omega(x_2), \qquad \mbox{or}\quad \omega(x_1)=\omega(x_2), \quad \p_{\omega(x_1)}(x_1)>\p_{\omega(x_2)}(x_2).
$$
To compare two points, we first compare their dimension. The point that has the higher dimension is ranked higher.  If both points have the same dimension, we compare their density value in that dimension. Then the population quantity that density ranking is approximating is
\begin{equation}
\alpha(x) = \P_{\sf GPS} (x\succeq X_1),
\label{eq::GPS_DR}
\end{equation}
where $X_1$ is a random variable with distribution function $\P_{\sf GPS}$. Theorem~\ref{lem::equiv} is needed to prove the next result that shows that $\hat{\alpha}(x)$ is a consistent estimator of $\alpha(x)$, which explains why density ranking yields stable results in measuring human activity spaces.

\begin{thm}
Given assumptions (K1-2), (S) and (P1-2) from Appendix \ref{sec::assumption}, and $\frac{nh^6}{\log n}\rightarrow \infty, h\rightarrow0$, we have
\begin{align*}
\int |\hat{\alpha}(x)-\alpha(x)|^2 \di \P_{\sf GPS}(x) &\overset{P}{\rightarrow}0, \\
\int |\hat{\alpha}(x)-\alpha(x)|^2 \di x &\overset{P}{\rightarrow}0 .
\end{align*}
\label{thm::alpha}
\end{thm}
The proof of Theorem~\ref{thm::alpha} is given in Appedix \ref{sec::proofs}. %The requirement $h\rightarrow0, \frac{nh^6}{\log n}\rightarrow \infty$ is needed to guarantee that the second derivative of $\hat{\sf p}$ is well-behaved.%Note that this Theorem is an application of Theorem 10 in \cite{chen2016generalized} and Lemma~\ref{lem::equiv}.

%Using Lemma~\ref{lem::equiv} and
%Theorem 10 in \cite{chen2016generalized},
%one can show that 
%\begin{align*}
%\int |\hat{\alpha}(x)-\alpha(x)|^2 dP_{\sf GPS}(x) &\overset{P}{\rightarrow}0, \\
%\int |\hat{\alpha}(x)-\alpha(x)|^2 dx &\overset{P}{\rightarrow}0 
%\end{align*}
%under suitable assumptions. 

%Thus, with Lemma~\ref{lem::equiv}, we have the following consistency theorem.
%\begin{thm}
%Assume xxxx and $h\rightarrow0, \frac{\log n}{nh^2}\rightarrow 0$.
%Then
%\begin{align*}
%\int |\hat{\alpha}(x)-\alpha(x)|^2 dP_{\sf GPS}(x) &\overset{P}{\rightarrow}0, \\
%\int |\hat{\alpha}(x)-\alpha(x)|^2 dx &\overset{P}{\rightarrow}0 .
%\end{align*}
%\end{thm}

%Specifically, we have
%%A good news is that \cite{chen2016generalized} proves several asymptotic theory of the estimator $\hat{\alpha}$ 
%%toward $\alpha(x)$. 
%%Specifically, 
%%for any point $x$, $\hat{\alpha}(x) \overset{P}{\rightarrow} \alpha(x)$.
%%Moreover, 
%\begin{align*}
%\int |\hat{\alpha}(x)-\alpha_H(x)|^2 dP_{\sf GPS}(x) &\overset{P}{\rightarrow}0, \\
%\int |\hat{\alpha}(x)-\alpha_H(x)|^2 dx &\overset{P}{\rightarrow}0 .
%\end{align*}
%Moreover, \cite{chen2016generalized} also showed that the KDE
%$$
%\hat{\sf p}(x) \rightarrow \infty
%$$
%when $h\rightarrow$ and $\tau(x)$
%Thus, $\hat{\alpha}(x)$ is a good estimator of $\alpha_H(x)$.
%This, together with the fact that $\alpha_H(x) = \alpha(x)$, 
%implies that $\hat{\alpha}(x)$ is a consistent estimator of $\alpha(x)$,
%which justifies the use of density ranking
%for a GPS dataset. 
%Moreover, the fact that $\alpha_H(x) = \alpha(x)$ implies that 

The collection of anchor locations $\mathcal{A}$ and roads connecting anchor points $\mathcal{R}$ are of key interest, and must be properly recovered from GPS data. Next we show that the level set of $\hat{\alpha}(x)$, under suitable choices of level, 
will be a consistent estimator of $\mathcal{A}$ and $\mathcal{R}$. Recall that $\hat{A}_\gamma = \{x: \hat{\alpha}(x)\geq 1-\gamma\}$ is the level set of density ranking.  
%With a good choice of $\gamma$, this level set will be a good estimator of
%anchor locations and roads. 
%Specifically, we define
%\begin{align*}
%\mathbb{K}_0 &= \{x: \omega(x) = 0\} = \\
%\mathbb{K}_1 &= \{x: \omega(x) = 0 \mbox{ or }1\}
%\end{align*}
%be the collection of anchor locations and roads. 
%Then the set $\hat{A}_{\pi_0}$ and $\hat{A}_{\pi_0+\pi_1}$ will be 
%consistent estimators of $\mathcal{A}$ and $\mathcal{R}\cup\mathcal{A}$, respectively. 
\begin{thm}
Given assumptions (K1-2), (P) and (S0) from Appendix \ref{sec::assumption}, and $ \frac{ nh^2}{\log n}\rightarrow \infty, h\rightarrow 0$, we have
\begin{align*}
\P_{\sf GPS}\left(\hat{A}_{\pi_0}\triangle \mathcal{A}\right) &\overset{P}{\rightarrow } 0,
\end{align*}
where for sets $A$ and $B$, $A\triangle B = (A\backslash B) \cup (B\backslash A)$ is their set difference
and $\P_{\sf GPS}(A) = P(X_1\in A)$, where $X_1$ is has distribution $\P_{\sf GPS}$. Moreover, if we further assume (S1) from Appendix \ref{sec::assumption}, we have
\begin{align*}
\P_{\sf GPS}\left(\hat{A}_{\pi_0+\pi_1}\triangle (\mathcal{A}\cup\mathcal{R})\right) &\overset{P}{\rightarrow } 0.
\end{align*}
\label{thm::conv_set}
\end{thm}
The proof of Theorem \ref{thm::conv_set} is given in Appendix \ref{sec::proofs}. The set difference $\triangle$ is a conventional measure of the difference between two sets. Applying the probability $\P_{\sf GPS}$ to the set difference is a common measure of the convergence of a set estimator \citep{mason2009asymptotic, rigollet2009optimal,qiao2017asymptotics,doss2018bandwidth}. Theorem~\ref{thm::conv_set} shows that $\hat{A}_{\pi_0}$ and $\hat{A}_{\pi_0+\pi_1}$ are consistent estimators of $\mathcal{A}$ and $\mathcal{A}\cup\mathcal{R}$, respectively. With this fact, we can use the difference $\hat{A}_{\pi_0+\pi_1}\backslash\hat{A}_{\pi_0}$ as an estimator of $\mathcal{R}$.  Namely, $\hat{A}_{\pi_0}$ can be used to recover the anchor locations and $\hat{A}_{\pi_0+\pi_1}\backslash\hat{A}_{\pi_0}$ can be used to reconstruct the sections of the roads covered by an individual's activity space. 

Finally, we show that under the mixture model in Eq. \eqref{eq::GPSfn}, the KDE $\hat{\sf p}(x) $ diverges with a probability tending to $1$
at any anchor location or any point on a road connecting two anchor points.

%In addition, at a point where $\omega(x) = 0$ or $1$ (an anchor location or a road), 
%%we 
%%one can see that when $h\rightarrow 0$,
%the KDE $\hat{\sf p}(x) $ will diverge with a probability tending to $1$.
\begin{thm}
Under assumptions (K1--2) from Appendix \ref{sec::assumption} and $h\rightarrow0$, we have $\E(\hat{\sf p}(x))\rightarrow \infty$
for any $x\in \mathcal{A}\cup \mathcal{R}$. 
\label{thm::kde}
\end{thm}
%Namely, 
%for any fixed number $M<\infty$, 
%$$
%P(\hat{\sf p}(x) >M) \rightarrow 1
%$$
%when $n\rightarrow \infty , h \rightarrow0$.
The proof of Theorem~\ref{thm::kde} is given in Appendix \ref{sec::proofs}. This result shows why the KDE does not give a stable estimator of $\mathcal{A}$ and $\mathcal{R}$ which explains why for the activity space in Figure~\ref{fig::alpha}, the KDE does not properly detect its structure. 

\section{Simulation study} \label{sec::sim}
 
We consider an example individual whose activity space has three anchor locations $\mathcal{A}$: home (located at $(0,0)$), office (located at $(0,2)$), and gym (located at $(2,0)$) \--- see Figure \ref{fig::simExam}. We assume that the anchor locations $\mathcal{A}$ are connected by three straight segments of road $\mathcal{R}$. The individual spends $60\%$ of their time in the anchor locations, and $30\%$ of their time traveling on the roads. In the rest of their time, this individual walks in the neighborhoods around their office and home, but never walks in the vicinity of their gym. When the individual is at an anchor location, they spend $50\%$ of their time at home, $30\%$ of their time at work, and $20\%$ of their time at the gym. For this individual's activity space, the mixture model in Eq. \eqref{eq::GPSfn} is written as ($\pi_0=0.6$, $\pi_1=0.3$, $\pi_2=0.1$): 
\begin{equation}\label{eq::GPSfnExam}
 \P_{\sf GPS}(x) = 0.6 \P_0(x) + 0.3 \P_1(x) + 0.1 \P_2(x),
\end{equation}
with
$$
 \P_0(x)  = 0.5 \delta_{(0,0)}(x) + 0.3\delta_{(0,2)} +0.2\delta_{(2,0)}.
$$
Here $\delta_{(a,b)}(x)$ is a function that puts a point mass at $(a,b)$. The time in which the individual travels between the anchor locations is divided as follows: $30\%$ on the road between home and gym, $20\%$ on the road between gym and office, and $50\%$ on the road between home and office. We assume that the individual travels with the same speed on all road segments. For $70\%$ of their total walk time the individual moves  uniformly within the square $[-0.5,0.5]\times [-0.5,0.5]$ centered at their home, and for remaining $30\%$ the individual moves uniformly within the square $[1.6,2.4]\times[-0.4,0.4]$ centered at their office. With these assumptions, the distributions $\P_1$ and $\P_2$ in Eq. \eqref{eq::GPSfnExam} are completely specified. 

\begin{figure}[!ht]
\center
\includegraphics[width=3in]{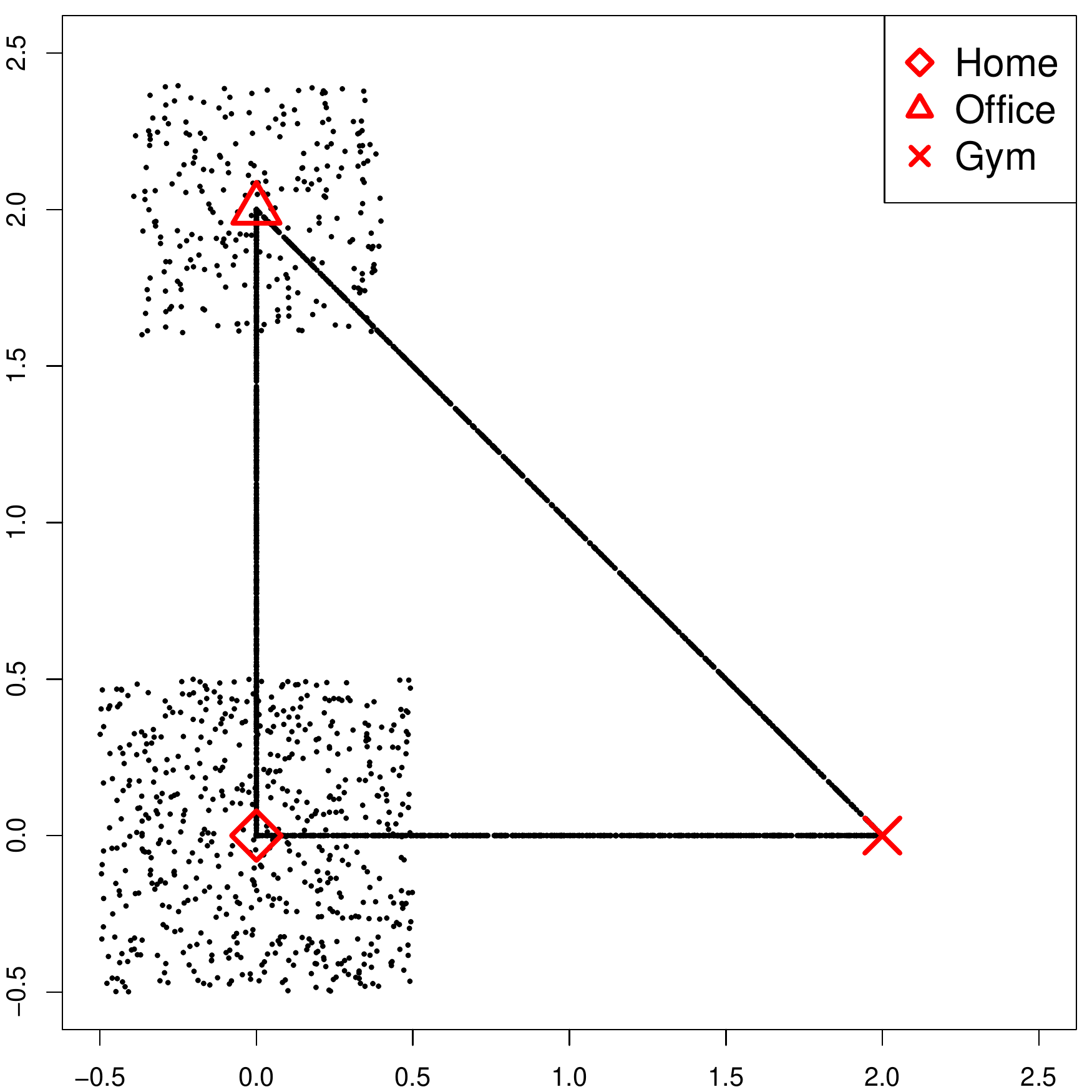}
\caption{Activity space for the simulation study, and scatter plot of $n=8,000$ locations sampled from the mixture model \eqref{eq::GPSfnExam}. The anchor locations are shown as follows: home (red diamond), office (red triangle), and gym (red cross).}
\label{fig::simExam}
\end{figure}

We generate $n=8,000$ samples from the mixture model \eqref{eq::GPSfnExam} \--- see Figure \ref{fig::simExam}. We use a smoothing bandwidth of  $0.5$ to compute the density ranking. The corresponding contours (top left panel of Figure \ref{fig::sim}) show a very good agreement with the anchor locations and the road segments (top right panel of Figure \ref{fig::sim}). We determine two level sets of density ranking (see the bottom panels of Figure \ref{fig::sim}): $\hat{A}_{0.6}$ ($\pi_0=0.6$) and $\hat{A}_{0.9}$ ($\pi_0+\pi_1=0.9$) corresponding with the mixture weights in Eq. \eqref{eq::GPSfnExam}. We see that $\hat{A}_{0.6}$ recovers all three anchor locations $\mathcal{A}$, while $\hat{A}_{0.9}$ recovers the anchor locations and the road segments $\mathcal{A}\cup\mathcal{R}$. This is consistent with our theoretical results, in particular, with Theorem~\ref{thm::conv_set}.

\begin{figure}[!ht]
\center
\includegraphics[width=1.6in]{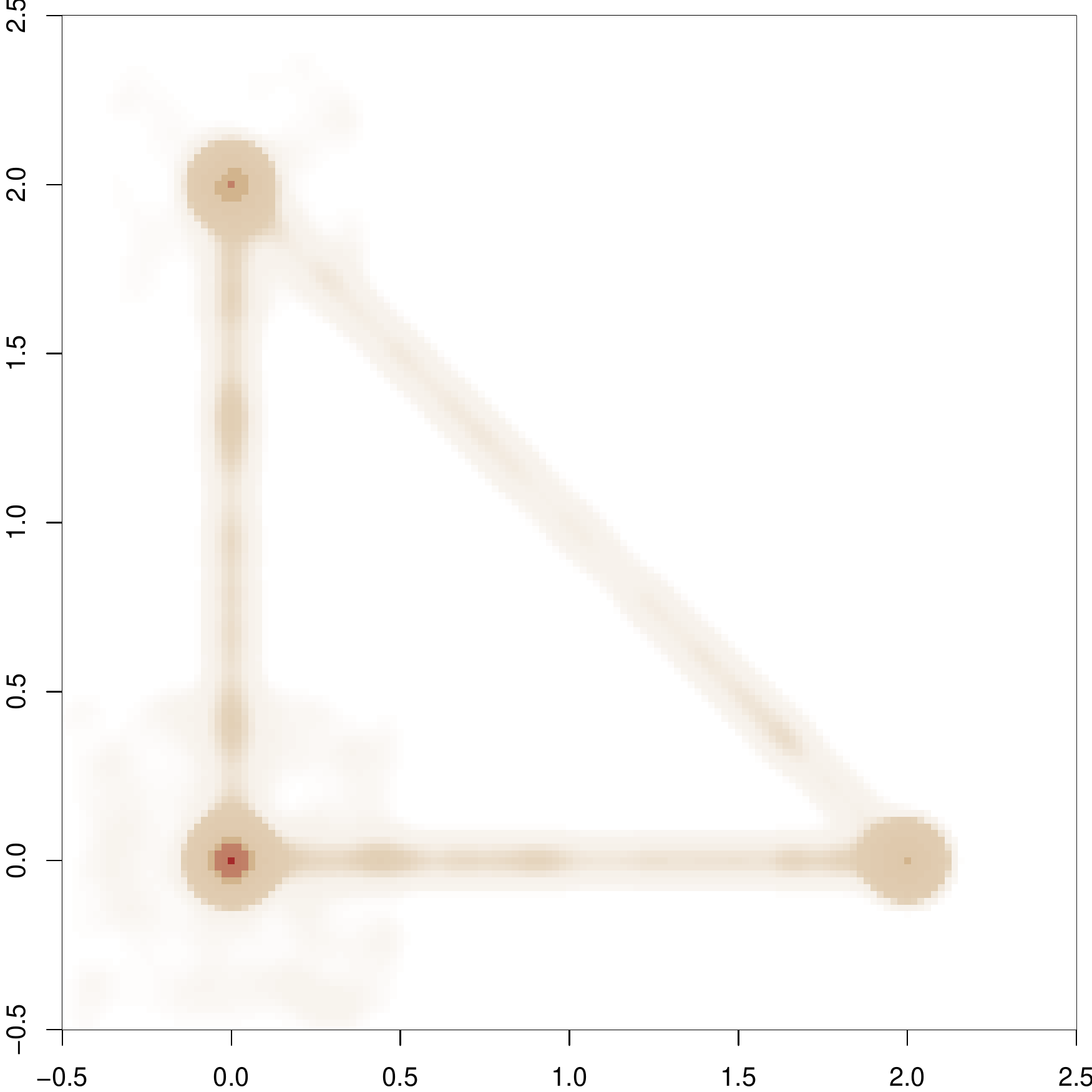}
\includegraphics[width=1.6in]{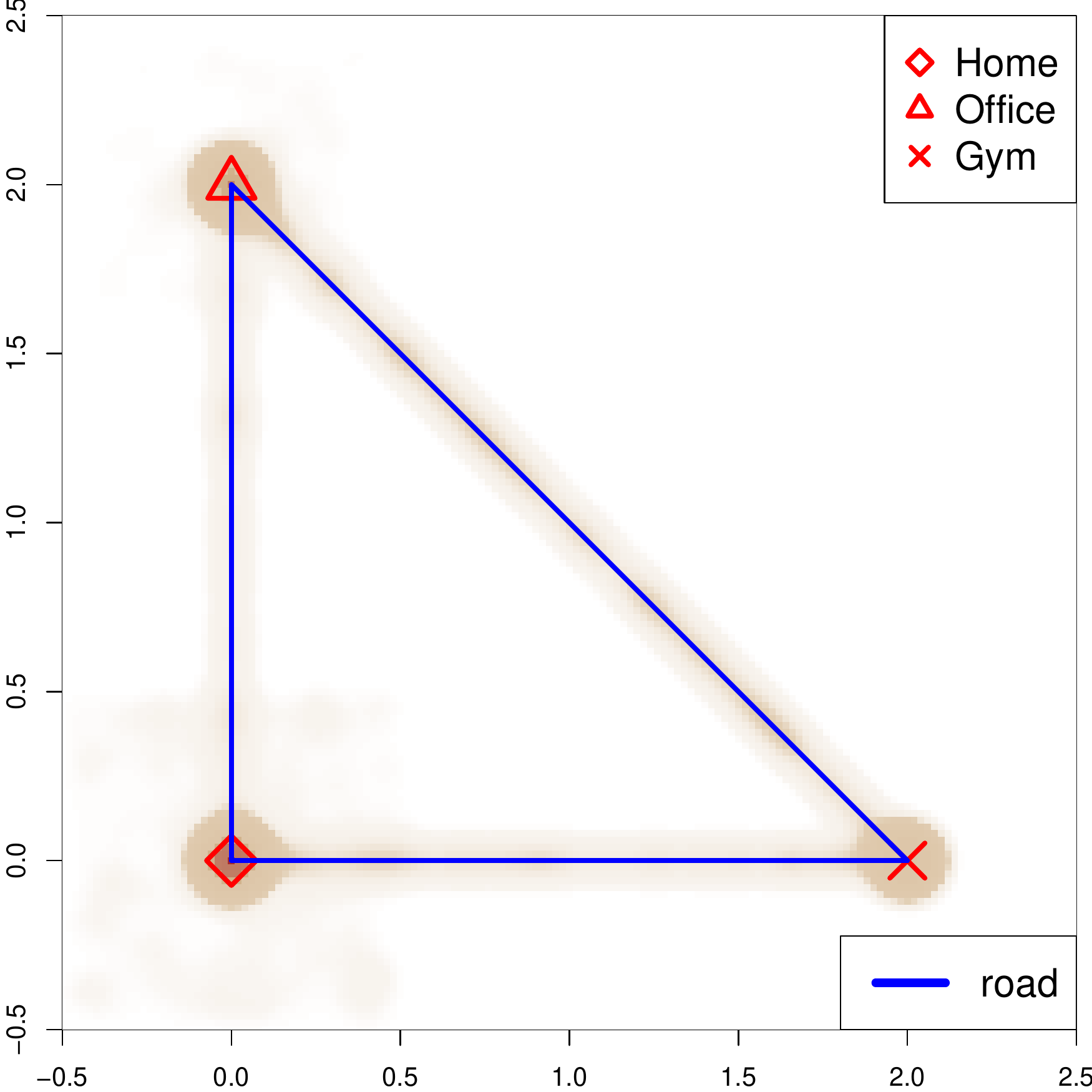}\\
\includegraphics[width=1.6in]{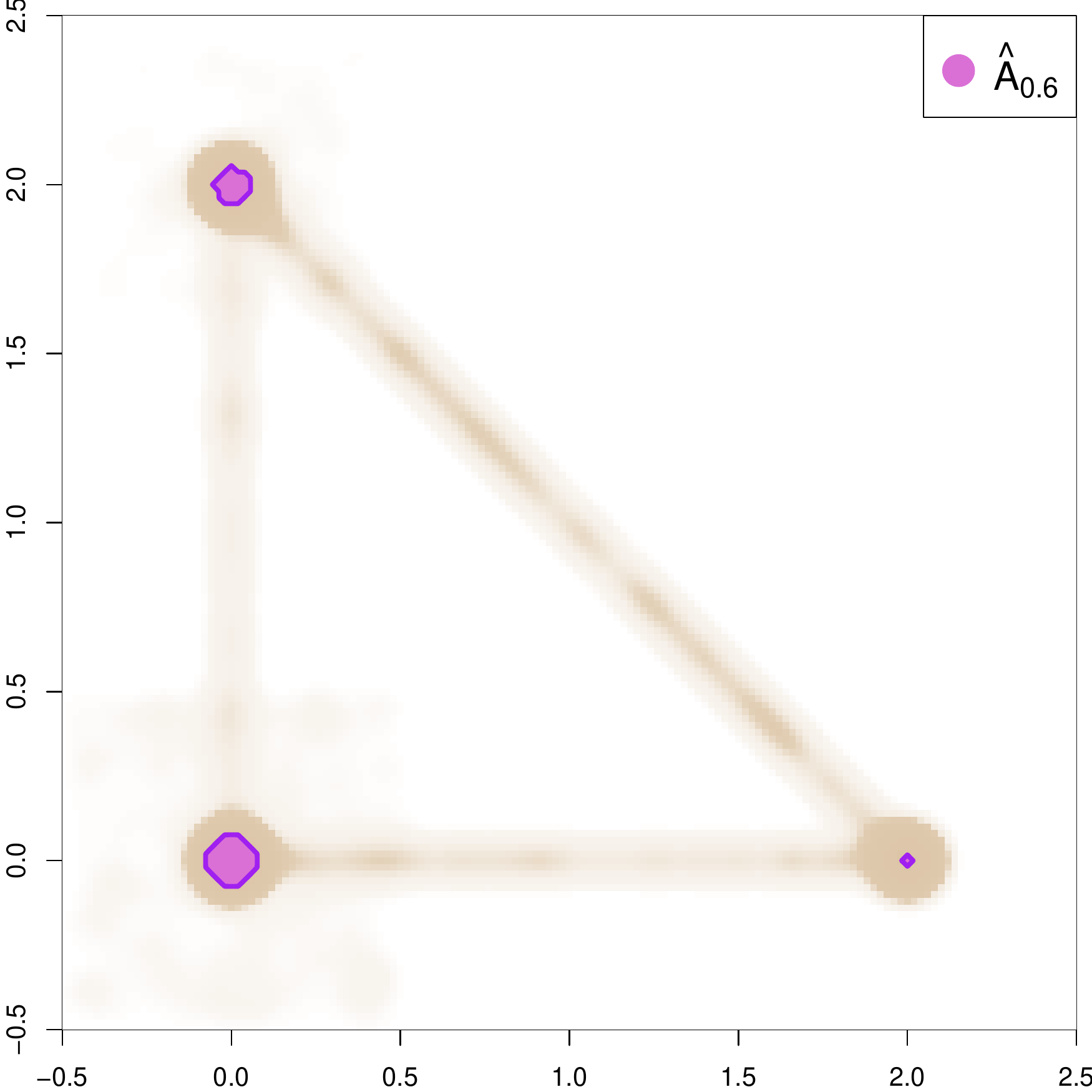}
\includegraphics[width=1.6in]{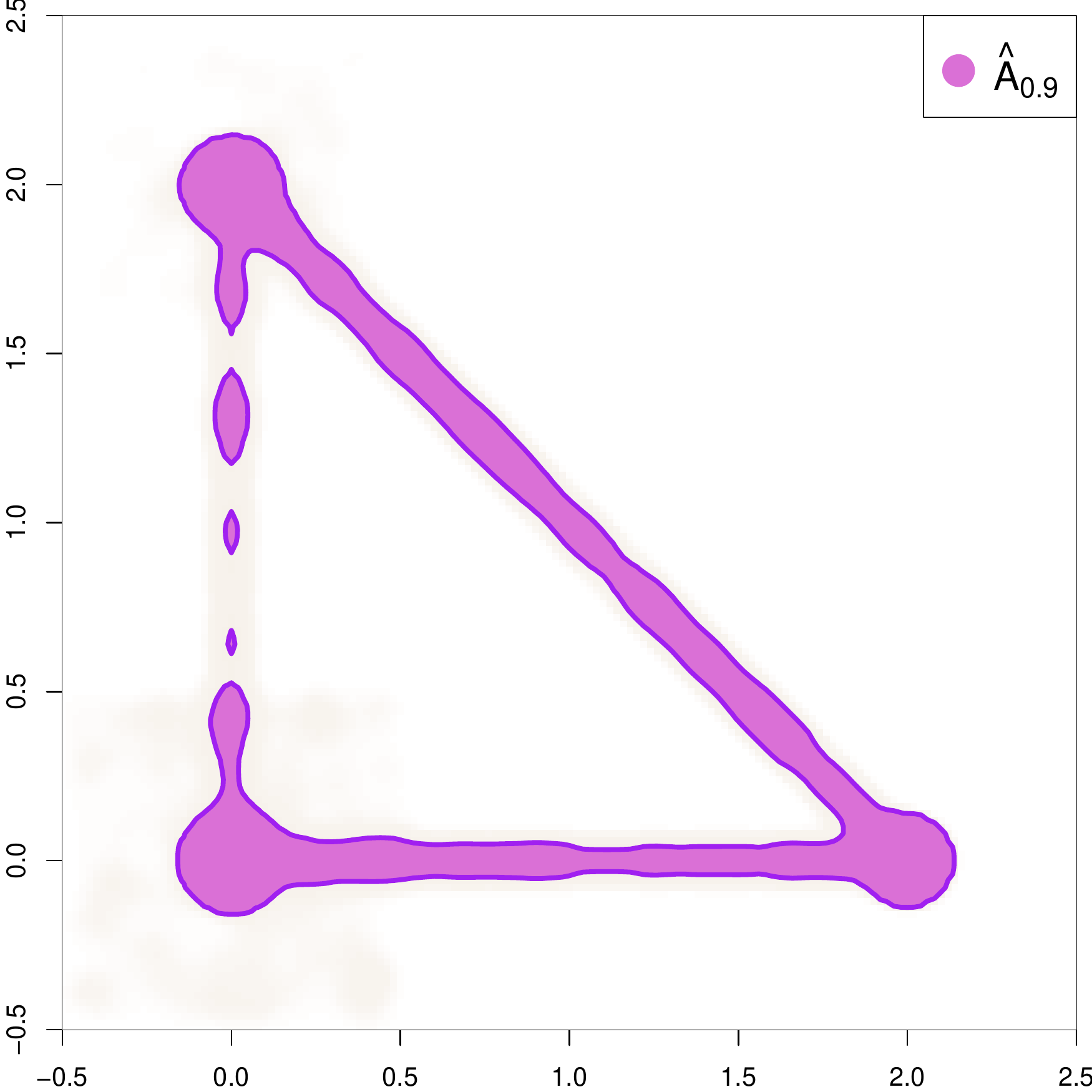}
\caption{Analysis of simulated data. Top left panel: contours of density ranking. Top right panel: the anchor locations and the road segments superimposed on the density ranking contours. Bottom left panel: the level set $\hat{A}_{0.6}$ of density ranking. Bottom right panel: the level set $\hat{A}_{0.9}$ of density ranking.}
\label{fig::sim}
\end{figure}

%We generates $n=8,000$ GPS records from this model. The top left panel of Figure~\ref{fig::sim} provides the scatter plot of the data.  We use a smoothing bandwidth $0.5$ to compute the density ranking and display the contours of density ranking in the top right panel. In the bottom left panel, we attach roads to the density ranking contour, which shows a strong agreement. To demonstrate that density ranking can recover the anchor locations, we consider the level set $\hat{A}_\gamma$ with $\gamma=\pi_0=0.6$ and $\gamma=\pi_0+\pi_1=0.9$, in the bottom middle and right panels, respectively. By comparing $\hat{A}_{\pi_0}$ (bottom middle panel) to the anchor locations, we see that $\hat{A}_{\pi_0}$ detects all anchor locations and the set is quiet close to the three anchor locations, which is consistent with Theorem~\ref{thm::conv_set}. Moreover, $\hat{A}_{\pi_0+\pi_1} $ (bottom right panel) also recovers most of the roads, which again is consistent with Theorem~\ref{thm::conv_set}.

We compare the relative performance of kernel density estimation and density ranking by simulating $n=8,000$ samples from the mixture model \eqref{eq::GPSfnExam} 100 times. For density ranking, we determine the level sets $\{\hat{A}_\gamma:\gamma=0.05,0.10,\ldots, 0.95\}$. For kernel density estimation, we determine the level sets $\{\hat{A}_{\gamma\cdot \max_x\hat{\sf p}(x)}:\gamma=0.05,0.10,\ldots, 0.95\}$. For each simulation experiment and each of level set $A$, we calculate the distances $\P_{\sf GPS}(A\triangle \mathcal{A})$ and $\P_{\sf GPS}(A\triangle(\mathcal{A}\cup\mathcal{R}))$. These distances represent the error of estimating the anchor locations $\mathcal{A}$ and the combined anchor locations and road segments $\mathcal{A}\cup\mathcal{R}$ with the level set $A$. The average estimation errors are displayed in Figure \ref{fig::sim2}. The standard errors of the curves in Figure \ref{fig::sim2}
are extremely small ($0.003-0.005$), and have been omitted.

\begin{figure}[!ht]
\center
\includegraphics[width=2.2in]{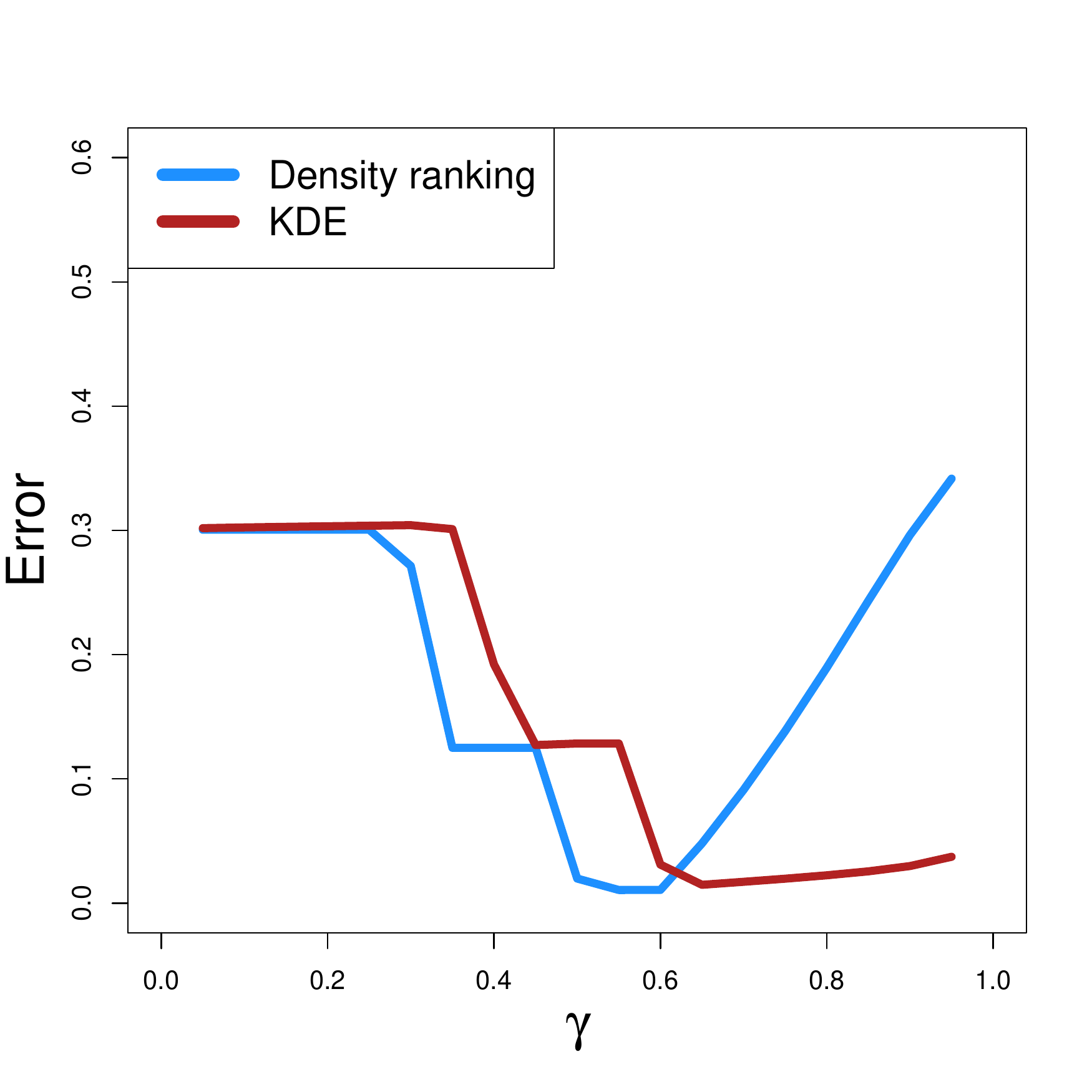}
\includegraphics[width=2.2in]{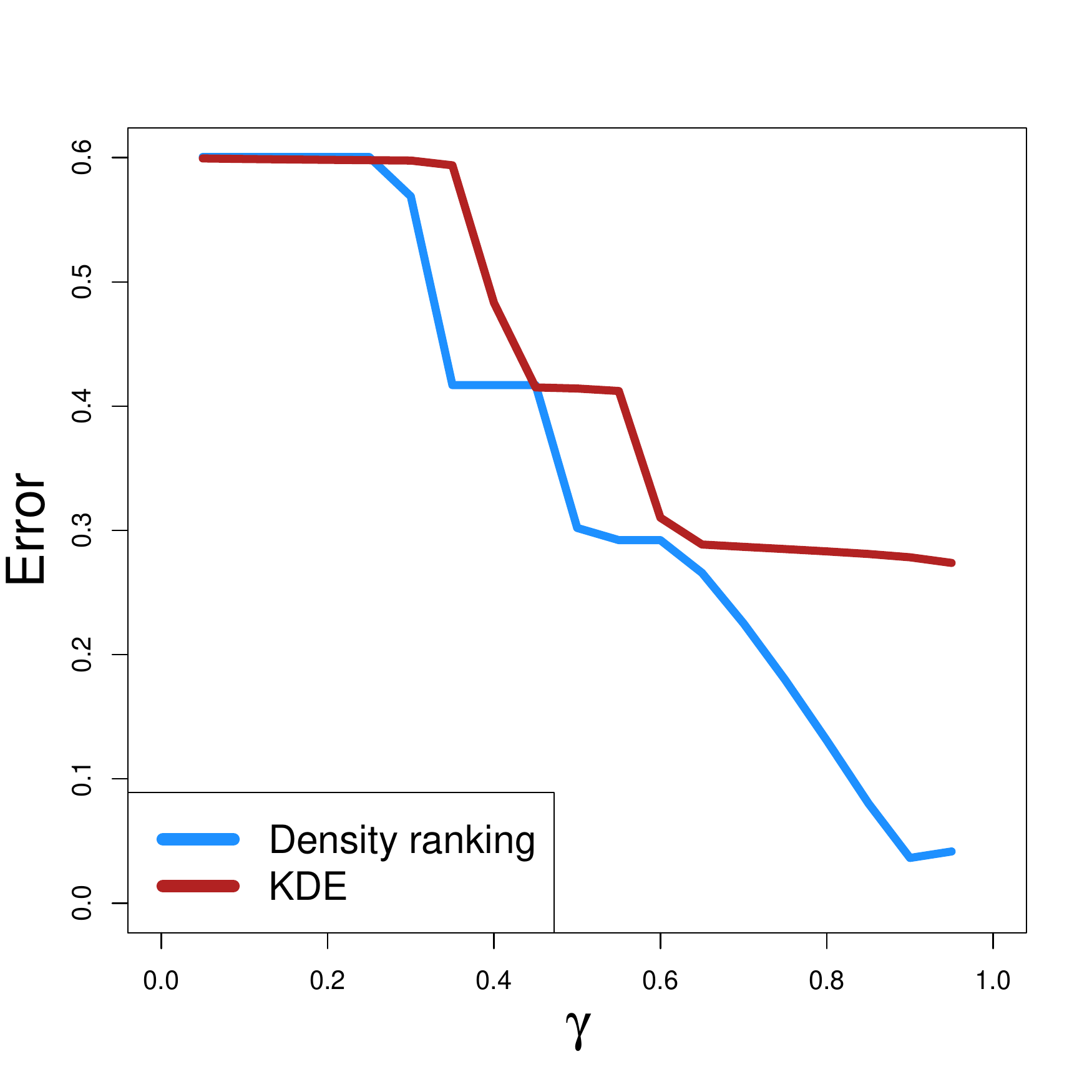}
\caption{Errors for detecting the anchor locations $\mathcal{A}$ (left panel) and the combined anchor locations and road segments $\mathcal{A}\cup\mathcal{R}$ (right panel) in the simulation study. The x-axis shows the level $\gamma$. In the left panel, the y-axis shows the error $\P_{\sf GPS}(A_{\gamma}\triangle \mathcal{A})$. In the right panel, the y-axis shows the error $\P_{\sf GPS}(A_{\gamma}\triangle(\mathcal{A}\cup\mathcal{R}))$.}
\label{fig::sim2}
\end{figure}

For the purpose of estimating the anchor locations $\mathcal{A}$, the left panel of Figure \ref{fig::sim2} shows that the level sets from both kernel density estimation and density ranking work well, although the level sets from density ranking achieve a smaller error for levels below $\pi_0=0.6$ which represents the true percentage of time spent in the anchor locations by the example individual. However, for the purpose of recovering the combined anchor locations and road segments $\mathcal{A}\cup\mathcal{R}$, the level sets from density ranking are significantly more accurate compared to the level sets from kernel density estimation. 

\begin{figure}[!ht]
\center
\includegraphics[width=1.5in]{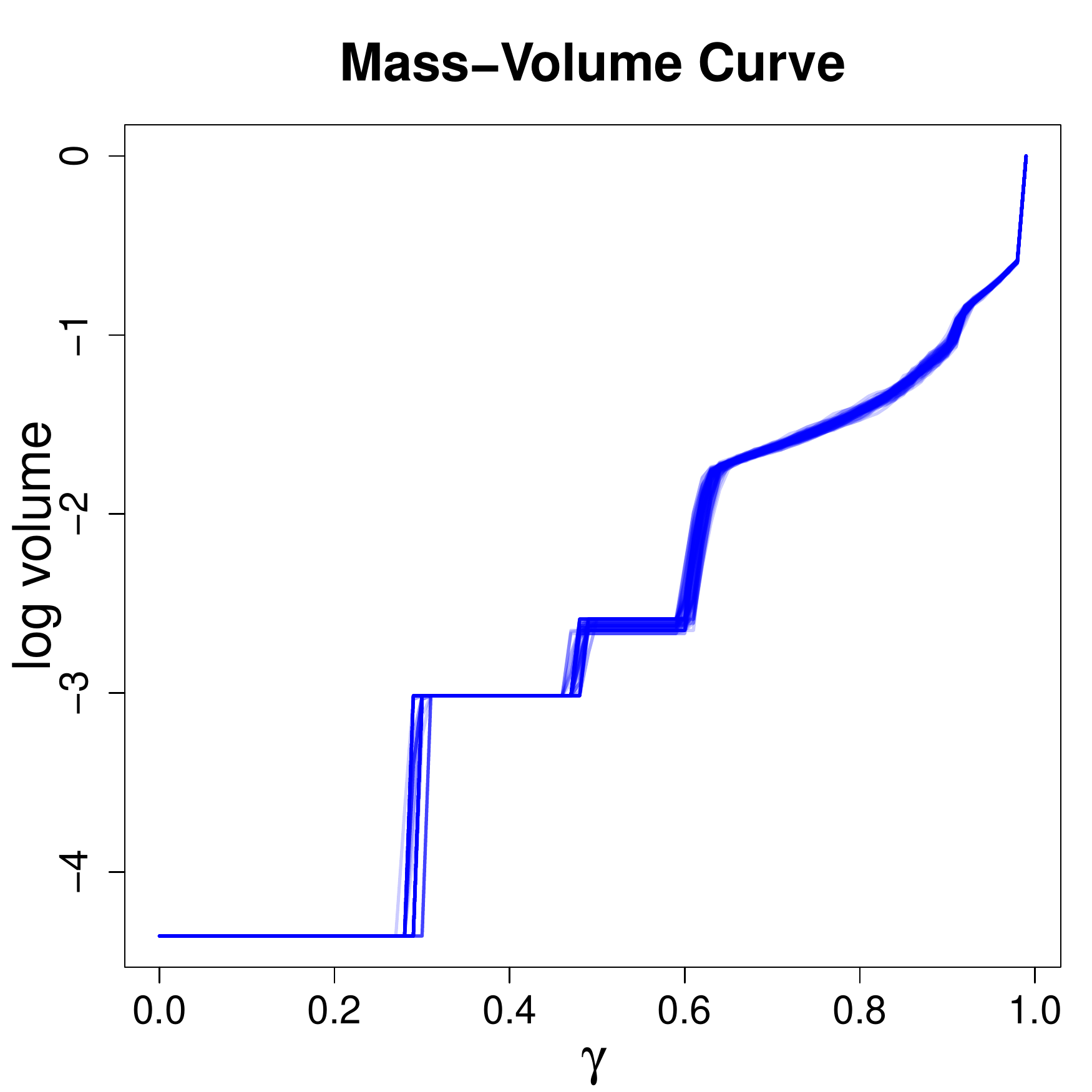}
\includegraphics[width=1.5in]{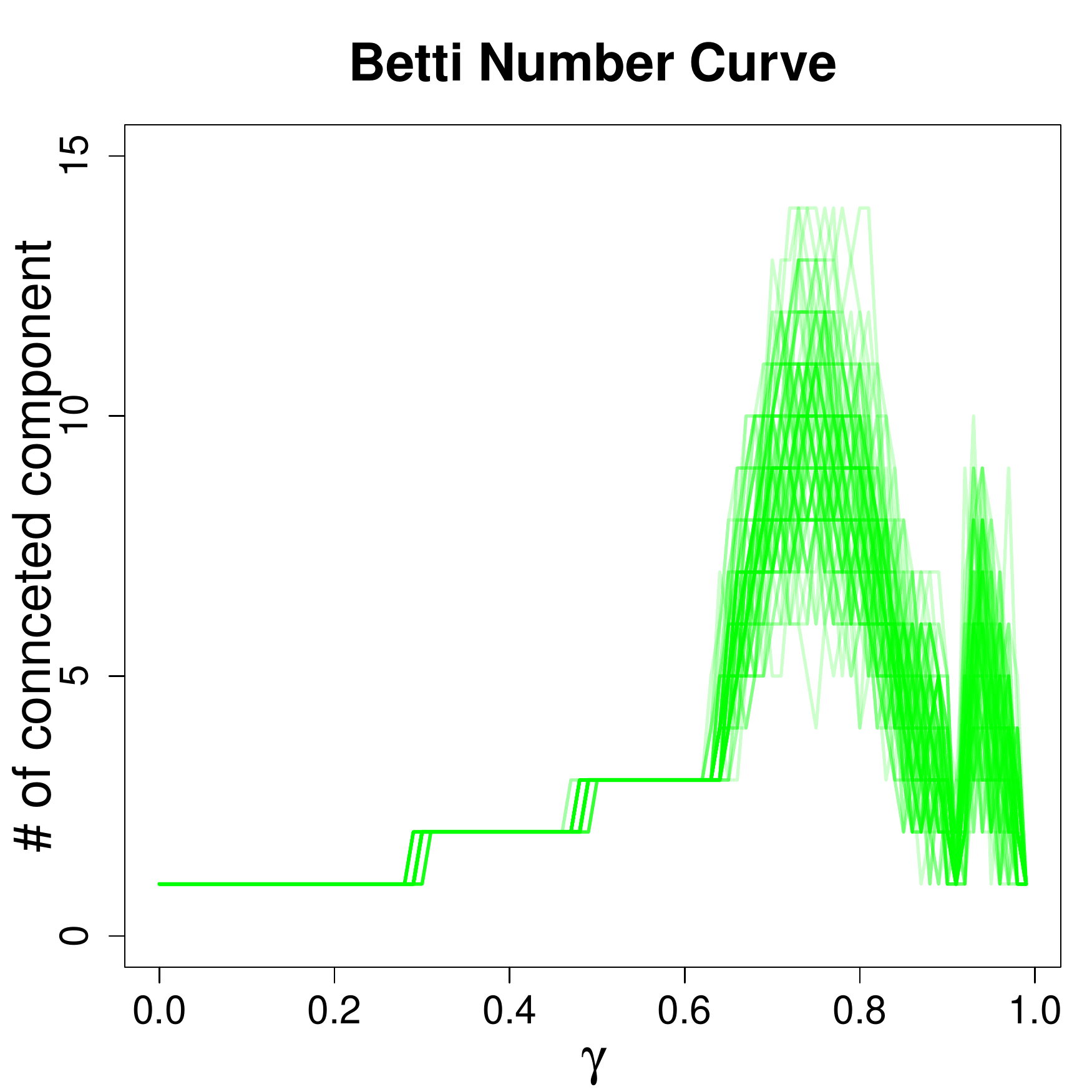}
\includegraphics[width=1.5in]{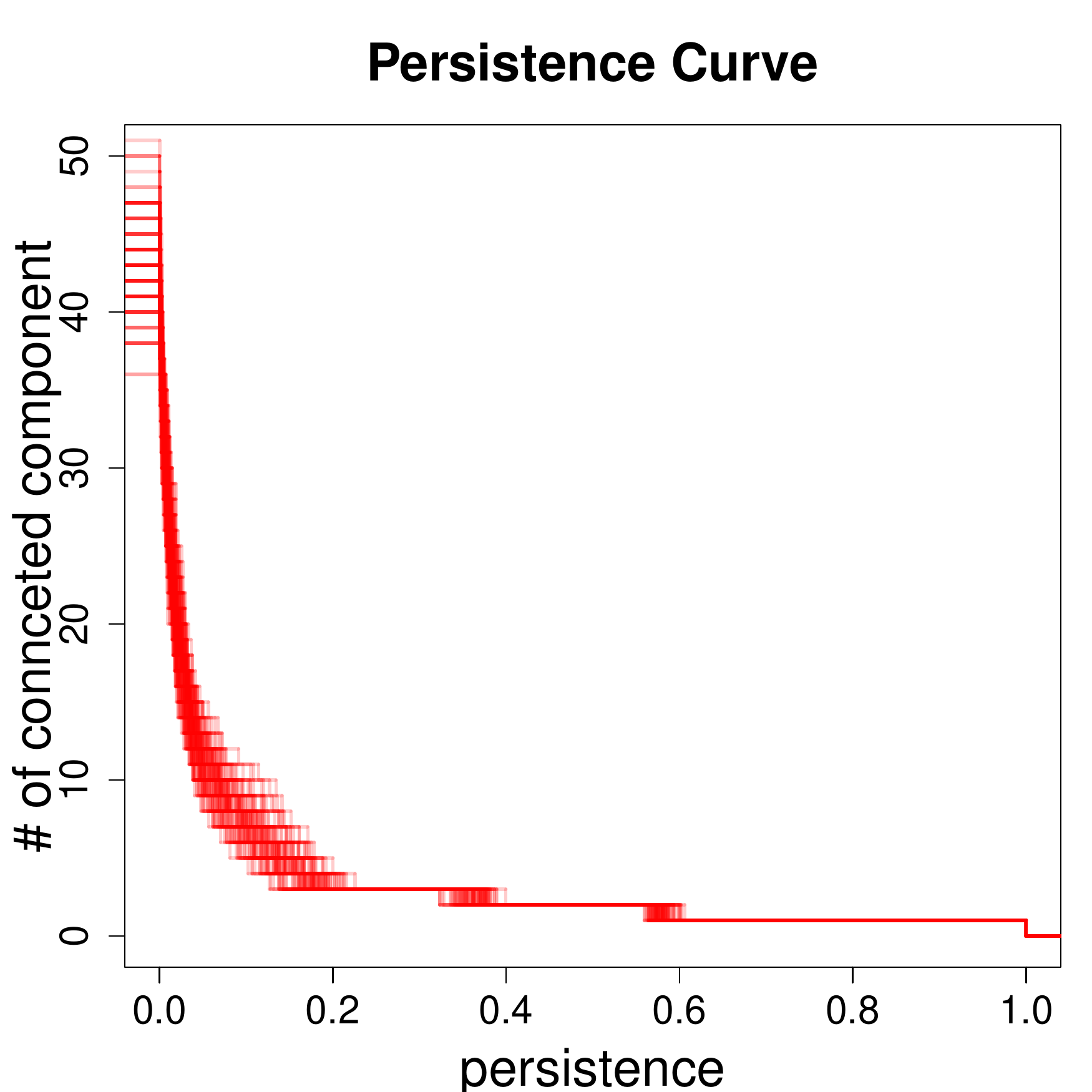}
\caption{Summary curves from the $100$ simulated datasets. The left panel shows the mass-volume curve, the middle panel shows the Betti number curve, and the right panel shows the persistence curve. The three flat regions in the left and middle panels correspond to the three anchor points $\mathcal{A}$. 
}
\label{fig::sim3}
\end{figure}

In Figure~\ref{fig::sim3}, we display the three summary curves presented in Section \ref{sec::tda}. Each panel contains $100$ curves corresponding with each simulation replicate, but many of these curves overlap. The mass-volume curve and the Betti number curve (left and middle panels) are flat around the intervals $[0, 0.3]$, $[0.3,0.5]$, and $[0.5, 0.6]$. These flat regions provide insight about the existence of anchor locations: $[0, 0.3]$ is for home, $[0.3,0.5]$ is for work, $[0.5,0.6]$ is for gym. The persistence curve (right panel) also indicates that there are three connected components with a high persistence. Each connected component is associated with an anchor point.

The summary curves allow us to choose the density ranking level to recover the anchor points. For example, the flat region at $[0.5,0.6]$ of the mass-volume curve and the Betti number curve correspond to the valley in the error curves in the left panel of  Figure~\ref{fig::sim2}. As such, using level sets $\hat{A}_\gamma$ with $\gamma\in [0.5,0.6]$ to estimate $\mathcal{A}$ yields the smallest estimation errors. Thus the summary curves are very informative about the choice of ranking thresholds to employ in the identification of anchor points.

\section{Analysis of GPS data}	\label{sec::DA}

We illustrate the application of our methodology to the GPS data from the pilot study described in Section~\ref{sec::GPS}.

\subsection{Density ranking}

We apply density ranking based on the KDE in Eq. (\ref{eq:kdequart}) with a smoothing bandwidth $h$ of $200$ meters. This choice implies that every observed GPS location will affect its neighborhood up to a distance of $200$ meters. In Appendix \ref{sec:app} we compare several smoothing bandwidths: $h=200$ seems to give an appropriate amount of smoothing for these data.

We consider GPS locations that belong to the zoom-in area shown in Figure~\ref{fig::gps1} since this area contains most locations of the 10 individuals in the pilot study. The density ranking of each individual is given in Figure~\ref{fig::17h200}.  The pattern of density ranking varies from individual to individual. Individuals 2, 4 and 6 have more widespread GPS location distributions, while individuals 1 and 5 recorded GPS locations that seem to be more clustered. There are two key locations shared by all 10 individuals: the workplace and the location of the center of the township. The density ranking of all 10 individuals is high at the locations of the workplace and the township, along the road that connects them.

\begin{figure}[!ht]
\center
\includegraphics[width=1.6in]{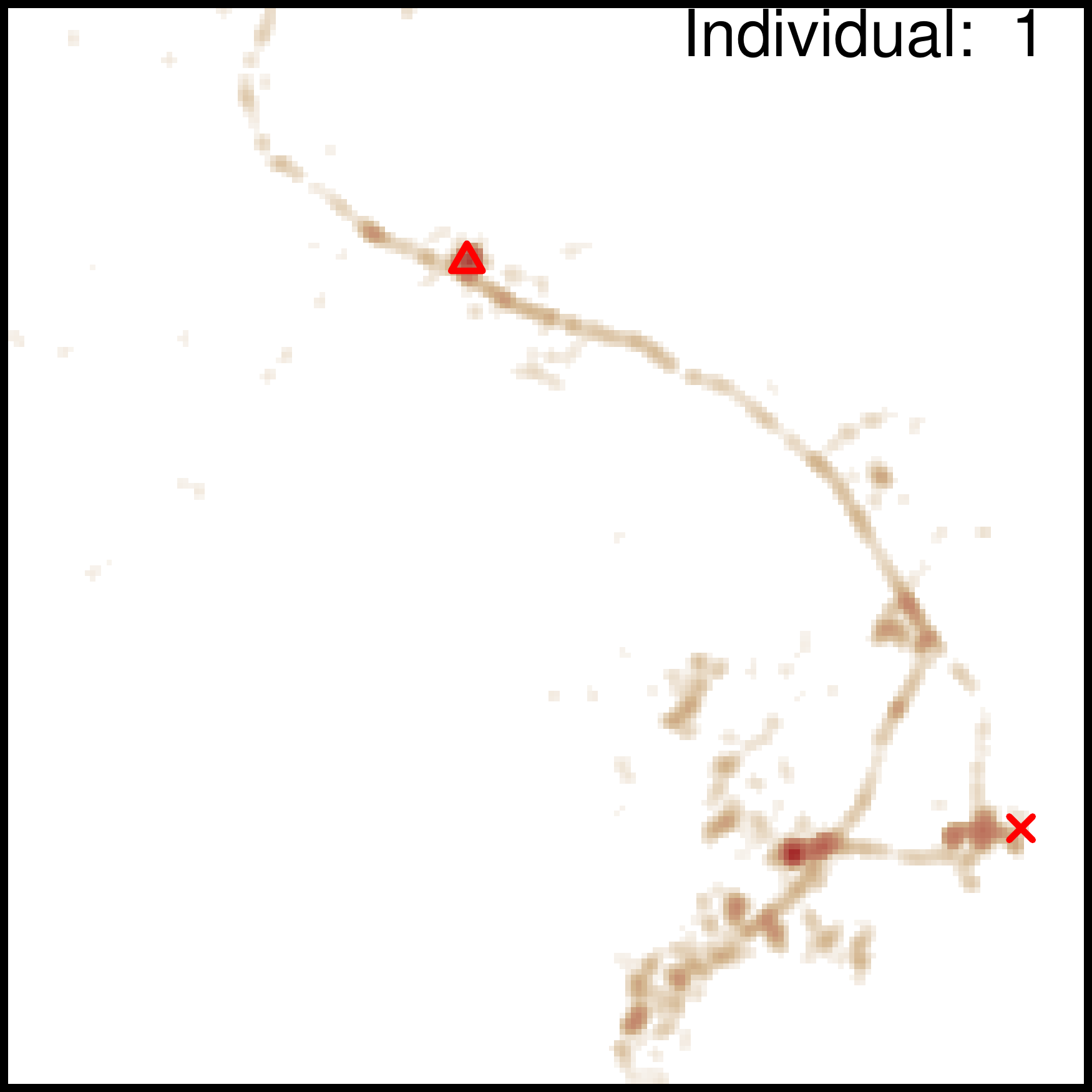}
\includegraphics[width=1.6in]{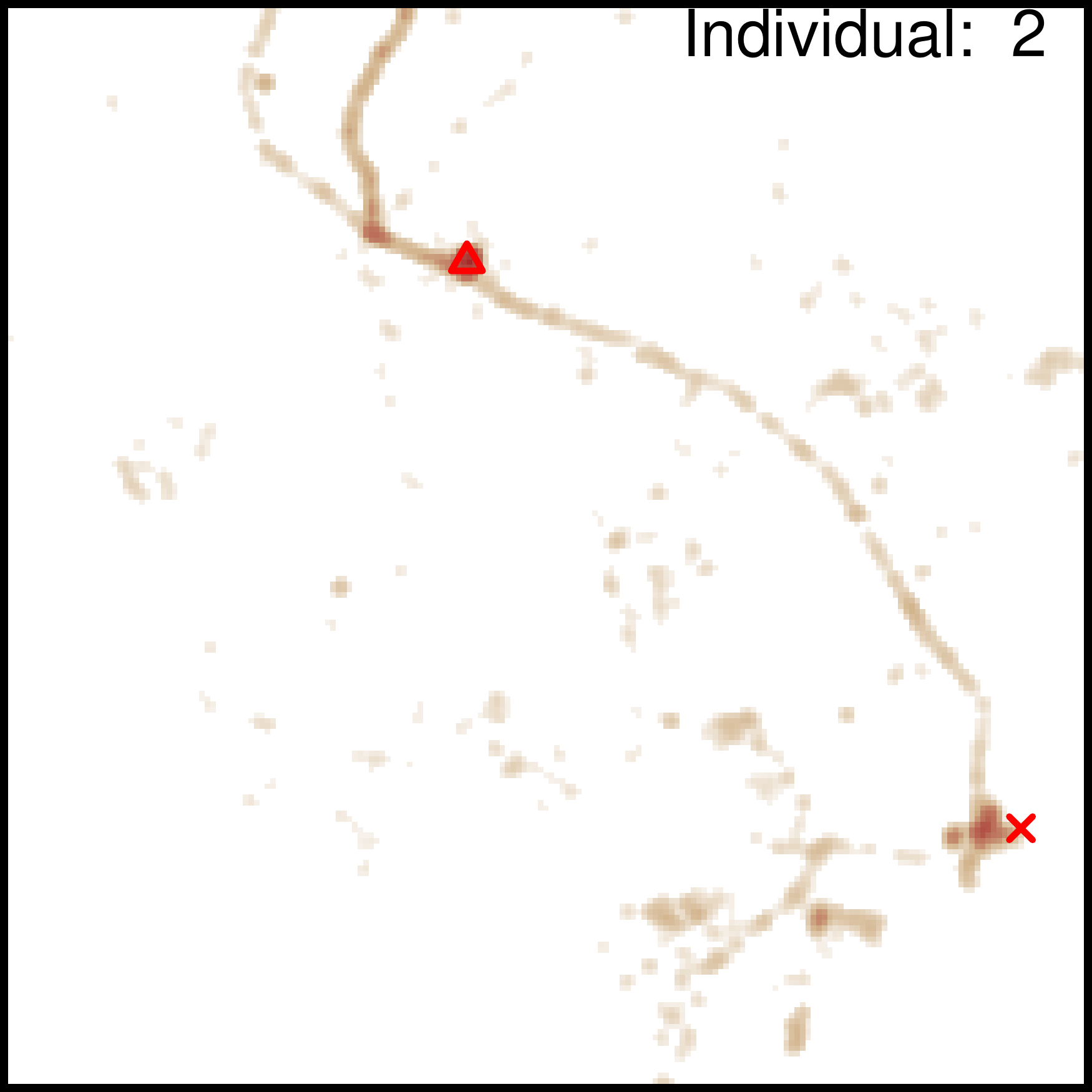}
\includegraphics[width=1.6in]{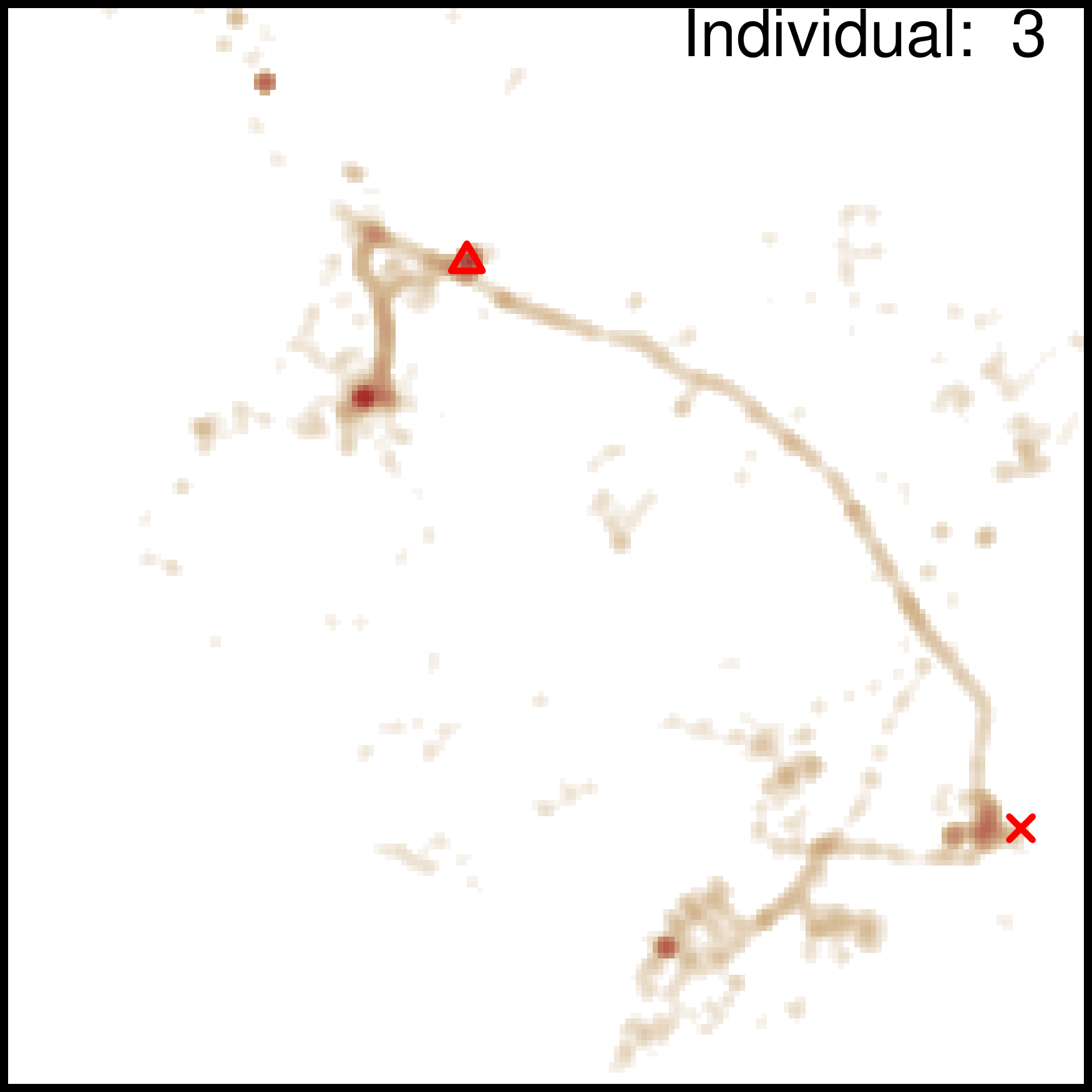}
\includegraphics[width=1.6in]{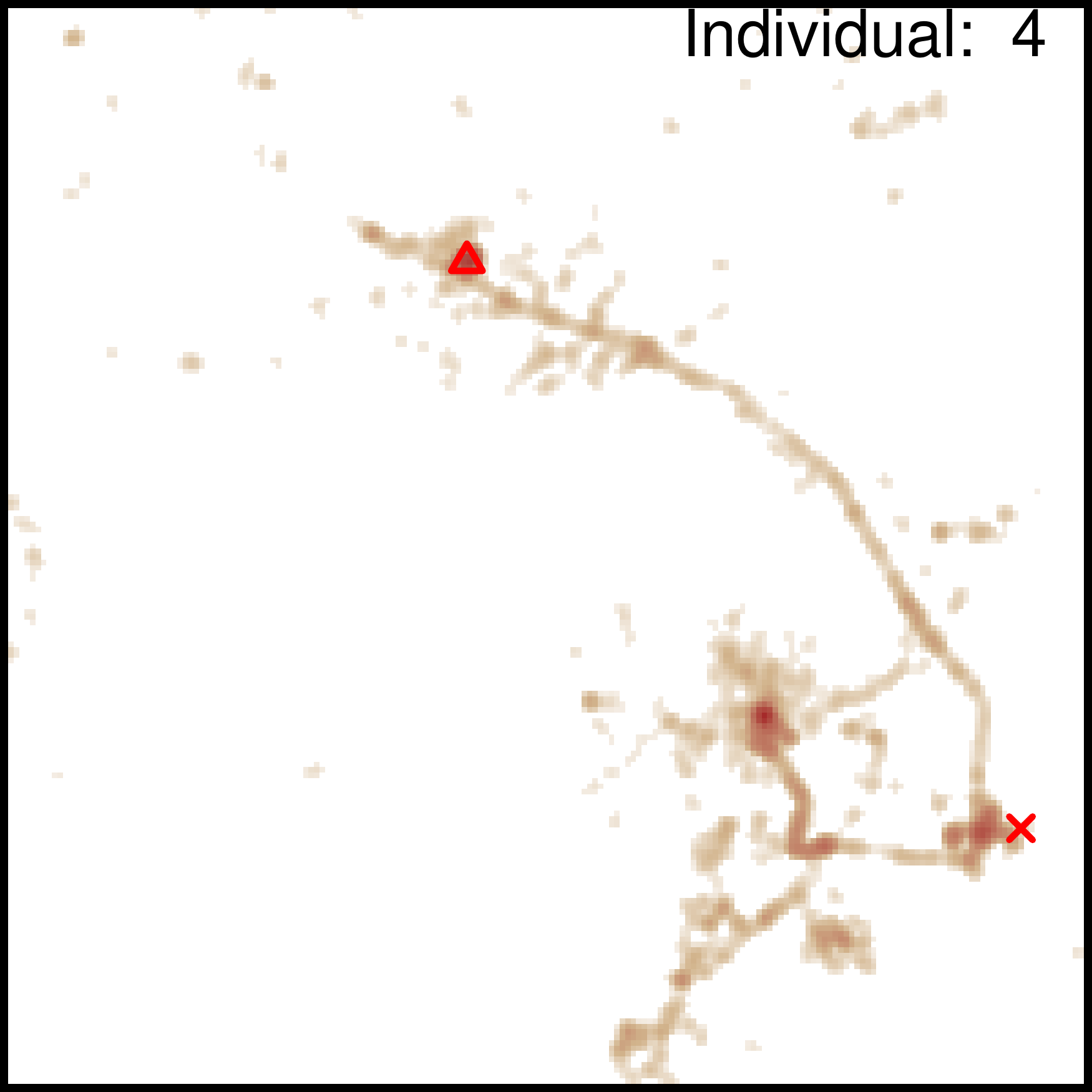}
\includegraphics[width=1.6in]{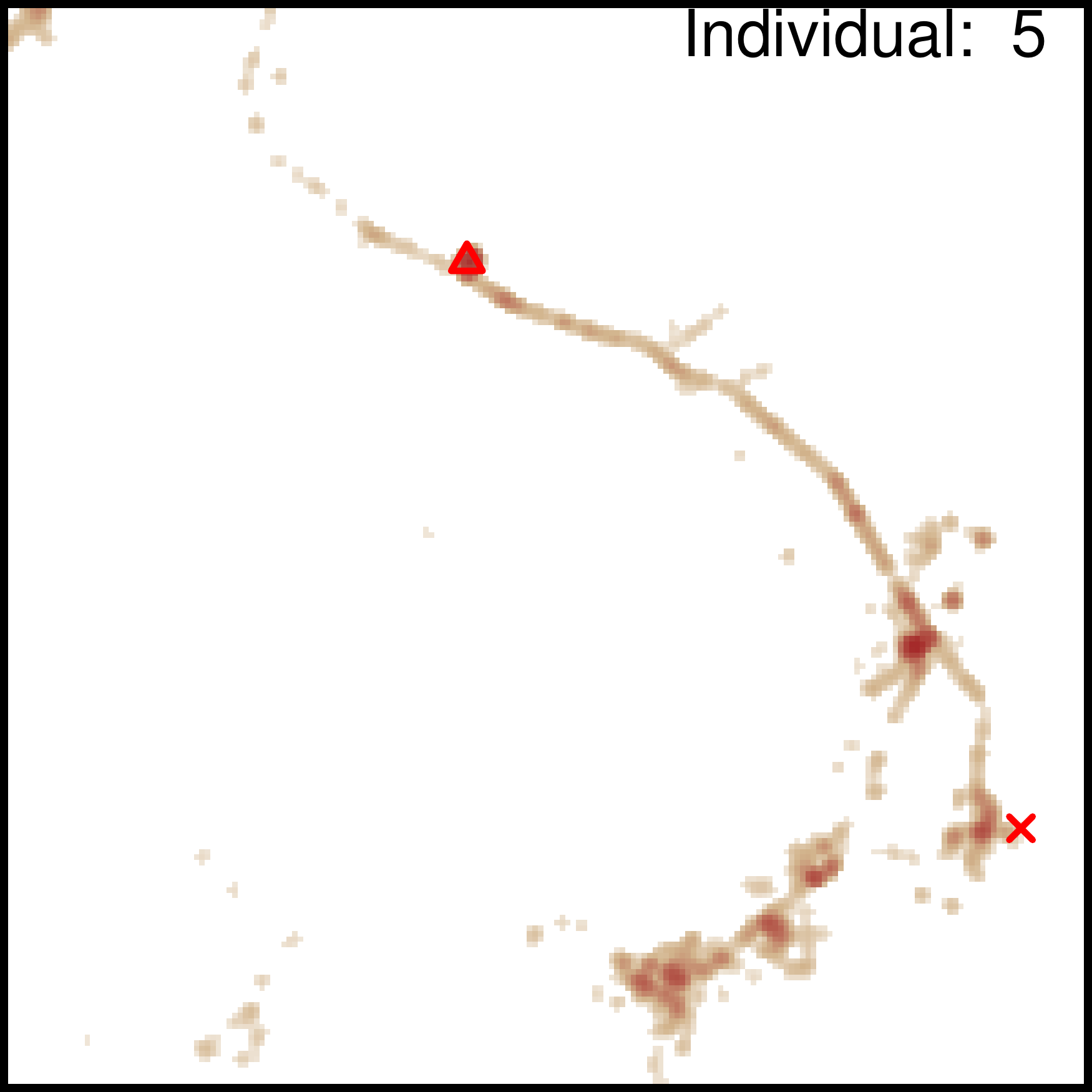}
\includegraphics[width=1.6in]{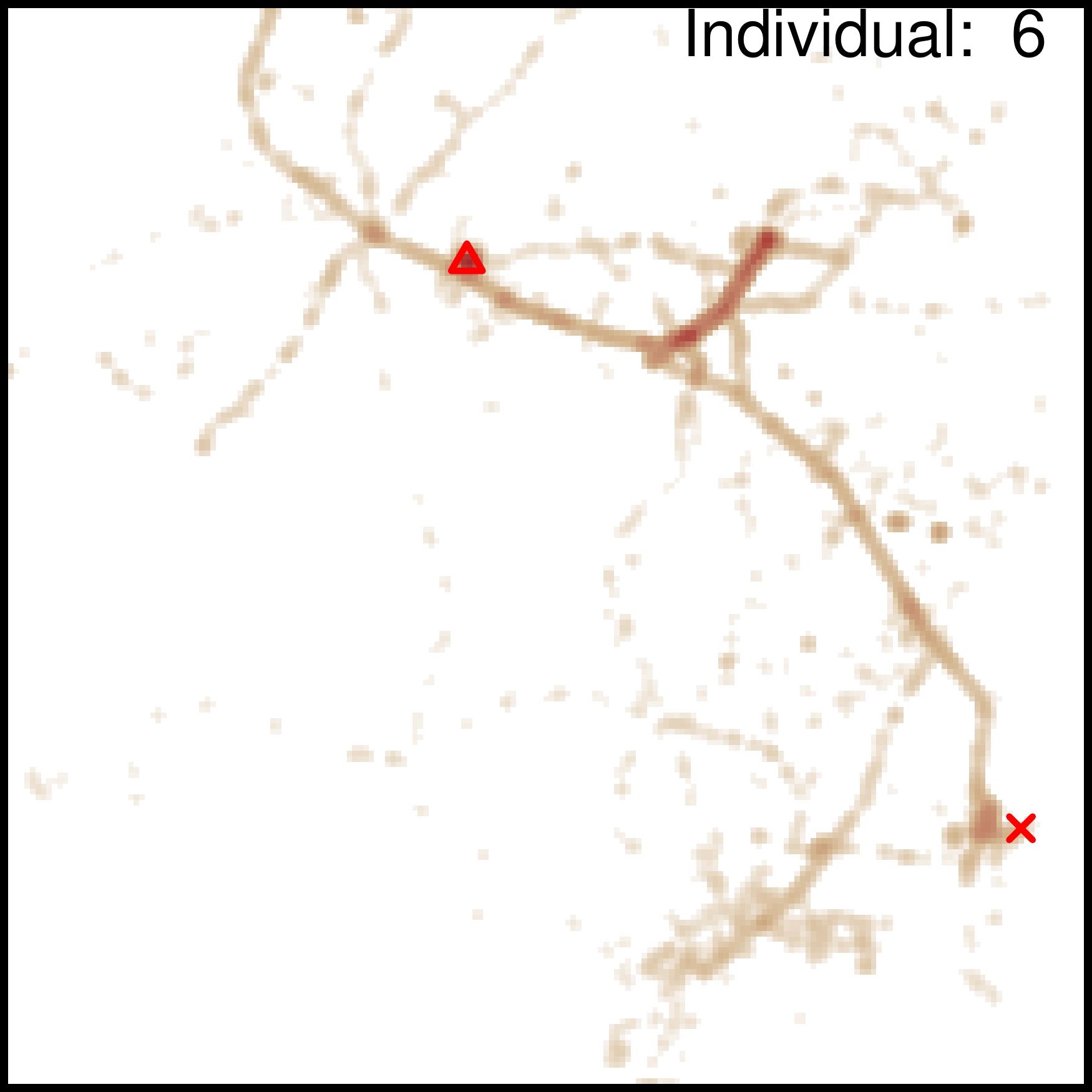}
\includegraphics[width=1.6in]{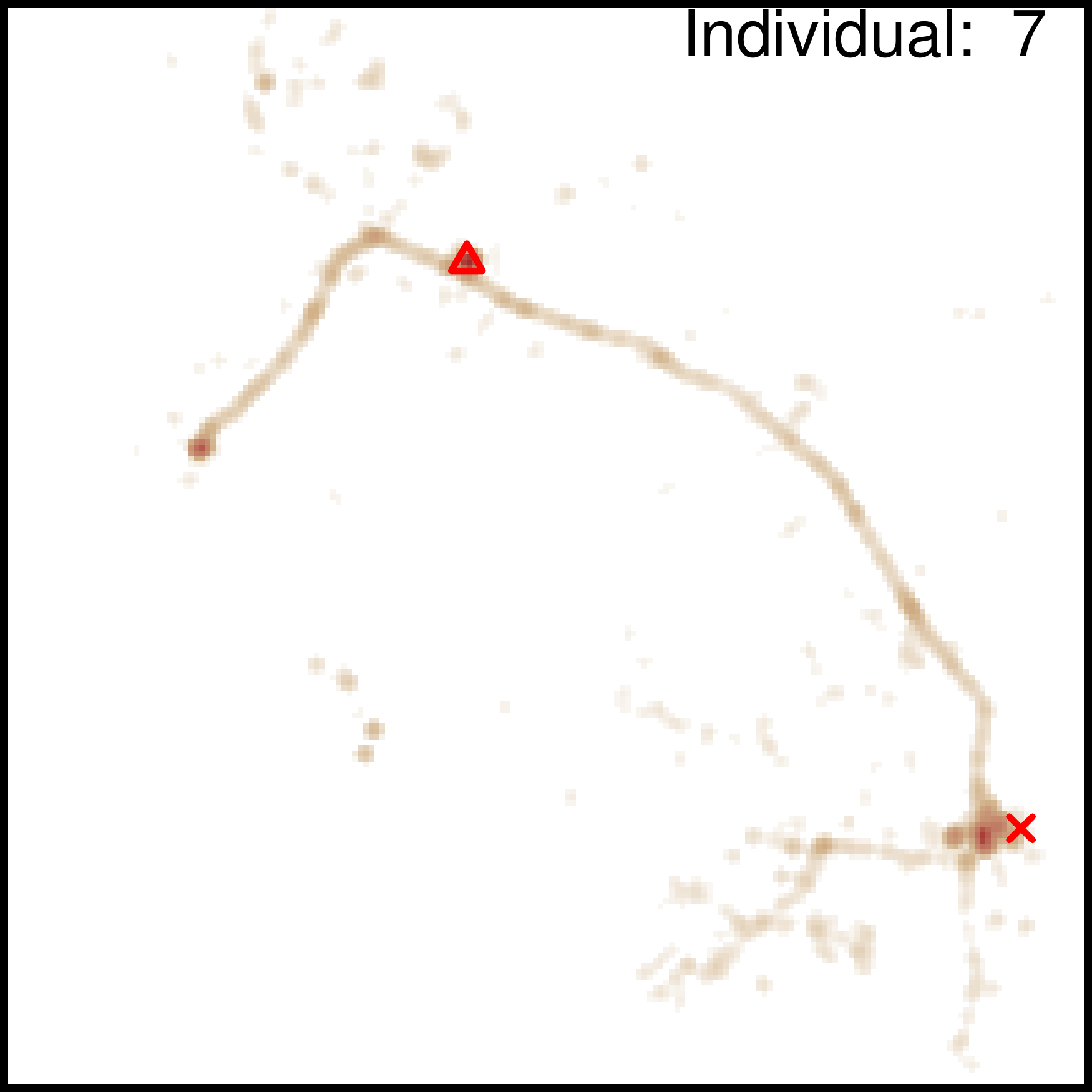}
\includegraphics[width=1.6in]{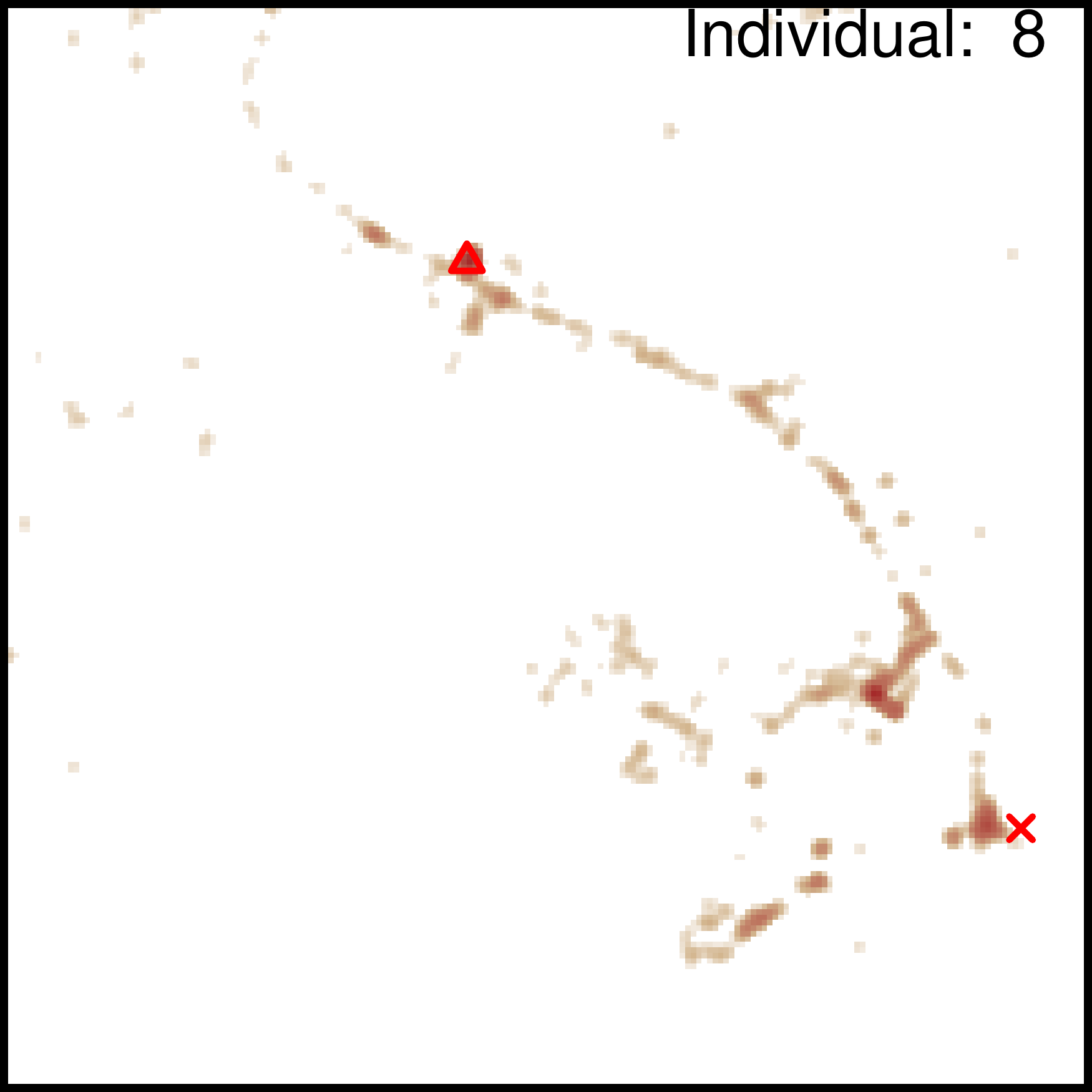}
\includegraphics[width=1.6in]{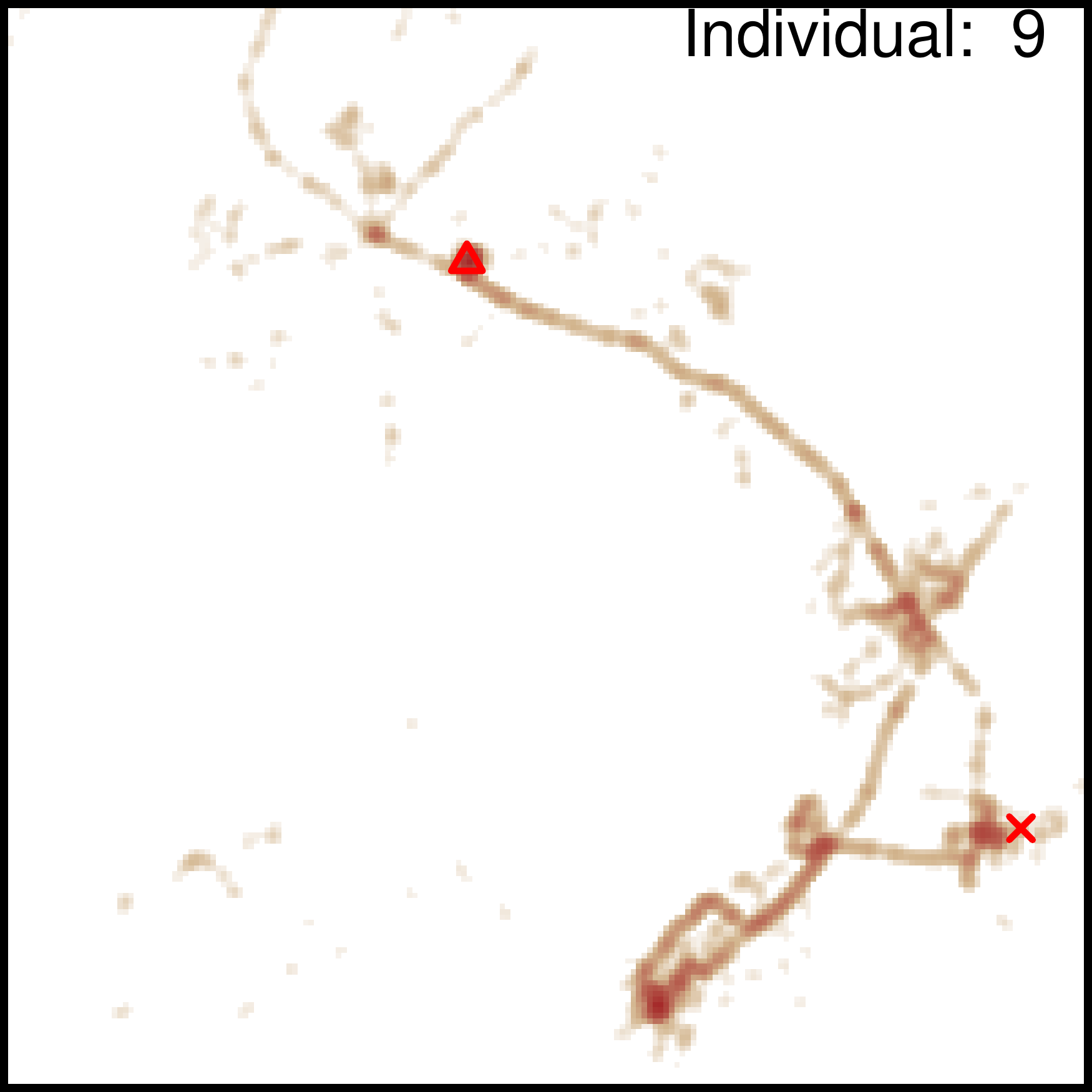}
\includegraphics[width=1.6in]{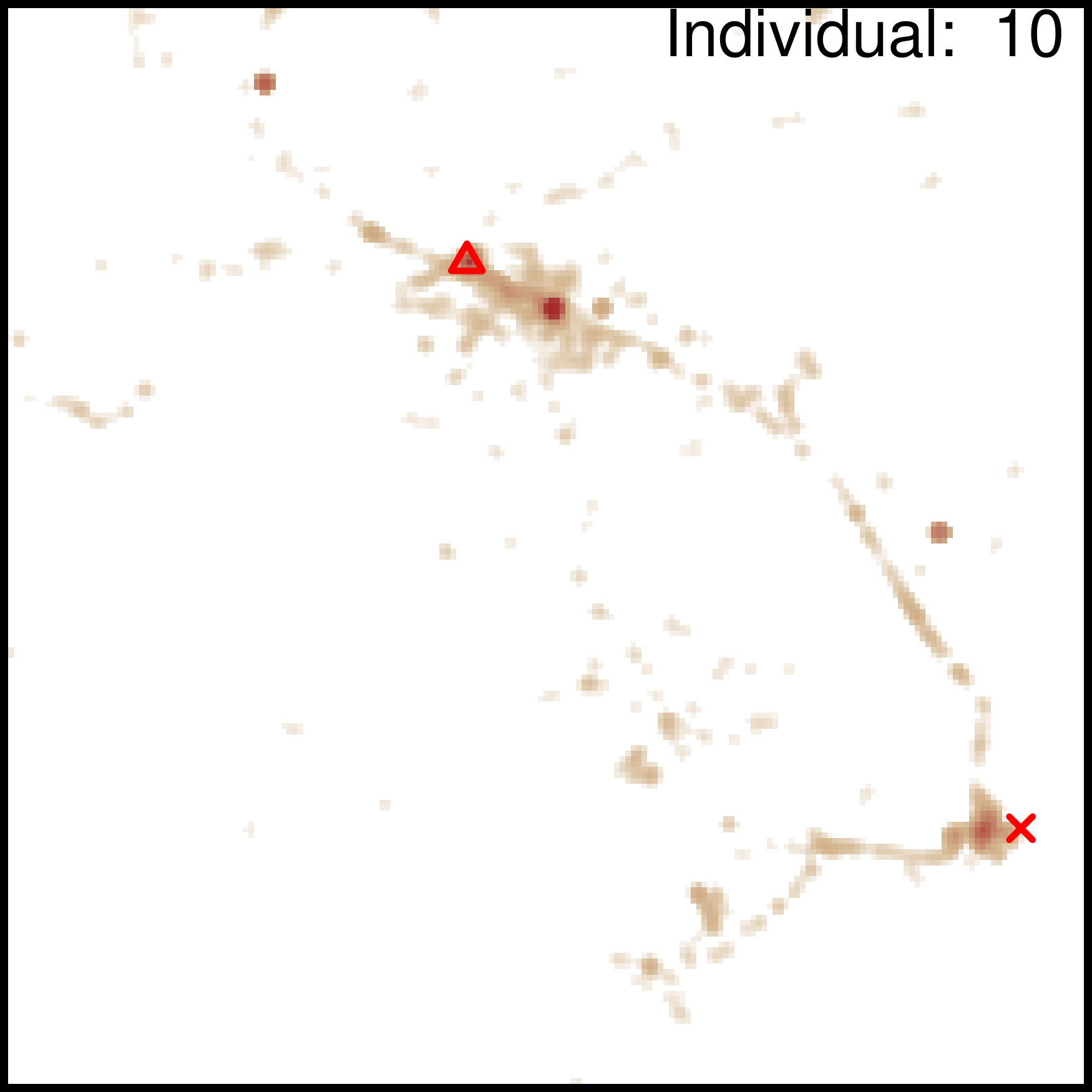}
\caption{
Maps showing the density ranking of the locations recorded for each of participant in the GPS pilot study. The location of the workplace of the study participants is marked with a red triangle, while the location of the center of the township in the study area is marked with a red cross.
}
\label{fig::17h200}
\end{figure}

In Figure~\ref{fig::17h200lv} we overlap the top activity spaces of the 10 individuals at three levels: $\gamma=0.2,0.5$, and $0.8$. The workplace was included in all top $20\%$ activity spaces (top left panel), while  the township was included in all top $50\%$ activity spaces (top right panel). Paths that follow several local roads are included in most of the top $80\%$ activity spaces (bottom panels). Except for the workplace and the township, the rest of the top $20\%$ activity spaces of the 10 individuals are not overlapping: these regions are probably indicative of the locations of their homes.

\begin{figure}[!ht]
\center
\includegraphics[width=2.4in]{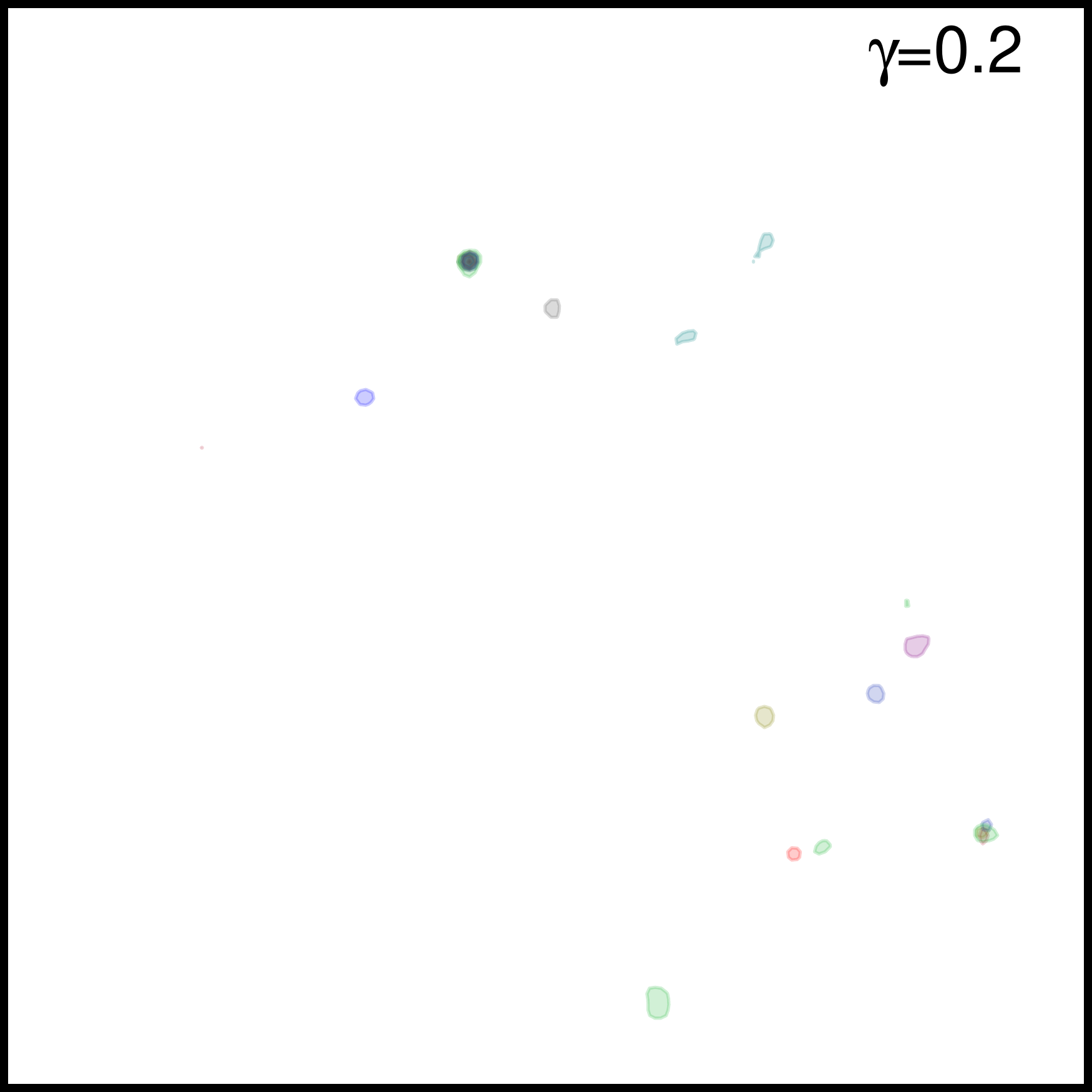}
\includegraphics[width=2.4in]{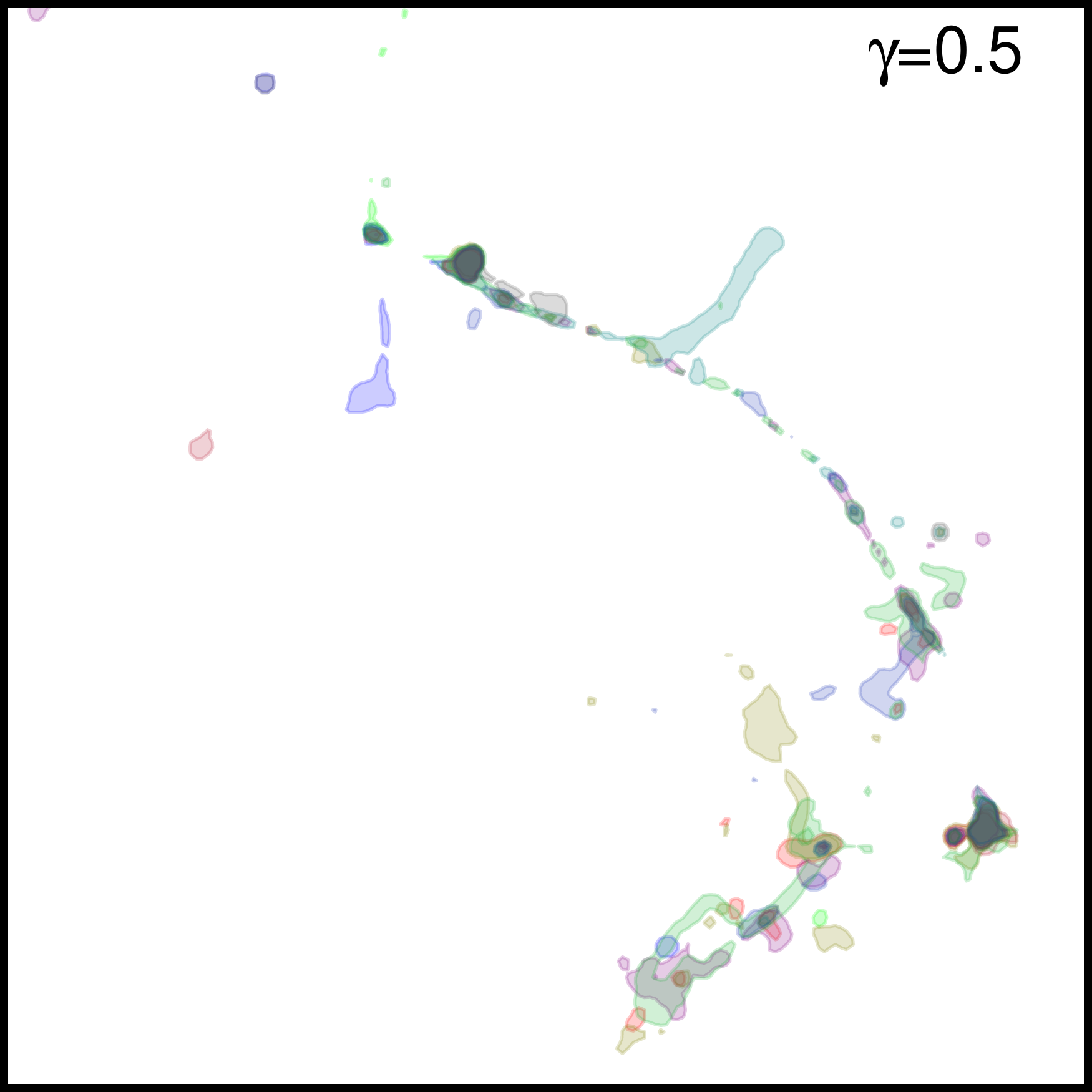}
\includegraphics[width=2.4in]{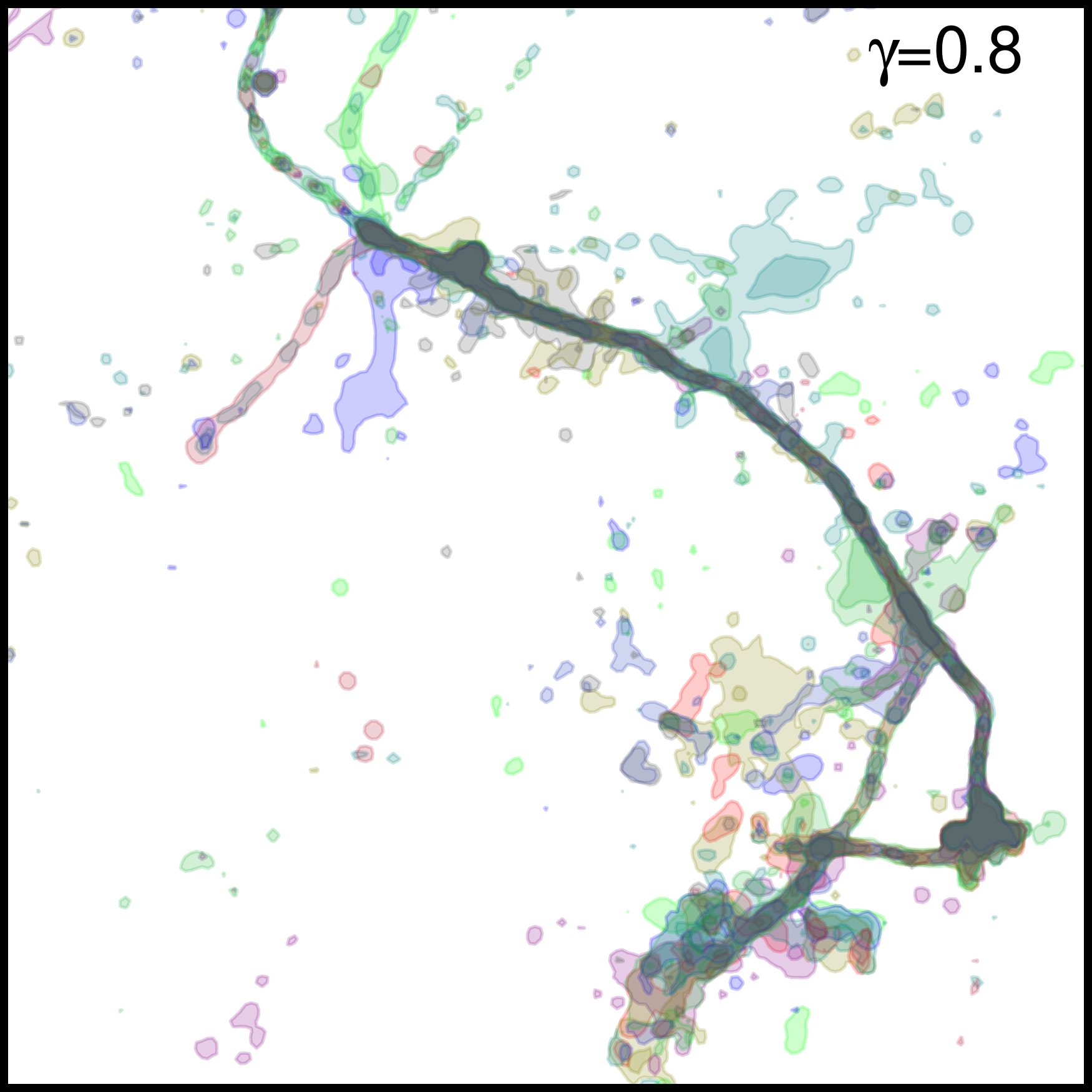}
\includegraphics[width=2.4in]{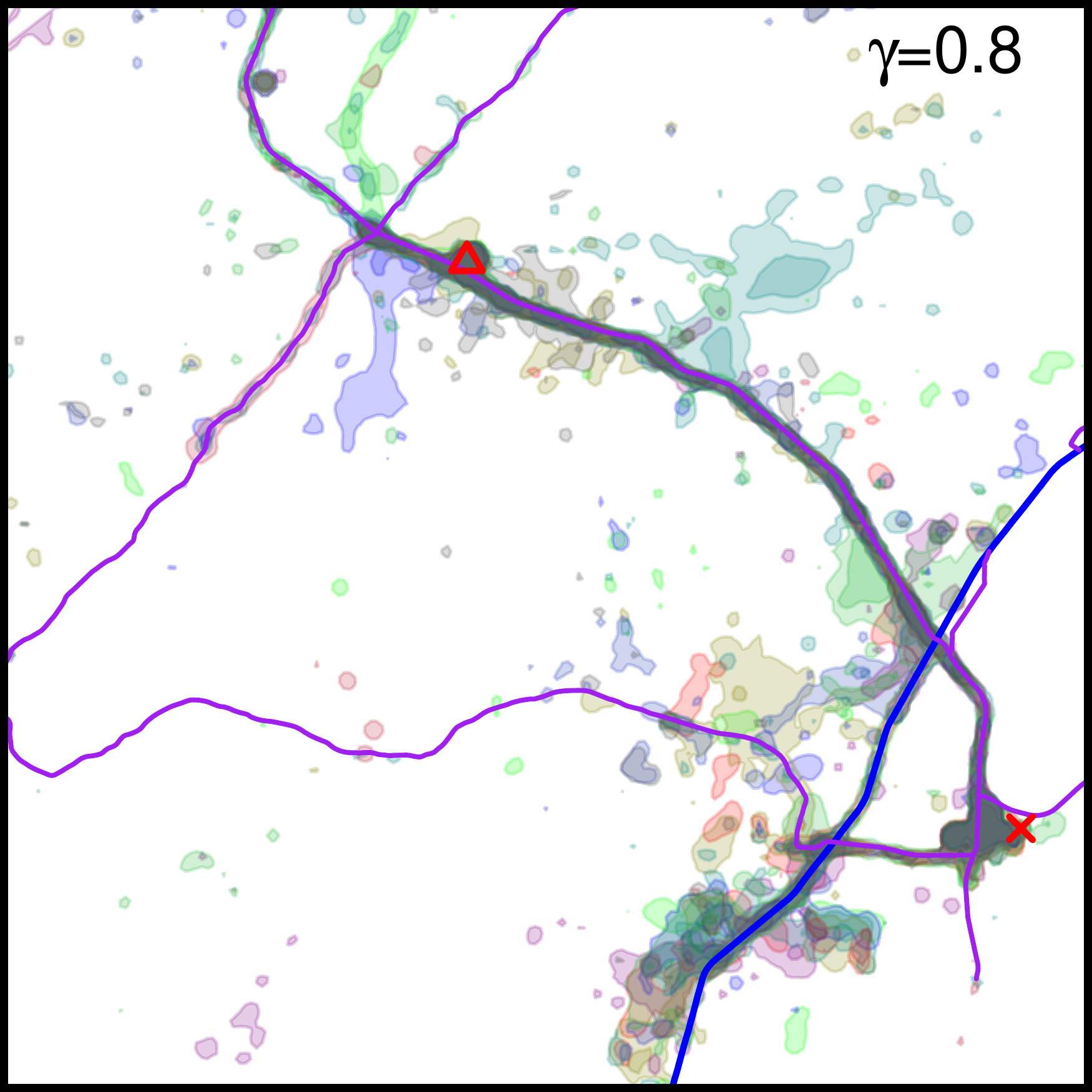}
\caption{
The top activity spaces at levels $\gamma=0.2,0.5,0.8$ from density ranking of the 10 study participants. 
Colors indicate different individuals. The bottom right panel also shows the primary and secondary roads. There are two key locations shared by all 10 individuals: the workplace (red triangle) and the location of the center of the township (red cross).}
\label{fig::17h200lv}
\end{figure}

\subsection{Mass-volume curve}
In Figure~\ref{fig::G15}, we give the mass-volume curves of the 10 study participants. We plot the function $\log V(\gamma)$ instead of $V(\gamma)$ since the size of activity space evolves rapidly when $\gamma$ changes. The gray curve which corresponds to individual 9 dominates the others in the range of $\gamma\in[0.1,0.7]$, while the purple curve which corresponds to individual 6 takes over when $\gamma>0.7$. This means that individual 9 has the highest degree of mobility when we consider the activity space of top 10\---70\% activities. Individual 6 has the highest degree of mobility in terms of the activity space of top 70\% or higher activities. The reason why these two curves dominate the others can be seen in Figure~\ref{fig::17h200}.  The regions $\hat{A}_\gamma$ for $\gamma\in[0.1,0.7]$ correspond to where the density ranking is between 0.3\---0.9 ($1-\gamma$), which is the region with a darker color.  The contours of individual 9 have a wider region with darker color compared to others.  When we consider regions with $\gamma>0.7$, we are looking at regions with a lighter color. In this case, we see that the density ranking of individual 6 spans a larger area compared to others.  The mass-volume curves flatten out when the log size of the area is roughly below $-2$.  This is due to the resolution of the raster grid of cells used to compute the size of the level sets. The corresponding calculations cannot be performed if the size of the level sets falls below the resolution of the grid.

\begin{figure}[!ht]
\includegraphics[width=2.4in]{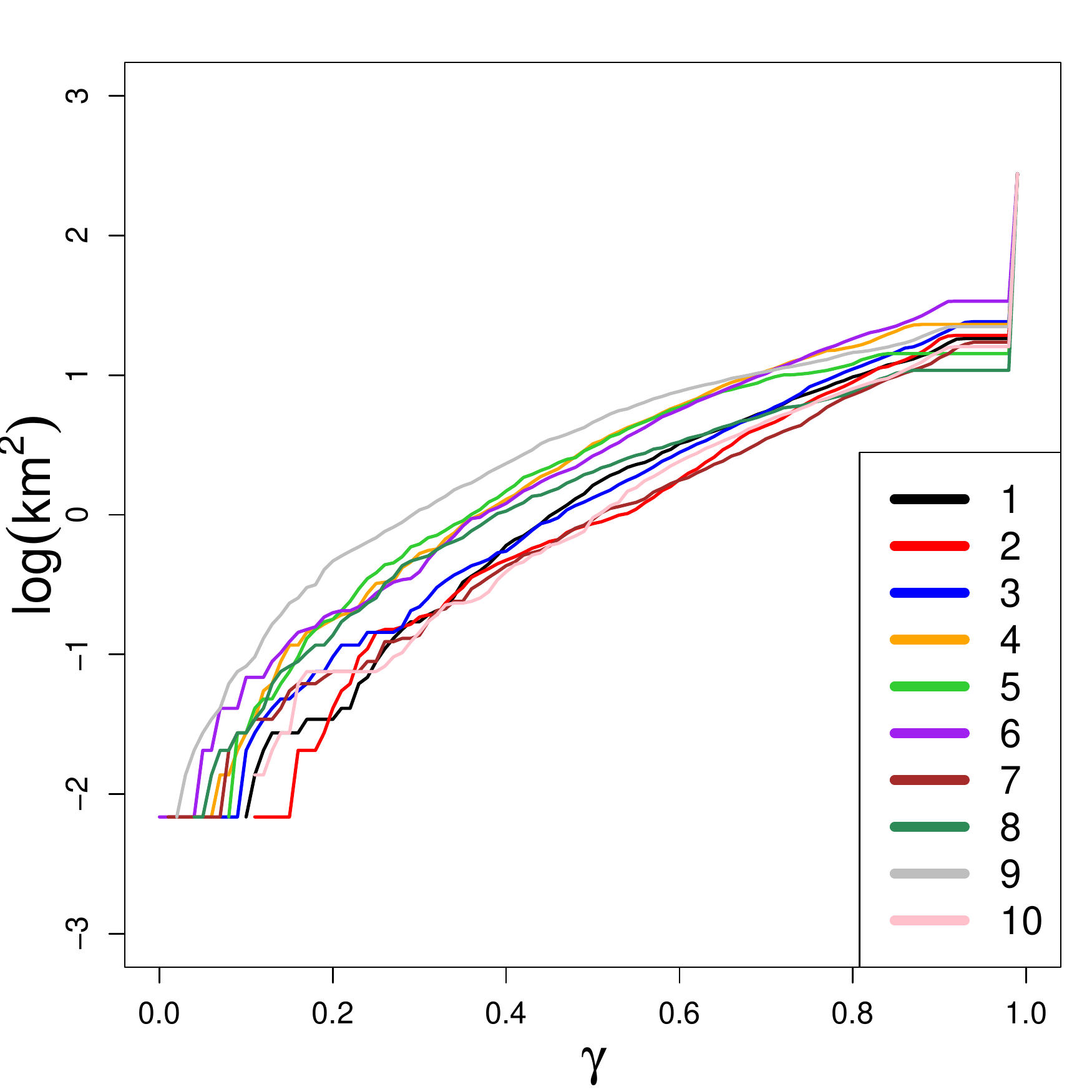}
\caption{
The mass-volume curve of the 10 individuals in the GPS pilot study measured on the log scale (log volume).}
\label{fig::G15}
\end{figure}

\subsection{Betti number curve}
Figure~\ref{fig::T15} displays the Betti number curve of every individual. There are three curves that dominate the others for different ranges of $\gamma$. When $\gamma<0.5$,  the gray curve (individual 9) dominates the others. When $0.5<\gamma<0.7$, the orange curve (individual 4) dominates.  When $0.7<\gamma$, the purple curve (individual 6) is the highest. This means that when we consider activity space of top $50\%$ activity (or an even higher level of activity),  the activity space of individual 9 has the largest number of connected components.  This can actually be seen in Figure~\ref{fig::17h200}: the darker regions in density ranking of individual 9 have more distinct connected components. The contours of density ranking of individuals 4 and 6 have many small bumps, resulting in a large number of connected components. The black and green curves associated with individuals 1 and 5 are smaller than 50. This is the result of their density ranking contours (Figure~\ref{fig::17h200}) being very concentrated.  Unlike the density ranking of individuals 4, 6 and 9 which have many little bumps, the contours of individuals 1 and 5 do not have a spurious distribution which keeps their Betti number curves at lower values. 

\begin{figure}[!ht]
\includegraphics[width=2.4in]{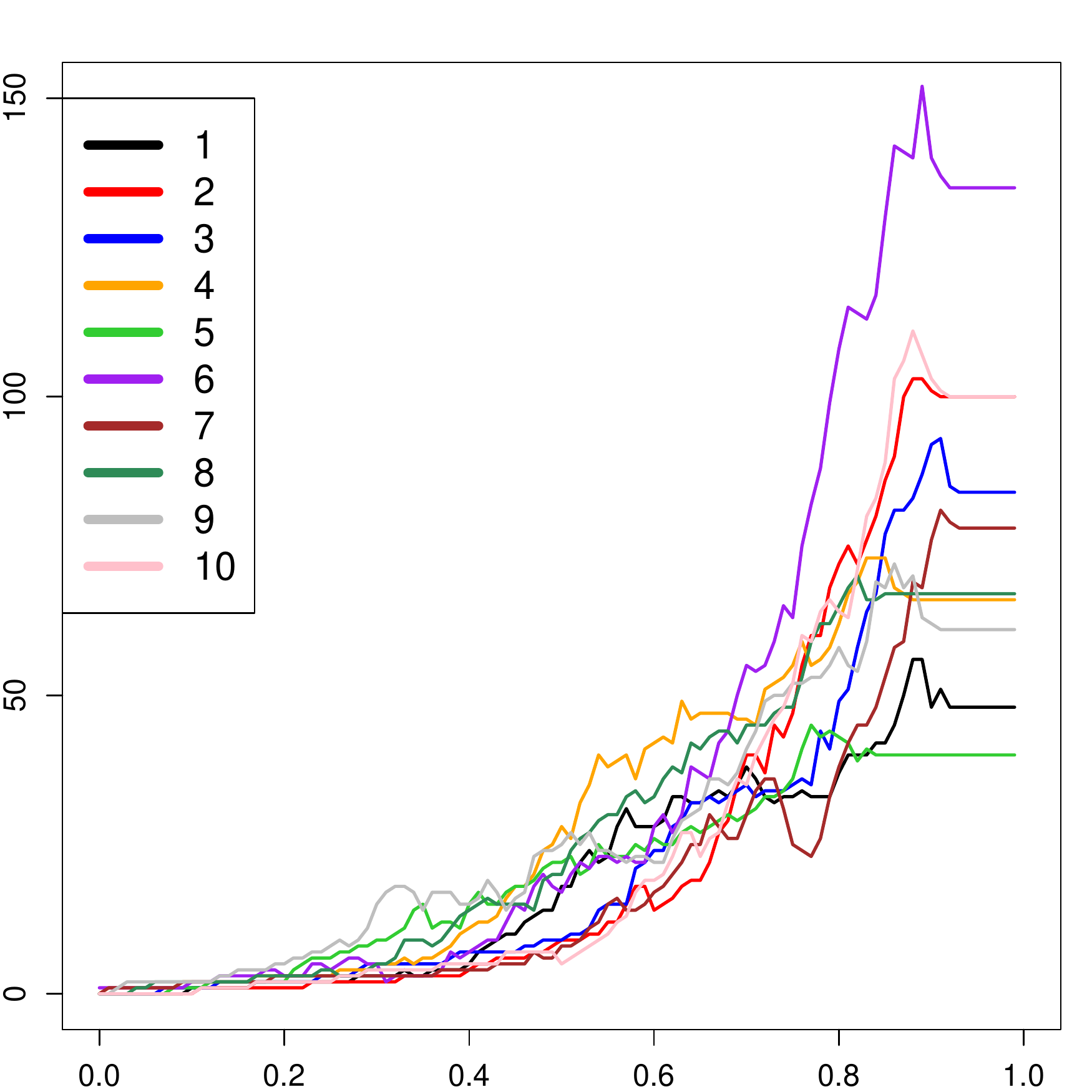}
\caption{
The Betti number curve of the 10 individuals in the GPS pilot study.
}
\label{fig::T15}
\end{figure}

Based on Figure~\ref{fig::T15}, we say that individuals 4, 6 and 9 have a higher degree of mobility, while individuals 1 and 5 have lower mobility in terms of the number of connected components of their activity spaces. Remark that the higher degree of mobility of individual 4 becomes apparent based on the Betti number curve, but is not evident based on the mass-volume curve.

\subsection{Persistence curve}

Figure~\ref{fig::P15} presents the persistence curves of the 10 study participants. The left panel displays the persistence curve in the full range, and the right panel is the zoom-in version with the range of y-axis restricted to the interval $[0,40]$. In the left panel, we see that the purple curve (individual 6) dominates the others when the range of persistence is within $[0,0.2]$.  This range corresponds to many small bumps in its distribution of density ranking \--- see Figure~\ref{fig::17h200}. These small bumps create several connected components with a short life span: they all merge with other connected components quickly, so they have small persistence. These connected components with short life spans contribute to the larger values of the persistence curve. The fact that individual 6 has many small and spurious bumps in their density ranking implies that this person repeatedly visits a larger number of locations. It is possible that this individual has a job that involves driving on a daily basis.  

\begin{figure}[!ht]
\includegraphics[width=2.4in]{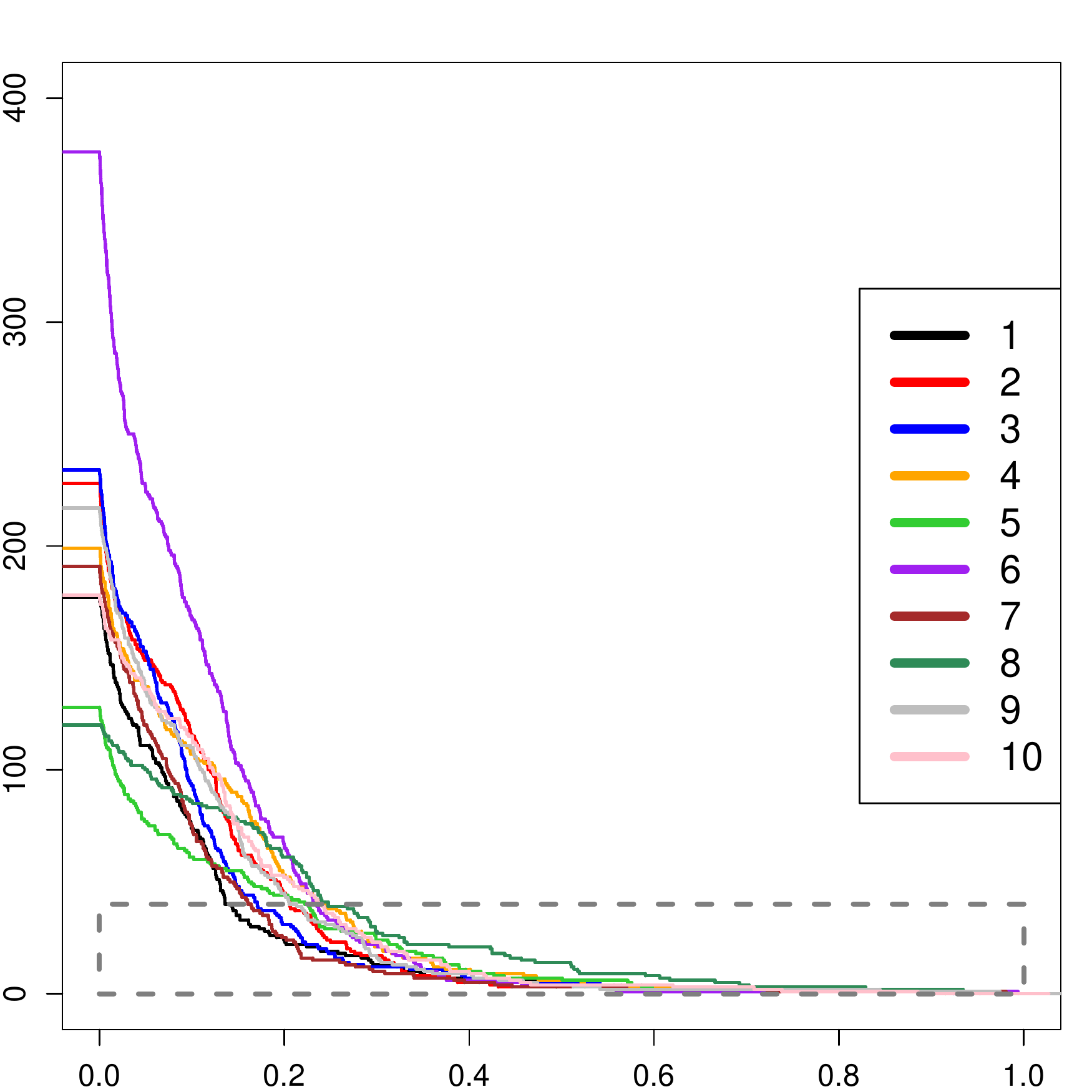}
\includegraphics[width=2.4in]{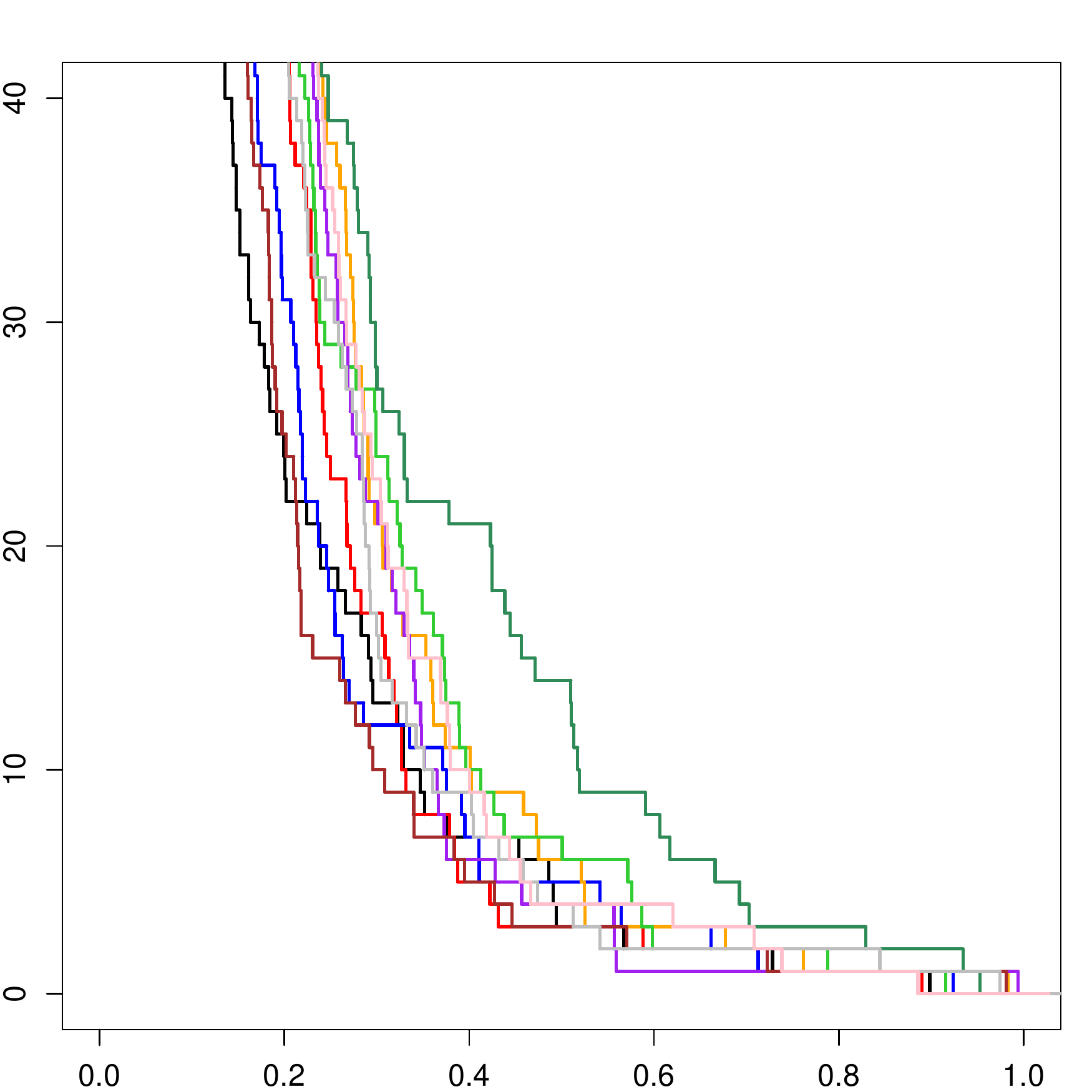}
\caption{
The persistence curve of the 10 individuals in the GPS pilot study. The left panel shows the full range of the persistence curve, and the right panel is the zoom-in version of the gray dashed box in the left panel.}
\label{fig::P15}
\end{figure}

In the right panel of Figure~\ref{fig::P15}, the dark green curve that corresponds to individual 8 stands out.  This means that individual 8 has more persistent connected components when we threshold on the persistence with a level above $0.4$. This also implies that the activities of individual 8 have several modes. From the density ranking distribution of individual 8 (Figure~\ref{fig::17h200}), we see that this individual has several distinct connected components isolated from each other, confirming that this individual's activities have several modes. This is not evident from the mass-volume curve and the Betti number curve. Therefore the persistence curve reveals key information about human mobility which complements the information provided by the other two types of curves we discussed.

\section{Discussion}	\label{sec::discussion}

In this paper we described the key elements of human activity spaces (anchor locations, roads and areas around anchor locations), and proposed a mixture model for representing these elements. We discussed density ranking as an alternative to KDE, presented three types of summary curves, and demonstrated their relevance for determining the geometry, size and structure of human activity spaces. We remark that these summary curves can also be calculated based on the KDE. However, using kernel density estimation instead of density ranking is not advisable since, as we proved in Section \ref{sec::model}, the KDE's expectation diverges at anchor locations and along road segments. 
%When individuals spend time at their anchor locations or travel on roads, the GPS locations will be distributed on a support that has a dimension less than 2, making the corresponding probability density function of recorded locations ill-defined. This causes KDE to fail to associate high intensities to grid cells along the trajectory followed by an individual. 
Density ranking has a powerful property that guarantees its convergence even when the underlying distribution contains lower dimensional structures \citep{chen2016generalized}. For this reason, it is a more appropriate to employ density ranking as opposed to KDE in the determination and measurement of human activity spaces from GPS data. 

The collection of high resolution movement data of individuals over long periods of time is possible thanks to today's technological advances. Smartphones are an especially versatile device that an evergrowing proportion of people from most countries carry around every day. At the present time, GPS datasets collected from smartphones are recorded as part of federally funded studies from many research fields. This collection effort will without doubt continue to expand in the coming years, and will provide detailed information about where people spend their time. The methods we presented in this paper could constitute a key component of these studies that will help translate raw GPS locations into meaningful, easily interpretable information about individuals' daily selective mobility. We demonstrated that density ranking and summary curves have substantial advantages over existent methods for activity space determination since they are not constrained to a fixed geometrical shape, allow the determination of anchor locations and roads used for travel, are less influenced by outlier locations, and are not dependent on the availability of quality road network data.

Human activity spaces are fundamental for health research \citep{perchoux-et-2013}, and can be interpreted as indicators of social activity, self-confidence and knowledge about the physical environment. They capture the dynamics of the geographic context \citep{RN149} which is critical in assessing individuals' exposure to social and environmental risk factors over multiple neighborhoods that are visited during activities of daily living. In particular, they are one of the foundation constructs of contextual expology \citep{kwan-2009,chaix-et-2012}. This is a subdiscipline that focuses on modeling the individuals' spatiotemporal patterns of exposure, and on the derivation of related multiplace environmental exposure variables. The premise is that even individuals from the same residential community could spend different amounts of time away from their home, and travel to locations with different characteristics. This leads to various levels of exposure to spatially-varying risk factors. Contextual expology creates customized exposure measures based on the shape, spatial spread, and configuration of the activity space of each person by taking into account their spatial polygamy \citep{matthews-2008,matthews-2011}, i.e. the amount of time spent at, around or traveling between their anchor locations. Density ranking and summary curves could be used in developing much needed exposure measures to contextual or environmental influences that take into account the spatiotemporal patterns of human mobility \citep{RN147}.

An open research question relates to linking sociodemographic characteristics of individuals with their activity spaces, and studying the interactions that might exist between the characteristics of places and the characteristics of individuals that visit these places \citep{schonfelder-axhausen-2003}. Research on activity spaces could lead to more effective individual-tailored interventions that take into consideration multiple geographic contexts. Such interventions could provide customized information to individuals about sources of healthy food, outdoor places to walk or exercise, or local social events based on their own spatial mobility patterns.

\section{Acknowledgments}

Y.C. received partial support from the National Science Foundation Grant DMS-1810960
and National Institutes of Health Grant U01-AG016976. 
A.D. received partial support from the National Science Foundation Grant DMS/MPS-1737746. The funders had no role in the study design, data collection and analysis, decision to publish, or preparation of the manuscript.

\bibliographystyle{abbrvnat}
\bibliography{GPS-submission}

\appendix

\section{The effect of varying the smoothing bandwidth} \label{sec:app}

\begin{figure}[!ht]
\center
\includegraphics[width=1.6in]{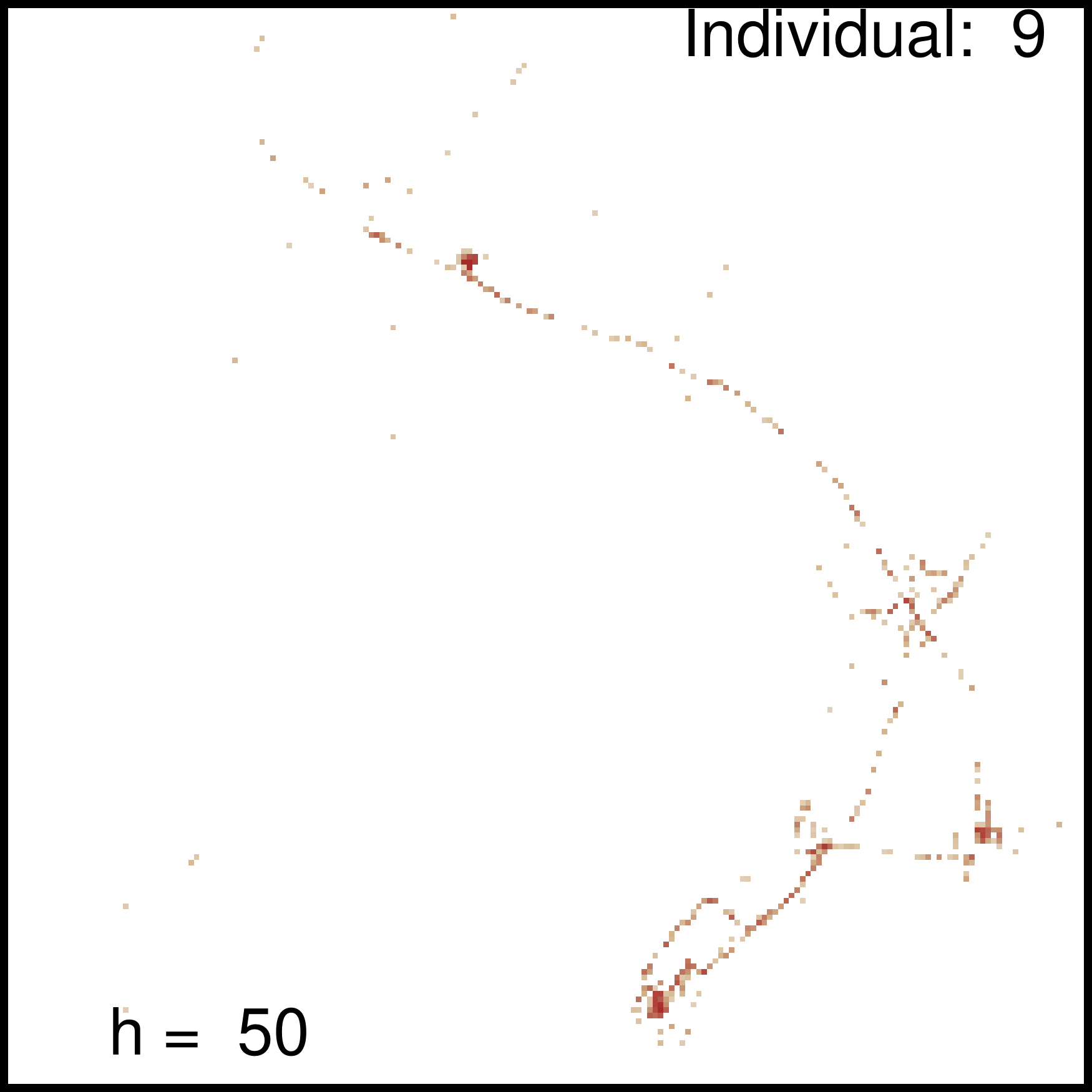}
\includegraphics[width=1.6in]{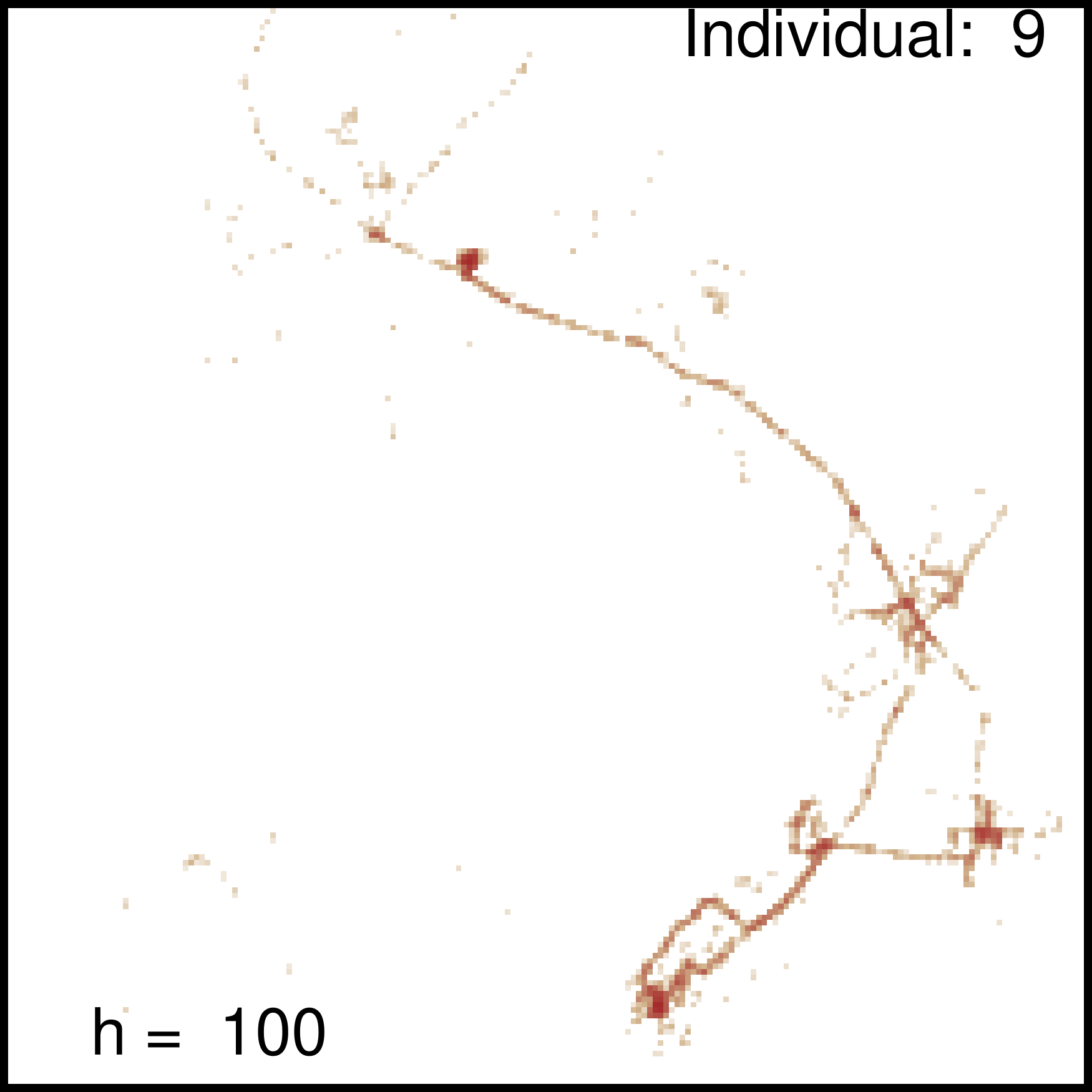}
\includegraphics[width=1.6in]{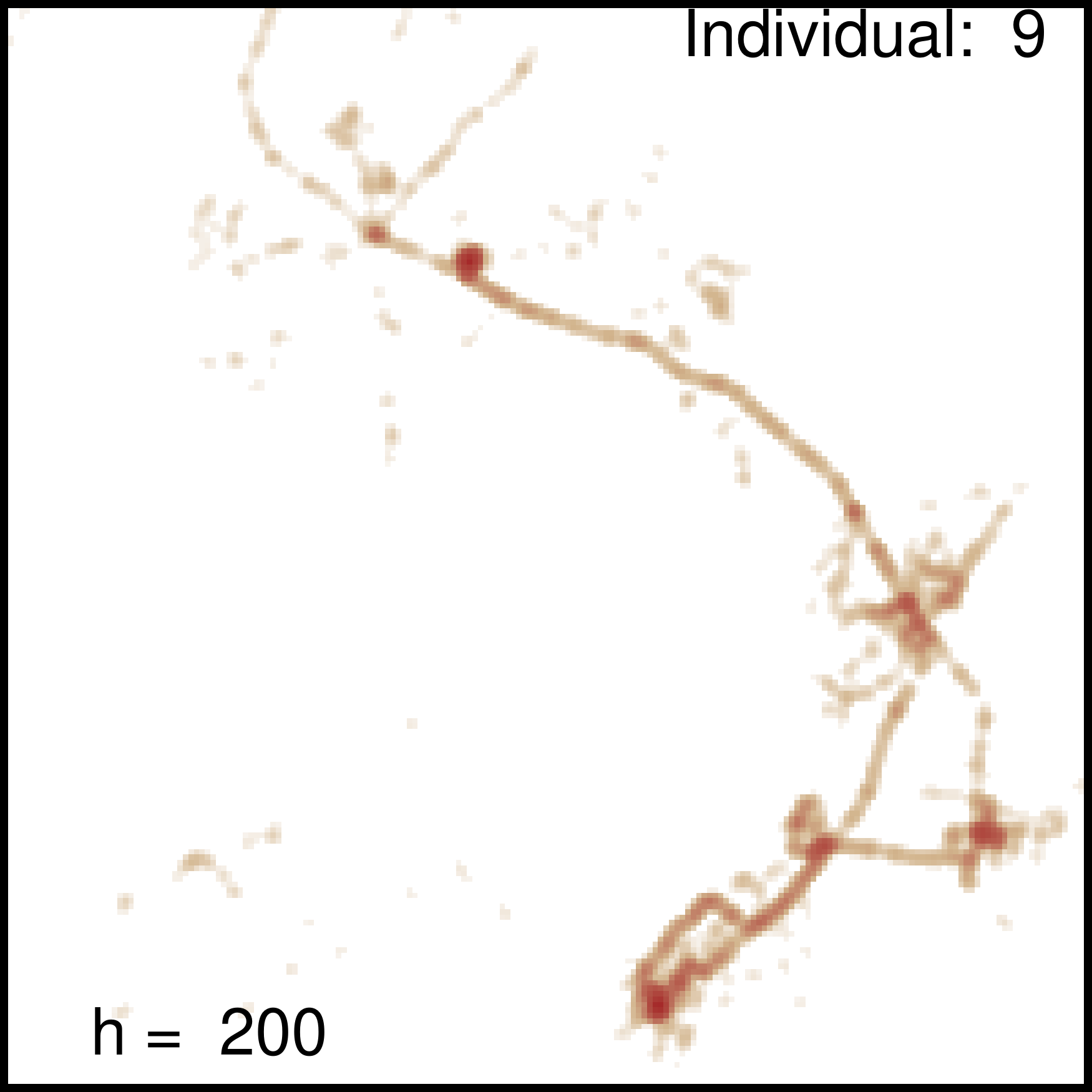}
\includegraphics[width=1.6in]{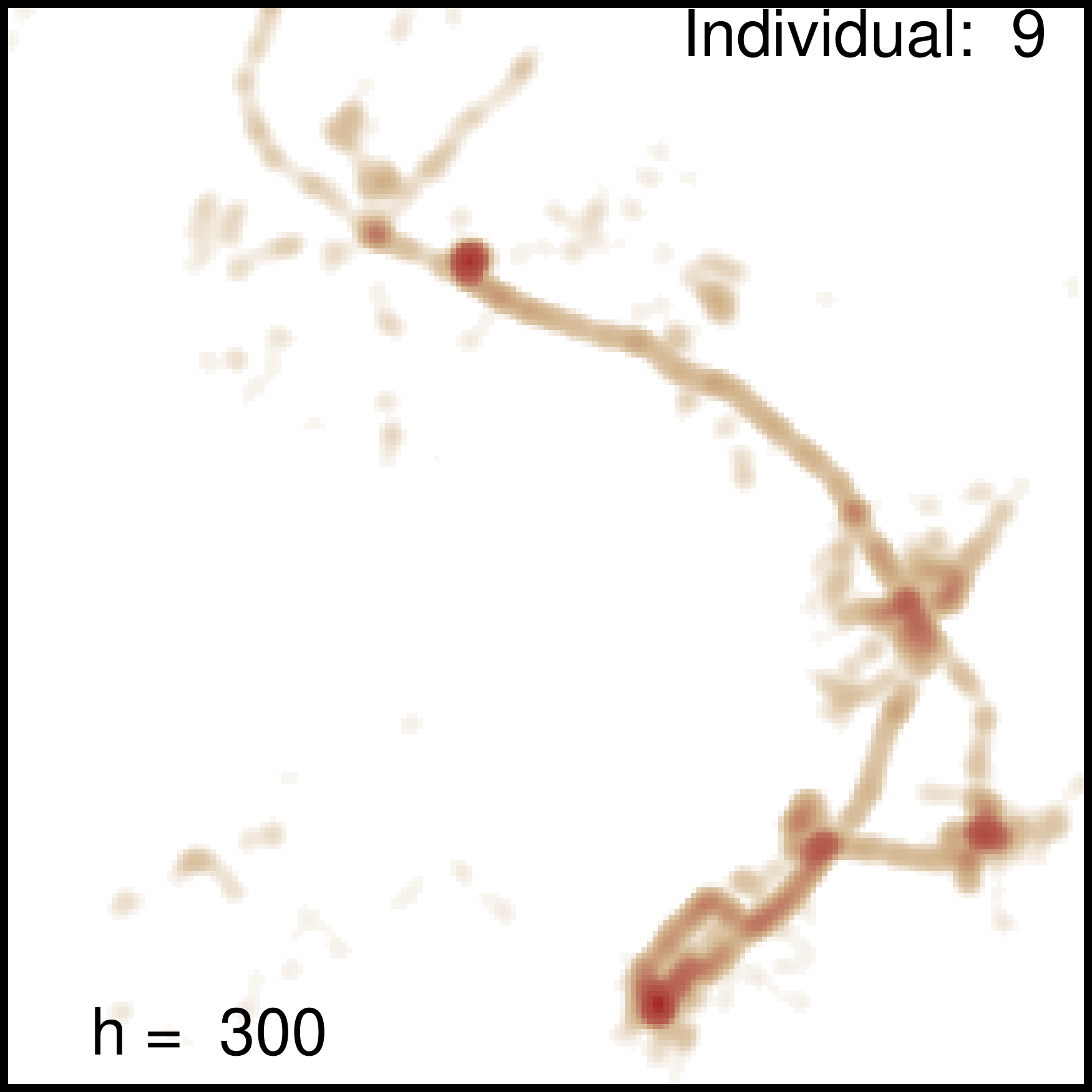}
\includegraphics[width=1.6in]{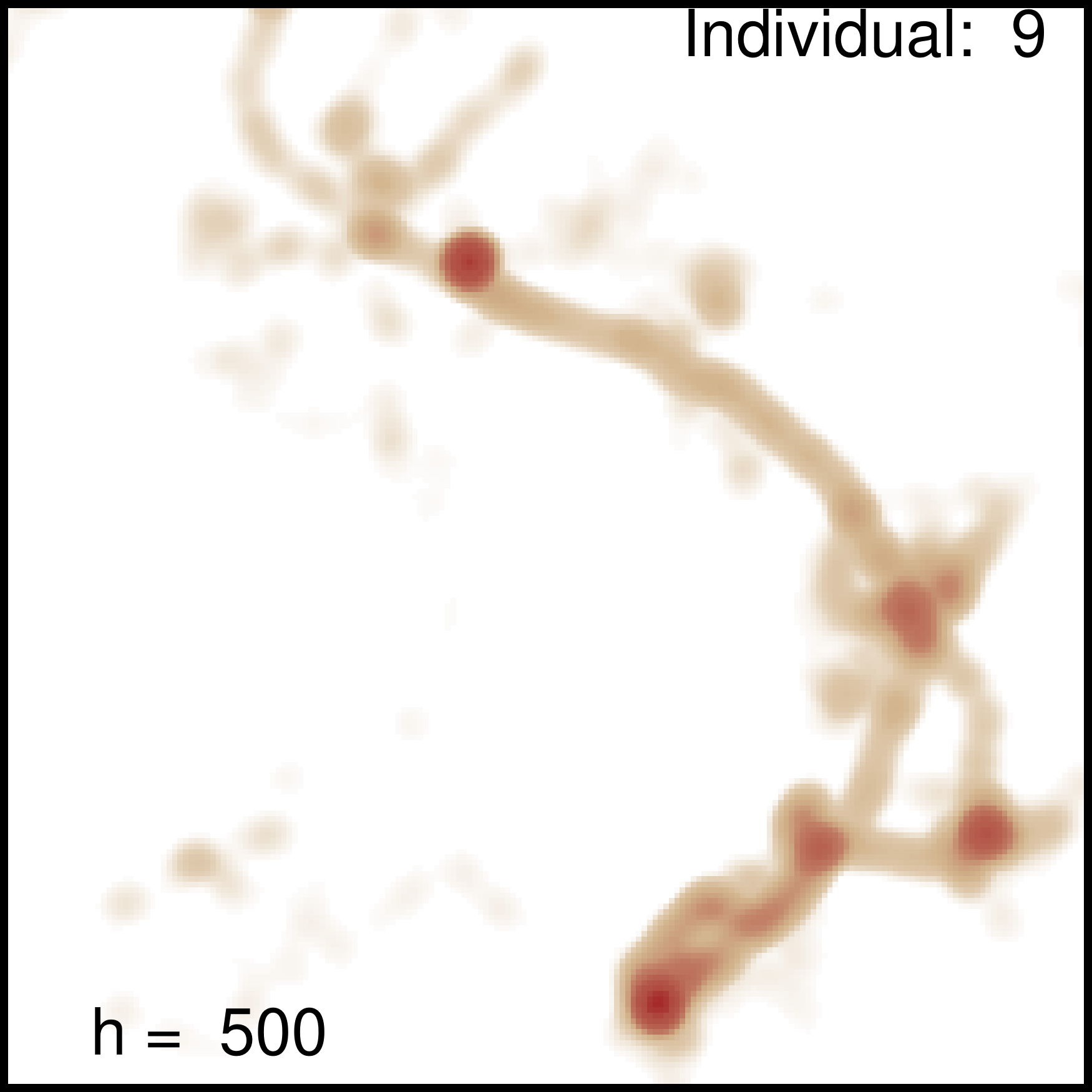}
\includegraphics[width=1.6in]{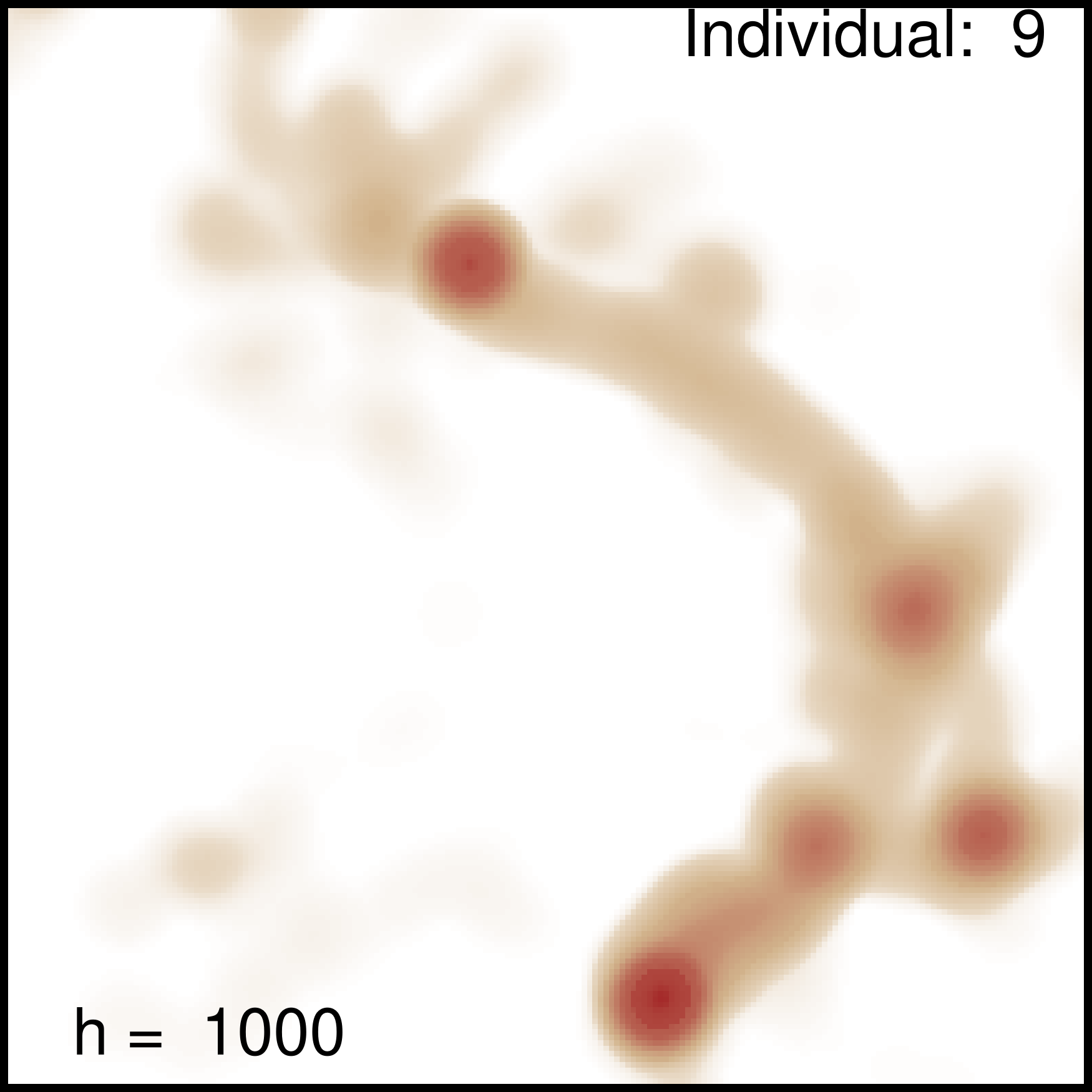}
\caption{
The effect of smoothing bandwidth on density ranking.
}
\label{fig::h_vary}
\end{figure}

%\begin{figure}[!ht]
%\center
%\includegraphics[width=1.6in]{figure/h_vary/h50_01G}
%\includegraphics[width=1.6in]{figure/h_vary/h200_01G}
%\includegraphics[width=1.6in]{figure/h_vary/h500_01G}
%\includegraphics[width=1.6in]{figure/h_vary/h50_02T5-1}
%\includegraphics[width=1.6in]{figure/h_vary/h200_02T5-1}
%\includegraphics[width=1.6in]{figure/h_vary/h500_02T5-1}
%\includegraphics[width=1.6in]{figure/h_vary/h50_03P10-4}
%\includegraphics[width=1.6in]{figure/h_vary/h200_03P10-4}
%\includegraphics[width=1.6in]{figure/h_vary/h500_03P10-4}
%\caption{
%The effect of smoothing bandwidth on mass-volume curve (left), Betti number curve (middle),
%and persistence curve (right).
%Note that 
%}
%\label{fig::h_vary2}
%\end{figure}

We explore the sensitivity of our methodology with respect to the choice of smoothing bandwidth. In Figure~\ref{fig::h_vary} we display the density ranking of locations observed for individual 9 in the GPS pilot study under different smoothing bandwidths $h$. The top left ($h=50$) and top middle ($h=100$) panels show under-smoothing: paths that connect anchor locations are separated into disjoint pieces. On the other hand, the bottom middle ($h=500$) and bottom right ($h=1000$) panels show over-smoothing: although the paths connecting anchor locations were recovered, many fine structures appear blurred due to excessive smoothing. The upper right ($h=200$) and the bottom left ($h=300$) panels seem to show an appropriate level of smoothing. We picked $h=200$ for the analysis presented in this paper.

We remark that, when evaluating density ranking, the resolution of the underlying grid of cells employed is important.  Grids with higher resolution are always preferred if the computational cost is not an issue: when the resolution of the grid is decreased, the sensitivity of density ranking to the choice of smoothing bandwidth is increased. However, one has to keep in mind that when the resolution of the grids improves, the computational cost
also increases. So, in practice, one has to balance between the quality of resolution and the computational burden.

\section{Theoretical assumptions} \label{sec::assumption}

We denote by $f^{(\ell)}(x)$ the $\ell$-th derivative of a function $f(x)$. A function is called a Morse function if all its critical points are non-degenerated.  In the development of our theoretical results, we make the following assumptions:
\begin{itemize}
\item[\bf(K1)] $K(x)$ has compact support and is non-increasing on $[0,1]$ and has at least second-order bounded derivative
and
$$
\int x^2 K^{(\beta)}(x) dx <\infty, \qquad \int  \left(K^{(\beta)}(x)\right)^2 dx <\infty
%\int x^2 K(x) dx <\infty, \qquad \int  \left(K(x)\right)^2 dx <\infty.
$$
for $\beta \leq 2$ and $K^{(2)}(0) <0$
%and $K^{(\beta)}$ is the $\beta$-th order derivative.
and $K^{(2)}(0)\geq k_2 >0$
for some constant $k_2$.

%and $\int K(\|x\|)dx=1$ and third order partial derivatives of $K(\|x\|)$ exists.
\item[\bf(K2)] Let 
\begin{align*}
\mathcal{K}_\beta &= \left\{y\mapsto K^{(\beta)}\left(\frac{x-y}{h}\right): x\in\mathbb{R}, \bar{h}>h>0\right\},
\end{align*}
be a collection of $\beta$-th derivatives of kernel functions,
where $\bar{h}$ is some positive number. Let $\mathcal{K}^*_2 = \bigcup_{r=0}^2 \mathcal{K}_r$.

%\begin{align*}
%\mathcal{K} &= \left\{y\mapsto K\left(\frac{x-y}{h}\right): x\in\mathbb{R}, \bar{h}>h>0\right\},
%\end{align*}
%be the collection of kernel funcions.
We assume that $\mathcal{K}^*_2$ is a VC-type class. i.e. 
there exists constants $A,v$ and a constant envelope $b_0$ such that
\begin{equation}
\sup_{Q} N(\mathcal{K}^*_2, \cL^2(Q), b_0\epsilon)\leq \left(\frac{A}{\epsilon}\right)^v,
\label{eq::VC}
\end{equation}
where $N(T,d_T,\epsilon)$ is the $\epsilon$-covering number for a
semi-metric set $T$ with metric $d_T$ and $\cL^2(Q)$ is the $L_2$ norm
with respect to the probability measure $Q$.

%\begin{align*}
%\mathcal{K}_\beta &= \left\{y\mapsto K^{(\beta)}\left(\frac{x-y}{h}\right): x\in\mathbb{R}, \bar{h}>h>0\right\},
%\end{align*}
%and $\mathcal{K}^*_l = \bigcup_{r=0}^l \mathcal{K}_r$
%and $\bar{h}$ is some positive number. 
%We assume that $\mathcal{K}^*_2$ is a VC-type class. i.e. 
%there exists constants $A,v$ and a constant envelope $b_0$ such that
%\begin{equation}
%\sup_{Q} N(\mathcal{K}^*_2, \cL^2(Q), b_0\epsilon)\leq \left(\frac{A}{\epsilon}\right)^v,
%\label{eq::VC}
%\end{equation}
%where $N(T,d_T,\epsilon)$ is the $\epsilon$-covering number for a
%semi-metric set $T$ with metric $d_T$ and $\cL^2(Q)$ is the $L_2$ norm
%with respect to the probability measure $Q$.

\item[\bf (S)]
$\mathcal{A}$ contains a finite number of points. $\mathcal{R}$ is the union of a finite number of smooth curves, and each of these curves is a closed set such that $\mathcal{A}\subset \mathcal{R}$. Furthermore, we assume that these curves intersect each other in a finite number of points. $\mathcal{O}$ is a compact and smooth set.
Moreover, there exists positive constants $a_0,A_0$ such that $a_0\leq \p_0(x)\leq A_0$ for 
$x\in \mathcal{A}$, and $a_0\leq \p_1(x)\leq A_0$ for 
$x\in \mathcal{R}$, and $a_0\leq \p_2(x)\leq A_0$ for 
$x\in \mathcal{O}$.

\item[\bf (P1)] 
The one dimensional density function $\p_1(x)$ is a Morse function on $\mathcal{R}$, and has bounded fourth order derivatives.
%The $1D$ density function $p_1(x)$ is continuous on $\mathcal{R}$
%and the global maxima of $p_1(x)$ exist and are well-separated from each other if there are more than
%one global maximum.
\item[\bf (P2)]
The two dimensional density function $\p_2(x)$ is a Morse function on $\mathcal{O}$, and has bounded fourth order derivatives.
% The $1D$ density function $p_1(x)$ has at least one global minimum on $\mathcal{R}$
%and if it has more than one global minima, these points are well-separated. 
%Moreover, the $2D$ density function $p_2(x)$ is continuous and has well-separated global maxima.

\end{itemize}

Assumption (K1) is a common condition on kernel functions \citep{wasserman2006all,scott2015multivariate} to control the bias and variance of the KDE for both density and density derivative estimation.  Assumption (K2) regularizes the complexity of kernel functions so  
we have the uniform convergence of the KDE \citep{Gine2002,Einmahl2005, genovese2014nonparametric,
chen2015asymptotic} and its derivatives when the probability density function exists and is smooth.
%Although we do not assume the existence of a global probability density function of $\P_{\sf GPS}$, we will use this result implicitly. 
Note that most common kernel functions, such as the Gaussian kernel, quartic kernel, or any compact support kernel, satisfy both assumptions (K1) and (K2).  The purpose of assumption (S) is to  regularize the behavior of the supports. The first part requires that every anchor points is connected to a road. The later part of the assumption assumes the densities are bounded from both the above and from $0$. This means that the activity pattern on an anchor point, a road, or an open space has to be different. Both assumptions are very reasonable for GPS data. Note that assumption (S) implies that each curve (path) is a one dimensional smooth manifold. 

%since anchor points should be connected to at least a road
%and an individual's activity pattern should depends on the
%
%ensure that the collection of anchor points and roads
%are smooth. 
%[xxx anchor points being connected to roads]
%This is a very common scenario in GPS data because 
%an individual only has a finite number of anchor locations
%and most roads are indeed curves. 

Assumptions (P1) and (P2) require that the critical points of $\p_1$ and $\p_2$
are well-defined. 
%Assumption (P1) requires that the 1D density function has a global maximum. 
Although these assumptions seem to be technical, they are quite reasonable because the actual (realized) speed of travel on a road is often location-specific: there will be regions with higher driving speeds, and regions with lower driving speeds due to legal speed limits, intersections, built environment or natural obstacles. %Note that although the distribution of our simulated data (Appendix \ref{sec::sim}) does not follow assumptions (P1) and (P2), the density ranking approach still works well. 

%Assumption (P1) and (P2) requires that the 1D density function has a global maximum
%and a global minimum
%%Assumption (P1) requires that the 1D density function has a global maximum. 
%Although these assumptions seem to be technical, they are reasonable in reality
%because our speed on a road is often location-specific -- there will
%be regions where we drive faster and regions we move slower (may be caused by
%speed limit, intersections, or landscapes).
%Note that although the distribution of our simulated data (Appendix \ref{sec::sim}) does not follow
%assumption (P1) and (P2), the density ranking still works well. 

\section{Proofs of theoretical results} \label{sec::proofs}

For any set $A$ and a positive number $r_0$, we define the set
$$
A\oplus r_0 = \{x: d(x,A)\leq r_0\},
$$
where $d(x,A) = \inf_{y\in A}\|x-y\|$ is the shortest distance from point $x$ to $A$.

\begin{proof}[Proof of Theorem~\ref{lem::equiv}]

{\bf Case of $\omega(x)=0$.} By definition, we have $\omega(x)=0\Leftrightarrow x\in\mathcal{A}\Leftrightarrow \p_0(x)>0$.
Thus, all we need to prove is the equivalence to the last definition. 
When $\p_0(x)>0$ and $\pi_0>0$, for any positive integer $s$, 
$$
\frac{\P_{\sf GPS}(B(x,r))}{r^s} \geq \frac{\pi_0 \P_{0}(B(x,r))}{r^s}  
$$
will diverge when $r\rightarrow 0$. 
Thus, 
$\max \{s: \mathcal{H}_s(x)<\infty\} = 0$. 
This proves that $\omega(x)=0$ implies $\max \{s: \mathcal{H}_s(x)<\infty\} = 0$. 
On the other hand, $\max \{s: \mathcal{H}_s(x)<\infty\} = 0$ implies 
that there is a point mass at $x$, so $\p_0(x) > 0$ and thus, $\omega(x) = 0$. 
This proves the case of $\omega(x)= 0$.

{\bf Case of $\omega(x)=1$.} Recall that, by definition, $\omega(x)=1\Leftrightarrow x\in\mathcal{R}\backslash\mathcal{A}$.
Using the definition of the support $\mathcal{A}$ and $\mathcal{R}$, $\omega(x)=1$ implies that $\p_0(x)=0$ ($x\notin \mathcal{A}$), $\p_1(x)>0$ ($x\in \mathcal{R}$), and $p_0(x) = 0$ with $p_1(x)>0$ implies $\omega(x)=1$. Thus we proved the equivalence of the first and the third definition. 

To show the equivalence to the second definition, the condition that $\p_1(x) \leq A_0$ for every $x\in \mathcal{R}$ implies that 
$\p_1(x)>0\Longrightarrow \mathcal{H}_1(x)<\infty$. 
And $\frac{\P_{\sf GPS}(B(x,r))}{r^2}\geq \frac{\pi_1\P_1(B(x,r))}{r^2}$
will diverge when $r\rightarrow 0$ so $\mathcal{H}_2(x)= \infty$.
Therefore, $\omega(x) = 1\Rightarrow \max \{s: \mathcal{H}_s(x)<\infty\} = 1$.

On the other hand, $\max \{s: \mathcal{H}_s(x)<\infty\} = 1$ implies $\mathcal{H}_2(x)= \infty$ and $\mathcal{H}_1(x)<\infty$.  We have $\mathcal{H}_2(x)= \infty\Rightarrow x\in \mathcal{A}\cup \mathcal{R}$, and also 
\begin{eqnarray*}
\mathcal{H}_1(x) & = & \lim_{r\rightarrow 0}\frac{P_{\sf GPS}(B(x,r))}{C_1 r},\\
 & = & \lim_{r\rightarrow 0}\frac{\pi_0 P_0(B(x,r)) + \pi_1 P_1(B(x,r))+\pi_2 P_2(B(x,r))}{C_1 r}.
\end{eqnarray*}
\noindent Thus we must have $\p_0(x) = 0$ because otherwise the first term will diverge. Moreover, $\p_0(x) = 0\Rightarrow x\notin \mathcal{A}$. Thus, $x\in \mathcal{R}\backslash \mathcal{A} \Leftrightarrow \omega(x)=1$.

{\bf Case of $\omega(x)=2$}. Because $\omega(x)=2\Leftrightarrow x\notin\mathcal{A}\cup\mathcal{R}$, $\p_0(x)=\p_1(x)= 0$. Thus, $\p_2(x) \geq 0$ so this is equivalent to the third definition. Now we derive the equivalence to the second definition. Assumption (S) implies that $\mathcal{A}\cup \mathcal{R}$ is a closed set. Thus, for a point $x\notin \mathcal{A}\cup \mathcal{R}$, there exists a constant $r_0>0$ such that 
$B(x,r_0)\cap (\mathcal{A}\cup \mathcal{R}) = \emptyset$. It follows that, when $r<r_0$, we have
$$
\frac{\P_{\sf GPS}(B(x,r))}{C_2 r^2} = \frac{\pi_2 \P_{2}(B(x,r))}{C_2 r^2}.
$$
Therefore
$$
\mathcal{H}_2(x) = \lim_{r\rightarrow 0}\frac{\P_{\sf GPS}(B(x,r))}{C_2 r^2}
= \lim_{r\rightarrow 0}\frac{\pi_2 \P_{2}(B(x,r))}{C_2 r^2} = \pi_2\p_2(x)<\infty
$$
which proves that $\omega(x)=2\Rightarrow \max\{s: \mathcal{H}_s(x)<\infty, s=0,1,2\}=2$.

To prove the other direction, we assume that $\mathcal{H}_2(x)<\infty$, and try to prove that $\p_0(x)= \p_1(x) = 0$ (this is equivalent to $\omega(x) = 2$).  We proceed by proof by contradiction.  Assume that $\p_0(x)$ or $\p_1(x)$ are positive.
%We must have $p_0(x)=p_1(x)=0$
%otherwise
By definition,
$$
\mathcal{H}_2(x) = \lim_{r\rightarrow 0}\frac{\P_{\sf GPS}(B(x,r))}{C_2 r^2}\geq 
\lim_{r\rightarrow 0}\frac{\P_{j}(B(x,r))}{C_2 r^2}
$$
for $j=0$ and $1$. 
Because $\p_0(x)$ or $\p_1(x)$ are positive, 
$$
\lim_{r\rightarrow 0}\frac{\P_{j}(B(x,r))}{C_2 r^2}
$$
will diverge, which implies $\mathcal{H}_2(x)= \infty$, a contradiction.
Thus, we conclude that $\p_0(x)=\p_1(x) = 0$, which completes the proof.

%and the last term $\frac{\pi_2 P_2(B(x,r))}{C_1 r}\rightarrow 0$. 
%$p_0(x)=0\Rightarrow x\notin \mathcal{A}$ and $\mathcal{H}_2(x)= \infty\Rightarrow  x\$

%
%So we obtain $\mathcal{H}_1(x) =  \lim_{r\rightarrow 0}\frac{\pi_1 P_1(B(x,r))}{C_1 r} = p_1(x)$,
%which concludes $p_1(x)>0$.

\end{proof}

\begin{proof}[Proof of Theorem~\ref{thm::alpha}]
By Theorem~\ref{lem::equiv}, we have
$$
\omega(x) = \max \{s: \mathcal{H}_s(x)<\infty, s=0,1,2\},
$$
hence $\alpha(x)$ is equivalent to the one defined in \cite{chen2016generalized}.

We will prove this theorem using Theorem 10 of \cite{chen2016generalized}. Note that Theorem 10 of \cite{chen2016generalized} requires four assumptions, two assumptions on kernel functions that are the same as ours. The other two assumptions of Theorem 10 of \cite{chen2016generalized} include a manifold assumption denoted as (S') and a density assumption denoted as (P'). 

To apply their result, we need to verify that our assumptions (S) and (P1-2) are sufficient to satisfy their assumptions (S') and (P').  Assumption (P') requires that $\p_1$ and $\p_2$ are Morse functions with bounded continuous fourth-order derivative, and each $\p_j(x)$ is uniformly bounded from the above and from $0$. Thus, our assumption (S) and (P1-2) are sufficient to the assumption (P').  Assumption (S') requires that the set $\mathcal{R}$ is a smooth manifold. Our assumption (S) only requires that $\mathcal{R}$ is the union of finite number of smooth curves ($1D$ manifolds), which is weaker than assumption (S'). 

We will argue that their results still apply. Recall that  assumption (S) requires that curves of $\mathcal{R}$ only intersect on a finite number of points.
Let $\mathbb{I}$ be the collection of these intersections. Let $r_n\rightarrow0$ be a sequence of positive numbers, and restrict our attention to the data within $(\mathbb{I}\oplus r_n)^C= \{x: d(x, \mathbb{I})\geq r_n\}$. Then under assumption (S), $\p_1$ within $(\mathbb{I}\oplus r_n)^C$
has a support that is a smooth manifold, which satisfies assumption (S'). Thus, Theorem 10 of \cite{chen2016generalized} applies and we obtain the desired result on $(\mathbb{I}\oplus r_n)^C$. Note that this introduce an error of $\P_{\sf GPS}(\mathbb{I}\oplus r_n)  = O(r_n)$, which shrinks to $0$ when $r_n\rightarrow0$. Therefore, asymptotically we obtain the desired result. 

Note that Theorem~\ref{thm::conv_set} is substantially a different result compared to the results presented in \citet{chen2016generalized}. Their results focus on the convergence of a function estimator rather than a level set estimator. Convergence of a function estimator does not imply the corresponding level sets converge. 
\end{proof}

Before proving Theorem~\ref{thm::conv_set}, we first define some useful notations. Under assumption (P2), there is a global minimum of $\p_1(x)$ on $\mathcal{R}$. Let $m_0$ be this global minimum. 

We define a sequence of sets 
\begin{equation*}
W_n = (\mathcal{A}\cup \mathcal{R}) \backslash (m_0 \oplus h^{-\frac{1}{4}}).
\end{equation*}
This sequence will be very useful when proving the second assertion of Theorem~\ref{thm::conv_set}. 

For any set $A$, we define
$\P_{\sf GPS}(A) = \P(X\in A)$ such that $X$ has a distribution function $\P_{\sf GPS}$ and $\hat{\P}_n(A) = \frac{1}{n}\sum_{i=1}^n I(X_i \in A)$.
Note that $\hat{\P}_n(\hat{A}_\gamma) = \gamma + O(n^{-1})$ by construction. 

Finally, we prove two useful lemmas. 
\begin{lem}
Given assumptions (K1-2) and (S) and (P1), we have
\begin{align*}
\P\left(\hat{A}_{\pi_0}\subset (\mathcal{A}\oplus h)\right)&\rightarrow 1.
\end{align*}
If we further assume (P2), we have
$$
\P\left(\hat{A}_{\pi_0+\pi_1}\subset ((\mathcal{A}\cup\mathcal{R})\oplus h)\right)\rightarrow 1.
$$
\label{lem::inside}
\end{lem}
\begin{proof}
{\bf First assertion.}
Let $\mathbb{K}_0(h) = (\mathcal{A}\oplus h)^C$ be the complement of $(\mathcal{A}\oplus h)$.
Note that 
$$
\hat{A}_{\pi_0}\subset (\mathcal{A}\oplus h) \Leftrightarrow \hat{A}_{\pi_0}\cap \mathbb{K}_0(h) = \emptyset.
$$
We will prove the result by showing that 
$\P(\hat{A}_{\pi_0}\cap \mathbb{K}_0(h) \neq \emptyset)\rightarrow 0$.

Because $\P_{\sf GPS}$ has point mass on $\mathcal{A}$, $\hat{\sf p}$ diverges much faster than
other regions. Moreover, because the kernel function $K(x)$ is smooth due to assumption (K1),
points around $\mathcal{A}$ also get smoothing effect from $\mathcal{A}$. As a result, any point within $\mathcal{A}\oplus \frac{h}{2}$ will gain the smoothing effect from $\mathcal{A}$.
Therefore, when $h\rightarrow0$, 
$$
\P\left(\sup_{x\in \mathbb{K}_0(h)} \hat{\sf p}(x) \leq \min_{x\in \mathcal{A}\oplus \frac{h}{2}} \hat{\sf p}(x)\right)\rightarrow 1. 
$$

Thus, if $\hat{A}_{\pi_0}$ contains any point in $\mathbb{K}_0(h)$, 
with a probability tending to 1, $\hat{A}_{\pi_0}$ must contain $\mathcal{A}\oplus \frac{h}{2}$. 
Because $\P_{\sf GPS}(\mathcal{A}\oplus \frac{h}{2}) = \pi_0 + O(h)$, 
$\hat{\P}(\mathcal{A}\oplus \frac{h}{2}) \geq \pi_0 + O(h)-\Delta_n$,
where $\Delta_n = \sup_{r>0} |\hat{\P}(\mathcal{A}\oplus r)-\P_{\sf GPS}(\mathcal{A}\oplus r)|$.
As a result, a necessary condition of $\hat{A}_{\pi_0}\cap \mathbb{K}_0(h) \neq \emptyset$
is $\hat{\P}(\mathcal{A}\oplus \frac{h}{2})< \pi_0$, which requires
$$
\pi_0 + O(h)-\Delta_n \leq \hat{\P}\left(\mathcal{A}\oplus \frac{h}{2}\right)< \pi_0.
$$
Thus, we need $\Delta_n > O(h)$. 

Because the set $\{\mathcal{A}\oplus r: r>0\}$ has a VC dimension $1$,
due to VC theory (e.g., Theorem 2.43 of \citealt{wasserman2006all}), 
$\P(\Delta_n >\epsilon) \leq a_0 n e^{-8n\epsilon^2}$ for some constant $a_0$.
Under the assumption that $\frac{nh^2}{\log n}\rightarrow \infty$,
we have 
$\P(\Delta_n >O(h)) \rightarrow 0$,
and this implies that $\P(\hat{A}_{\pi_0}\cap \mathbb{K}_0(h) \neq \emptyset) \rightarrow 0$,
the desired result.

{\bf Second assertion.}
The proof of the second assertion follows the same way as the first assertion.
The key is replacing $\mathcal{A}$ by $\mathcal{A}\cup \mathcal{R}$
and use the fact that any point within $(\mathcal{A}\cup \mathcal{R})\oplus \frac{h}{2}$
diverges faster than any point outside $(\mathcal{A}\cup \mathcal{R})\oplus h$. 
Thus, we omit the proof.

%By Theorem 8 of \cite{chen2016generalized}, 
%there exists a constants $C_j$ depending only on $j$ and the kernel function such that
%$$
%\sup_{x\in \mathbb{K}_0(h)}\|\|
%$$
%
%xxx a short one

\end{proof}

\begin{lem}
We assume (K1-2) and (S) and (P0-1). Given the above notations, we have
$$
\P\left(W_n\subset \hat{A}_{\pi_0+\pi_1}\right)\rightarrow 1.
$$
\label{lem::Wn}
\end{lem}
\begin{proof}

Recall that $W_n = (\mathcal{A}\cup \mathcal{R}) \backslash (m_0\oplus h^{-1/4})$.
Because $m_0$ is the global minimum of $\p_1(x)$ on $\mathcal{R}$, and $\p_1(x)$ is a smooth function along $\mathcal{R}$, 
we can assume that every point of $W_n$ has a one dimensional density that is above or equal to $\p_1(m_0)+ c_1 \sqrt{h}$
for some constant $c_1$.
Namely,
$$
\inf_{x\in W_n} \p_1(x) \geq p_1(m_0) + c_1 \sqrt{h},
$$
where $c_1$ is a constant related to the second derivative of $\p_1(x)$. 

We will derive the probability by considering the complement event, i.e.,
$W_n\not\subset \hat{A}_{\pi_0+\pi_1}$,
and then show that such an event occurs with a probability tending to $0$. 
If $W_n\not\subset \hat{A}_{\pi_0+\pi_1}$, 
we can then find a point $x_0\in \hat{A}_{\pi_0+\pi_1}$ but $x_0\notin W_n$. 
Therefore, $\p_1(x_0) \geq \p_1(m_0) + c_1 \sqrt{h}$. 

Let $\hat{\p}_*$ be the density threshold used for constructing $\hat{A}_{\pi_0+\pi_1}$, i.e.,
$$
\hat{A}_{\pi_0+\pi_1} =  \{x: \hat{\alpha}(x)\leq \pi_0+\pi_1\} = \{x: \hat{\sf p}(x)\leq \hat{\p}_*\}. 
$$
Because $x_0\notin \hat{A}_{\pi_0+\pi_1}$, $\hat{\sf p}(x_0)< \hat{\p}_*$.

By Theorem~8 of \cite{chen2016generalized}, $\hat{\sf p}$ will be a consistent estimator of $\p_1$
after rescaling. Specifically, we have
$$
\Delta_{1,n} = \sup_{x\in \mathcal{R}\backslash (\mathcal{A}\oplus h)}|C^\dagger_1 \cdot h\cdot \hat{\sf p}(x)- \p_1(x)|= O(h) +O_P\left(\sqrt{\frac{\log n}{nh}}\right),
$$
where $C^\dagger_1$ is a constant depending only on the kernel function.
This implies that 
$$
C^\dagger_1 \cdot h\cdot\hat{\sf p}(x_0) \geq \p_1(x_0) - \Delta_{1,n} \geq \p_1(m_0)+ c_1 \sqrt{h} -\Delta_{1,n}.
$$
Thus,
\begin{equation}
C^\dagger_1 \cdot h\cdot \hat{\p}_* > \p_1(m_0)+ c_1 \sqrt{h} -\Delta_{1,n}.
\end{equation}

Consider a set 
$$
\Gamma = \left\{x\in\mathcal{R}: \p_1(x) \leq \p_1(m_0)+ \frac{1}{3}c_1 \sqrt{h}\right\}.
$$
Whenever $\Delta_{1,n} < \frac{1}{3}c_1 \sqrt{h}$, uniformly for every $x_1 \in \Gamma$, we have
\begin{align*}
C^\dagger_1 \cdot h\cdot \hat{\sf p}(x_1) &\leq  \p_1(x_1) +\Delta_{1,n}\\
&< \p_1(x_1)+\frac{1}{3}c_1 \sqrt{h} \\
&\leq  \p_1(m_0)+\frac{2}{3}c_1 \sqrt{h}\\
&\leq \p_1(m_0) + c_1h -\Delta_{1,n}\\
&<  C^\dagger_1 \cdot h\cdot \hat{\p}_* .
\end{align*}
Thus, $\Gamma \cap \hat{A}_{\pi_0+\pi_1} = \emptyset.$
When $h\rightarrow 0$ and $\frac{nh^2}{\log n}\rightarrow \infty$, 
$$
\P\left(\Delta_{1,n} < \frac{1}{3}c_1 h\right) \rightarrow 1
$$
because $\Delta_{1,n} = O(h) +O_P\left(\sqrt{\frac{\log n}{nh}}\right)$. 
Therefore, 
when $W_n\not\subset \hat{A}_{\pi_0+\pi_1}$, 
with a probability tending to $1$, 
$\Gamma \cap \hat{A}_{\pi_0+\pi_1} = \emptyset.$

Next, we will show that it is very unlikely that the set $\hat{A}_{\pi_0+\pi_1}$ does not contain the set $\Gamma$. By Lemma~\ref{lem::inside}, 
$\hat{A}_{\pi_0+\pi_1}\subset (\mathcal{A}\cup\mathcal{R})\oplus h$ 
with a probability tending to $1$.
In the event $W_n\not\subset \hat{A}_{\pi_0+\pi_1}$, we can assume $\Gamma \cap \hat{A}_{\pi_0+\pi_1} = \emptyset$
because it happens with a probability tending to $1$.
Then 
$$
\hat{A}_{\pi_0+\pi_1}  = \hat{A}_{\pi_0+\pi_1}\backslash \Gamma \subset
\left((\mathcal{A}\cup\mathcal{R})\oplus h\right)\backslash \Gamma.
$$
Let 
$$
\Delta_{2,n} = \sup_{r>0}|\hat{\P}_n(\left((\mathcal{A}\cup\mathcal{R})\oplus r\right)\backslash \Gamma)-
\P_{\sf GPS}(\left((\mathcal{A}\cup\mathcal{R})\oplus r\right)\backslash \Gamma)|.
$$
Then
\begin{align*}
\pi_0+\pi_1 = \hat{\P}_n(\hat{A}_{\pi_0+\pi_1}) &\leq \hat{\P}_n(\left((\mathcal{A}\cup\mathcal{R})\oplus h\right)\backslash \Gamma)\\
&\leq \P_{\sf GPS}(\left((\mathcal{A}\cup\mathcal{R})\oplus h\right)\backslash \Gamma) +\Delta_{2,n} \\
&\leq \pi_0+\pi_1 + O(h^2) +\Delta_{2,n} - \P_{\sf GPS}(\Gamma).
\end{align*}
Therefore,
we need $O(h^2) +\Delta_{2,n} - \P_{\sf GPS}(\Gamma)>0$
so that there is no contradiction within the inequalities. 
Because $\Gamma$ is the collection of points on $\mathcal{R}$ around $m_0$, $\P_{\sf GPS}(\Gamma) \geq c_2 h^{-1/4}$,
where $c_2$ is some constant that is related to the lower bound on the second derivative of $p_1$
around $m_0$. 
Thus,  a necessary condition (ignoring $O(h^2)$ since it is of a smaller order) is 
$\Delta_{2,n}> c_2 h^{-1/4}$, which occurs with a probability
$$
\P(\Delta_{2,n}> c_2 h^{-1/4}) \rightarrow 0,
$$
because the VC theory implies that $\Delta_{2,n}= O_P\left(\frac{\log n}{n}\right)$
and we require $\frac{nh^2}{\log n}\rightarrow \infty$. Thus a necessary condition to $W_n\not\subset \hat{A}_{\pi_0+\pi_1}$
is $\Delta_{2,n}> c_2 h^{-1/4}$, which occurs with a probability tending to $0$.
So we conclude that 
$$\P(W_n\not\subset \hat{A}_{\pi_0+\pi_1}) \rightarrow 0,$$
which implies
$$
\P(W_n\subset \hat{A}_{\pi_0+\pi_1})\rightarrow 1.
$$
\end{proof}

\begin{proof}[Proof of Theorem~\ref{thm::conv_set}]
{\bf Part I: recovering $\mathcal{A}$.}
Because of Lemma~\ref{lem::inside}, we assume that $\hat{A}_{\pi_0}\subset (\mathcal{A}\oplus h)$.
We will first prove that $\P\left(\mathcal{A}\subset \hat{A}_{\pi_0}\right)\rightarrow 1$,
and then show that $\P_{\sf GPS}(\hat{A}_{\pi_0}\backslash \mathcal{A})\overset{P}{\rightarrow} 0$.

Let $E_1 = \{\mathcal{A}\not\subset \hat{A}_{\pi_0}\} = \{\exists s_0\in \mathcal{A}: s_0\notin \hat{A}_{\pi_0}\}$. 
We now assume $E_1$ is true and study the probability $\P(E_1)$. Let $s_0\in\mathcal{A}$ and $s_0\notin \hat{A}_{\pi_0}$. 
Using Lemma~\ref{lem::inside}, $\hat{A}_{\pi_0}\subset (\mathcal{A}\oplus h)$ so
$$
\P_{\sf GPS}(\hat{A}_{\pi_0}) \leq \P_{\sf GPS}(\mathcal{A}\oplus h).
$$
Because $s_0\notin \hat{A}_{\pi_0}$, we can rewrite the above inequality as
\begin{equation}
\P_{\sf GPS}(\hat{A}_{\pi_0}) \leq \P_{\sf GPS}((\mathcal{A}\oplus h)\backslash s_0) \leq \pi_0 + O(h) - \p_0(s_0).
\label{eq::S0::1}
\end{equation}
Using the property that $\hat{\P}_n(\hat{A}_{\pi_0}) = \pi_0 + O(n^{-1})$, we obtain
\begin{equation}
\pi_0 + O(n^{-1}) = \hat{\P}_n(\hat{A}_{\pi_0}) \leq \hat{\P}_n((\mathcal{A}\oplus h)\backslash s_0)
\leq \P_{\sf GPS}((\mathcal{A}\oplus h)\backslash s_0) + \Delta'_{n},
\label{eq::S0::2}
\end{equation}
where $\Delta'_n  = \sup_{r>0}|\hat{\P}_n((\mathcal{A}\oplus r)\backslash s_0)-\P_{\sf GPS}((\mathcal{A}\oplus r)\backslash s_0)|$. 

When $E_1$ is true, both equations \eqref{eq::S0::1} and \eqref{eq::S0::2} must hold, which requires
$$
\pi_0 + O(n^{-1}) - \Delta'_{n} \leq \P_{\sf GPS}((\mathcal{A}\oplus h)\backslash s_0)\leq \pi_0 + O(h) - \p_0(s_0). 
$$
Namely, we need 
$$
\Delta'_n \geq \p_0(s_0) -O(h) + O(n^{-1}). 
$$
Again, using the fact that the set $\{(\mathcal{A}\oplus r)\backslash s_0: r>0\}$
has a VC dimension $1$, the VC theory implies
%However, $\Delta_n$ is the difference between an empirical measure and a probability measure,
%which entails an exponential concentration (e.g., Theorem 2.41 of \citealt{wasserman2006all}):
$$
\P(\Delta'_n>\epsilon) \leq c_0 n e^{-8 \cdot n\epsilon^2}. 
$$
Thus, 
$$
\P(E_1) \leq \P(\Delta'_n \geq \p_0(s_0) -O(h) + O(n^{-1})) \leq O(c_0 n e^{-8 \cdot n\p_0(s_0)^2}) \rightarrow 0.
$$
We conclude that
\begin{equation}
\P(\mathcal{A}\subset \hat{A}_{\pi_0})\rightarrow 1.
\label{eq::S0::p11}
\end{equation}

Now we prove the other result. 
Because $\hat{A}_{\pi_0}\subset (\mathcal{A}\oplus h)$ and the fact that
$$
\P_{\sf GPS}( (\mathcal{A}\oplus h)\backslash \mathcal{A}) = O(h),
$$
it follows that
$$
\P_{\sf GPS}(\hat{A}_{\pi_0}\backslash \mathcal{A}) \leq \P_{\sf GPS}( (\mathcal{A}\oplus h)\backslash \mathcal{A}) = O(h)
\rightarrow 0
$$
so $\P_{\sf GPS}(\hat{A}_{\pi_0}\backslash \mathcal{A})$ is a random variable with a bound shrinking at rate $O(h)$,
which implies
$\P_{\sf GPS}(\hat{A}_{\pi_0}\backslash \mathcal{A})\overset{P}{\rightarrow} 0 $.

Putting it altogether,
$$
\P_{\sf GPS}(\hat{A}_{\pi_0}\triangle \mathcal{A}) = \P_{\sf GPS}(\hat{A}_{\pi_0}\backslash \mathcal{A}) + \P(\mathcal{A}\backslash\hat{A}_{\pi_0} ),
$$
where the first quantity $\P_{\sf GPS}(\hat{A}_{\pi_0}\backslash \mathcal{A})\overset{P}{\rightarrow}  0$ as we have demonstrated in the above
and the second quantity 
$$
\P_{\sf GPS}(\mathcal{A}\backslash\hat{A}_{\pi_0} )
\begin{cases}
& =0,\quad \mbox{with a probability $\rightarrow 1$},\\
& \leq 1, \quad \mbox{with a probability $\rightarrow 0$}.
\end{cases}
$$
Hence $\P_{\sf GPS}(\hat{A}_{\pi_0} \triangle \mathcal{A})\overset{P}{\rightarrow}0$,
which proves the first assertion.

{\bf Part II: recovering $\mathcal{A}\cup\mathcal{R}$.} Again due to Lemma~\ref{lem::inside} we will assume 
$\hat{A}_{\pi_0+\pi_1}\subset ((\mathcal{A}\cup\mathcal{R})\oplus h)$.

We will make use of the set $W_n$ because by construction
$$
W_n\subset (\mathcal{A}\cup\mathcal{R}).
$$
Lemma~\ref{lem::Wn} states that with a probability tending to $1$,
$$
W_n \subset \hat{A}_{\pi_0+\pi_1}.
$$
Because $W_n$ is a subset of both $(\mathcal{A}\cup\mathcal{R})$ and $\hat{A}_{\pi_0+\pi_1}$,
\begin{align*}
(\mathcal{A}\cup\mathcal{R})\backslash \hat{A}_{\pi_0+\pi_1} &\subset (\mathcal{A}\cup\mathcal{R})\backslash W_n,\\
\hat{A}_{\pi_0+\pi_1}\backslash(\mathcal{A}\cup\mathcal{R}) &\subset \hat{A}_{\pi_0+\pi_1}\backslash W_n.
\end{align*}
Thus,
\begin{align*}
\P_{\sf GPS}((\mathcal{A}\cup\mathcal{R})\backslash \hat{A}_{\pi_0+\pi_1})&\leq
\P_{\sf GPS}((\mathcal{A}\cup\mathcal{R})\backslash W_n)\\
& \leq \P_{\sf GPS}((\mathcal{A}\cup\mathcal{R})\cap (m_0\oplus \sqrt{h}))\\
& = O(\sqrt{h}).
\end{align*}
Moreover, using the fact that $\hat{A}_{\pi_0+\pi_1}\subset ((\mathcal{A}\cup\mathcal{R})\oplus h)$,
\begin{align*}
\P_{\sf GPS}(\hat{A}_{\pi_0+\pi_1}\backslash(\mathcal{A}\cup\mathcal{R}))&\leq
\P_{\sf GPS}(\hat{A}_{\pi_0+\pi_1}\backslash W_n)\\
&\leq \P_{\sf GPS}(((\mathcal{A}\cup\mathcal{R})\oplus h)\backslash W_n)\\
& = O(h^2) + O(\sqrt{h})\\
& = O(\sqrt{h}).
\end{align*}
The above inequalities show that the two probabilities are bounded random variables
with a bound shrinking at rate $O(\sqrt{h})$,
so $\P_{\sf GPS}(\hat{A}_{\pi_0+\pi_1}\backslash(\mathcal{A}\cup\mathcal{R})) = O_P(\sqrt{h})$
and $\P_{\sf GPS}((\mathcal{A}\cup\mathcal{R})\backslash \hat{A}_{\pi_0+\pi_1}) = O_P(\sqrt{h})$.
When $h\rightarrow 0$, we conclude
$$
\P_{\sf GPS}(\hat{A}_{\pi_0+\pi_1}\triangle(\mathcal{A}\cup\mathcal{R}) ) \overset{P}{\rightarrow}0,
$$
which completes the proof.

\end{proof}

\begin{proof}[Proof of Theorem~\ref{thm::kde}]
%Recall that the KDE
%$$
%\hat{\sf p}(x) = \frac{1}{nh^2}\sum_{i=1}^n K\left(\frac{d(x,X_i)}{h}\right).
%$$
Due to assumption (K1),
the kernel function is non-increasing between $x\in[0,1]$
so
$$
K(x)\geq \begin{cases}
K\left(\frac{1}{2}\right),\quad \mbox{if $x\leq1/2$},\\
0,\quad \mbox{if $x>1/2$}.
\end{cases}
$$
That is, $K(x)\geq K\left(\frac{1}{2}\right) I\left(x\leq\frac{1}{2} \right)$. Using this inequality, the KDE
\begin{align*}
\hat{\sf p}(x) &= \frac{1}{nh^2}\sum_{i=1}^n K\left(\frac{d(x,X_i)}{h}\right),\\
&\geq  \frac{1}{nh^2}\sum_{i=1}^n K\left(\frac{1}{2}\right) I\left(\frac{d(x,X_i)}{h}\leq\frac{1}{2} \right),\\
&= \frac{K\left(\frac{1}{2}\right)}{4} \frac{1}{n(h/2)^2} I\left(d(x,X_i)\leq\frac{h}{2} \right).
\end{align*}
Thus,
\begin{align*}
\E(\hat{\sf p}(x)) &\geq \frac{K\left(\frac{1}{2}\right)}{4}\frac{1}{(h/2)^2} P\left(d(x,X_i)\leq\frac{h}{2} \right)\\
& = \frac{\pi\cdot K\left(\frac{1}{2}\right)}{4} \frac{1}{\pi (h/2)^2} \P_{\sf GPS}\left(B\left(x,\frac{h}{2}\right)\right).
\end{align*}

%By Theorem~\ref{thm::equiv}, 
Because $x\in\mathcal{A}\cup \mathcal{R}$, 
$\frac{1}{\pi (h/2)^2} \P_{\sf GPS}\left(B\left(x,\frac{h}{2}\right)\right) \rightarrow \infty$ when $h\rightarrow0$.
Therefore,
 $$
 \E(\hat{\sf p}(x))\rightarrow \infty
 $$
when $h\rightarrow0$, which completes the proof. 

%Thus, we have shown that the expectation diverges. 
%To show that the estimator $\hat{\sf p}(x)$ diverges with a probability tending to $1$,
%we only need to show that the variance of $\hat{\sf p}(x)$ is bounded.
%A direct calculation shows that 
%\begin{align*}
%{\sf Var} (\hat{\sf p}(x)) & = \frac{1}{nh^4} {\sf Var}\left(K\left(\frac{d(x,X_i)}{h}\right)\right).
%\end{align*}

\end{proof}

\end{document}